\documentclass[11pt,a4paper,final]{article}

\usepackage{amsmath}
\usepackage{amsfonts}
\usepackage{amssymb}
\usepackage{amsthm}
\usepackage{mathrsfs}
\usepackage{graphicx}
\usepackage{color,xcolor}
\usepackage[colorlinks,citecolor=blue,urlcolor=blue]{hyperref}
\usepackage[authoryear, comma]{natbib}
\usepackage[margin=0.8in]{geometry}
\usepackage[titletoc]{appendix}
\usepackage{bm}
\usepackage{algorithm}
\usepackage{algpseudocode}
\usepackage{algpascal}
\usepackage{url}
\usepackage{afterpage}
\usepackage{enumitem} 
\usepackage{multirow}
\usepackage{booktabs}
\usepackage{subfigure}

\newlength{\vslength}
\numberwithin{equation}{section}
\usepackage{makecell} 
\usepackage{caption} 
\usepackage{kotex}

\usepackage{mathtools}
\usepackage{etoc}
\usepackage{minitoc} 
\usepackage[most]{tcolorbox}

\usepackage{mathrsfs}
\usepackage{color}
\usepackage{bm}
\usepackage{algorithm}
\usepackage{algpseudocode}
\usepackage{algpascal}
\usepackage{url}
\usepackage{afterpage}
\usepackage{enumitem} 
\usepackage{multirow}
\usepackage{booktabs}
\usepackage{subfigure}

\usepackage{mathtools}
\usepackage{dsfont}
\usepackage{alphabeta}
\usepackage{bbm}

\usepackage{xpatch}
\xpatchcmd{\citet}{;}{,}{}{}

\usepackage{tikz}

\usepackage{bibentry}

\newtheorem{theorem}{Theorem}[section]
\newtheorem{lemma}[theorem]{Lemma}

\newtheorem{corollary}[theorem]{Corollary}

\theoremstyle{remark}

\def\cC{\mathcal C}

\def\cE{\mathcal E}
\def\cF{\mathcal F}

\def\cH{\mathcal H}

\def\cN{\mathcal N}

\def\cX{\mathcal X}

\newcommand{\bb}{{\bf b}}
\newcommand{\bB}{{\bf B}}

\newcommand{\bfe}{{\bf e}}
\newcommand{\bff}{{\bf f}}

\newcommand{\bH}{{\bf H}}
\newcommand{\bI}{{\bf I}}

\newcommand{\bM}{{\bf M}}

\newcommand{\bR}{{\bf R}}

\newcommand{\bS}{{\bf S}}

\newcommand{\bU}{{\bf U}}

\newcommand{\bx}{{\bf x}}
\newcommand{\bX}{{\bf X}}
\newcommand{\by}{{\bf y}}
\newcommand{\bY}{{\bf Y}}
\newcommand{\bz}{{\bf z}}
\newcommand{\bZ}{{\bf Z}}

\newcommand{\bbE}{{\mathbb E}}

\newcommand{\bbR}{{\mathbb R}}
\newcommand{\bbZ}{{\mathbb Z}}
\newcommand{\bbN}{{\mathbb N}}

\newcommand{\bbP}{{\mathbb P}}

\newcommand{\Cov}{\mbox{{\rm Cov}}}

\newcommand{\balpha}{\bm{\alpha}}

\newcommand{\bSigma}{\bm{\Sigma}}

\newcommand{\bmu}{\bm{\mu}}

\newcommand{\cf}{{\it cf. }}

\newcommand{\Holder}{H\"{o}lder}

\newcommand{\bc}{\begin{center}}
\newcommand{\ec}{\end{center}}
\newcommand{\be}{\begin{equation}}
\newcommand{\ee}{\end{equation}}
\newcommand{\ba}{\begin{array}}
\newcommand{\ea}{\end{array}}
\newcommand{\bean}{\setlength\arraycolsep{1pt}\begin{eqnarray*}}
\newcommand{\eean}{\end{eqnarray*}}
\newcommand{\bea}{\setlength\arraycolsep{1pt}\begin{eqnarray}}
\newcommand{\eea}{\end{eqnarray}}
\newcommand{\ben}{\begin{enumerate}}
\newcommand{\een}{\end{enumerate}}
\newcommand{\bed}{\begin{itemize}}
\newcommand{\eed}{\end{itemize}}

\DeclareMathOperator*{\argmin}{argmin}

\def\d{{\rm d}}
\def\D{{\rm D}}

\def\defeq{ \stackrel{\rm def}{=} }

\newcommand{\vertiii}[1]{{\left\vert\kern-0.25ex\left\vert\kern-0.25ex\left\vert #1 
    \right\vert\kern-0.25ex\right\vert\kern-0.25ex\right\vert}}

\newcommand{\tnorm}[1]{\vert\mkern-3mu\vert\mkern-3mu\vert #1 \vert\mkern-3mu\vert\mkern-3mu\vert}
\newcommand{\Tnorm}[1]{\left\vert\mkern-3mu\left\vert\mkern-3mu\left\vert #1 \right\vert\mkern-3mu\right\vert\mkern-3mu\right\vert}

\begin{document}




\thispagestyle{empty}

\title{ \vspace*{-9mm}
Nonparametric undirected graphical model selection
\\
using diffusion models }

\author{
	Hyeok Kyu Kwon$^{1}$, Myeonggu Kang$^{1}$, Minwoo Chae$^{1}$ and Wanjie Wang$^{2}$
	\\ [2mm] \vspace{-0.1 in} 
	{\small\it {}$^{1}$Department of Industrial and Management Engineering}
        \\ 
	{\small\it  Pohang University of Science and Technology (POSTECH), Pohang, 37673, South Korea }
    \\
    {\small\it {}$^{2}$Department of Statistics and Data Science, National University of Singapore, Singapore, 117546, Singapore}
}
\date{}

\maketitle 
\begin{abstract}
    Undirected graphical models provide a fundamental framework for representing conditional independence structures among high-dimensional random variables. While undirected graphical model selection has become a central problem in high-dimensional statistics, most existing methods are restricted to parametric settings. In this paper, we develop a nonparametric approach to undirected graphical model selection based on diffusion models. Recent work has shown that diffusion models can adapt to the unknown graph structure of the underlying distribution, yet utilizing these models for explicit graph estimation remains unexplored. To bridge this gap, we introduce a novel diffusion-based method for nonparametric undirected graphical model selection. We establish the model selection consistency of the proposed method and demonstrate its empirical performance through extensive simulations and two real data analyses.
    \medskip
    \\
    \noindent \textbf{Keywords:} 
    Diffusion models, Model selection consistency, Nonparametric estimation, Undirected graphical models
\end{abstract}

\addtocontents{toc}{\protect\setcounter{tocdepth}{-1}}
\section{Introduction}
\label{sec:introduction}

Let $\bX^{1}, \ldots, \bX^{n}$ be i.i.d.\ copies of $\bX_0 = (X_{0,1}, \ldots, X_{0,D})$ with common distribution $P_0$ supported on $\cX \subseteq \bbR^{D}$. 
The random vector $\bX_0$ (or equivalently, its distribution $P_0$) is said to satisfy the (pairwise) Markov property with respect to an undirected graph $G = (V,E)$ if $X_{0,i}$ and $X_{0,j}$ are conditionally independent given $\bX_{0,-(i,j)}$ for all $i \neq j \in V$ with $(i,j) \notin E$, where $V = \{1, \ldots, D\}$ denotes the vertex set and $E \subset V \times V$ the edge set.
Here, $\bX_{0,-(i,j)}$ denotes the subvector of $\bX_0$ obtained by removing $X_{0,i}$ and $X_{0,j}$. 
Learning this graph\footnote{For a given distribution $P_{0}$, there may be multiple graphs satisfying the pairwise Markov property. We focus on the graph with the smallest edge set among them; see Section~\ref{ssec:prelim_graph} for details.} structure from the observations $\bX^1, \ldots, \bX^n$ is commonly referred to as undirected graphical model selection. 
Undirected graphical models provide a fundamental framework for representing conditional independence structures among high-dimensional random variables \citep{dawid1979conditional, lauritzen1996graphical, wainwright2008graphical}, and their structure learning has become a central problem in high-dimensional statistics and machine learning \citep{meinshausen2006high}, with applications ranging from genomics \citep{yin2011sparse} and neuroscience \citep{bullmore2009complex} to finance \citep{talih2005structural} and social networks \citep{newman2016structure}. 
Accurate recovery of sparse graphical structures enables interpretable modeling and facilitates downstream inference and prediction. 

To date, most approaches to graphical model selection rely on parametric assumptions about the underlying distribution. 
A prominent example is the Gaussian graphical model, in which conditional independence is characterized by sparsity of the precision matrix. In the binary setting, the Ising model plays an analogous role. 
We refer readers to Chapter 9 of \cite{hastie2015statistical} for a comprehensive overview of parametric approaches to undirected graphical model selection.
While these methods are attractive due to their statistical tractability and computational efficiency, their theoretical guarantees rely critically on correct model specification. 
When the true distribution deviates from Gaussianity or the assumed parametric form, such methods may lead to inconsistent graph recovery; see Section~\ref{sec:supp_simul} in the Appendix for empirical evidence.

While nonparametric graph recovery is widely recognized as an important statistical problem, it has long been regarded as statistically challenging. This difficulty arises because it is intrinsically linked to estimating the high-dimensional density $p_0$ of $\bX_0$, a problem that is known to suffer from the curse of dimensionality. 
For instance, classical minimax rates for $\beta$-smooth densities scale as $n^{-\beta/(2\beta + D)}$ \citep{tsybakov2008introduction}, where $n$ denotes the sample size, and deteriorate rapidly as the ambient dimension $D$ increases. 
This reflects the fundamental difficulty of nonparametric density estimation and related inference problems. 
Moreover, existing nonparametric approaches to undirected graphical model selection are closely tied to the estimation of the Hessian of $\log p_0$, which is statistically more demanding than estimating $p_0$ itself due to the additional complexity of second-order derivative estimation; see Section~\ref{ssec:supp_related} for related works.

Recent advances in diffusion models and score-based generative modeling offer a new perspective on high-dimensional density estimation. 
Rather than directly estimating the density $p_0$, diffusion models estimate a family of score functions, that is, the gradients of the log-density; see Section~\ref{ssec:prelim_diffusion} for a brief overview.
Beyond their empirical success in generative modeling (e.g., images, videos, and language), recent theoretical developments have analyzed diffusion models from the perspective of nonparametric density estimation \citep{oko2023diffusion, kwon2026nonparametric, fan2025optimal, tang2024adaptivity, azangulov2024convergence, stephanovitch2025generalization, chakraborty2026generalization}. 
These works show that diffusion-based estimators can adapt to various low-dimensional structures of the underlying distribution $P_0$, thereby outperforming classical nonparametric methods in structured settings. 
In particular, it has been established that when $p_0$ admits a factorization into low-dimensional components, diffusion models can avoid the full curse of dimensionality, achieving convergence rates governed by the intrinsic dimensional complexity of $p_0$ rather than the ambient dimension $D$~\citep{kwon2026nonparametric, fan2025optimal}.
Notably, the Hammersley--Clifford theorem guarantees such a factorization whenever $p_0$ is Markov with respect to a sufficiently sparse undirected graph $G$.

While these works~\citep{kwon2026nonparametric, fan2025optimal} demonstrate that diffusion models adapt to unknown undirected graph structures, their theoretical results are developed in the context of nonparametric density estimation and do not directly guarantee strong performance in graphical model selection.
In particular, to the best of our knowledge, diffusion models have not been investigated for the purpose of recovering the underlying graph structure.
Nevertheless, since diffusion models adapt to undirected graph structures at the level of density estimation, it is natural to expect that diffusion models, if properly leveraged, can outperform existing approaches to graph recovery.

Motivated by this, we propose a novel diffusion-based method for nonparametric undirected graphical model selection. 
The proposed procedure first estimates a family of score functions, following the standard diffusion modeling framework, and then recovers the graph structure using samples generated from the corresponding score estimators. 
A key ingredient of our approach is Tweedie's formula (see \eqref{eq:tweedie}), which enables graph recovery without directly estimating second-order derivatives of $\log p_0$. 
Instead, the procedure relies on estimating certain covariance matrices from the generated samples, which can be done via simple Monte Carlo methods.

We establish rigorous theoretical guarantees and provide strong empirical evidence for the proposed method.
In particular, under mild regularity conditions, we establish selection consistency of the proposed procedure in a fixed-$D$ regime. 
Although extending the theory to settings where $D$ diverges would be desirable for high-dimensional applications, such results typically require substantially more delicate analysis. 
Since the existing theoretical understanding of diffusion-based density estimation is also largely restricted to the fixed-$D$ regime, we leave the study of diverging-$D$ asymptotics for future work. 
Instead, we present extensive experimental results demonstrating that our method substantially outperforms existing nonparametric approaches to graphical model selection.

The remainder of this paper is organized as follows. 
Section~\ref{sec:preliminary} reviews background material on undirected graphical models and diffusion models. 
Section~\ref{sec:method} presents the proposed method, and Section~\ref{sec:consistency} establishes the corresponding theoretical results. 
Section~\ref{sec:tuning} discusses some practical strategies for tuning parameter selection.
Simulation results and real data analyses are presented in Sections~\ref{sec:simulation} and~\ref{sec:realdata}, respectively. Section~\ref{sec:conclusion} concludes the paper. 
All technical proofs and comprehensive information, such as implementation specifications, are deferred to the Appendix.

\subsection*{Notations and definitions}

For a positive integer $n$, let $[n] = \{1, \ldots, n\}$. Vectors and matrices are denoted by boldface letters. For $\bx \in \bbR^D$ and $\bM \in \bbR^{D_1 \times D_2}$, the corresponding non-bold symbols with subscripts, such as $x_i$ and $M_{ij}$, denote their respective entries.
For a vector $\bx$, we denote its $\ell^p$-norm, $1 \le p \le \infty$, by $\|\bx\|_p$. 
Let $\phi_{\sigma}$ denote the density of the multivariate normal distribution $\cN(\bm{0}_D, \sigma^2 \bI_D)$, where $\bm{0}_D$ and $\bI_D$ denote the $D$-dimensional zero vector and the $D \times D$ identity matrix, respectively. For $a,b \in \bbR$, let $a \vee b = \max(a,b)$ and $a \wedge b = \min (a,b)$. The notation $C = C(A_1, \ldots, A_n)$ indicates that the constant $C$ depends only on $A_1, \ldots, A_n$. 
The notation $a \lesssim b$ means that $a \leq C b$, where $C > 0$ is a constant not relevant to the main argument.
Similarly, $a \asymp b$ means that $a \lesssim b$ and $b \lesssim a$.

\section{Preliminaries} \label{sec:preliminary}


\subsection{Undirected graphical models}
\label{ssec:prelim_graph}

Undirected graphical models, also called Markov random fields, provide a convenient framework for representing conditional independence relationships.
For comprehensive reviews, we refer to the literature~\citep{lauritzen1996graphical, koller2009probabilistic, drton2017structure}.

In general, the pairwise Markov property does not uniquely determine an undirected graph from $P_0$.
Specifically, if a graph $G$ satisfies the pairwise Markov property with respect to $\bX_0$, then any graph obtained by adding edges to $G$ also satisfies the same property.
We therefore focus on the graph $G_0 = ([D], E_0)$ defined by the condition that, for all $i \neq j \in [D]$, $(i,j) \notin E_{0}$ if and only if $X_{0,i}$ and $X_{0,j}$ are conditionally independent given $\bX_{0,-(i,j)}$.
We call $G_{0}$ the \textit{conditional independence graph} of $P_{0}$ \citep{drton2017structure}.
By construction, $G_0$ is the unique graph with the smallest edge set among all undirected graphs satisfying the pairwise Markov property.
Throughout the paper, our goal is to estimate $G_0$.

If $P_0$ admits a strictly positive density $p_{0}$ with respect to a product measure on $\cX$, then, by the celebrated Hammersley--Clifford theorem \citep{hammersley1971markov, lauritzen1996graphical}, the density factorizes as
\be
p_{0}(\bx) = \prod_{C \in \cC_{0}} g_{C}(\bx_{C}), \qquad \bx \in \cX,
\label{eq:factorization}
\ee
for some functions $g_{C}$, where $\bx_{C} = (x_{c})_{c \in C}$ and $\cC_{0}$ denotes the set of all (maximal) cliques in the graph $G_{0}$.
Here, a clique is a fully connected subset of the vertex set.

This factorization often yields convenient characterizations of the graph structure in certain distribution families.
Indeed, if $p_{0}$ is strictly positive and twice continuously differentiable on a suitably regular domain $\cX \subseteq \bbR^{D}$, then
\be
(i,j) \notin E_{0}
\quad\Longleftrightarrow\quad
\frac{\partial^2 \log p_{0}(\bx) }{\partial x_i \partial x_j}=0
\quad \forall \bx \in \cX,
\qquad i \neq j \in [D].
\label{eq:iff_cond}
\ee 
This equivalence holds, for example, when $\cX$ is $[-1,1]^D$, $(-1,1)^D$, or $\bbR^D$; see Lemma~2 and its proof in \cite{spantini2018inference}.
Motivated by this characterization, several recent works \citep{zheng2023generalized, baptista2024learning, liaw2025learning} focus on estimating the Hessian of $\log p_{0}$ to recover the graph $G_{0}$.
A detailed review of these works, along with classical parametric methods and their semiparametric extensions, is provided in Section~\ref{ssec:supp_related}.

\subsection{Diffusion models}
\label{ssec:prelim_diffusion}


Let $( \bX_{t} )_{t \geq 0}$ be the standard Ornstein--Uhlenbeck (OU) process defined by the stochastic differential equation (SDE)
\be
\d {\bX}_t = -  {\bX}_t \d t + \sqrt{2} \d \bB_t,\quad \bX_0 \sim P_0,
\label{eq:OU}
\ee
where $( \bB_{t} )_{t \geq 0} $ is a standard $D$-dimensional Brownian motion.
Although we focus on this OU process for simplicity, our main result (Theorem~\ref{thm2}) extends to more general time-inhomogeneous diffusion processes, including the widely used DDPM \citep{ho2020denoising} and EDM \citep{karras2022elucidating} frameworks.
For the OU process, the conditional distribution of $\bX_t$ given $\bX_0 = \bx_0$ is Gaussian with mean vector $\mu_t \bx_0$ and covariance matrix $ \sigma_t^2 \bI_{D}$, where $\mu_t = e^{-t}$ and $\sigma_t^2 = 1-\mu_t^2$.
Hence, the marginal distribution of $\bX_{t}$, denoted by $P_{t}$, has Lebesgue density $p_{t}$ given by
$
p_{t} (\bx) = \int \phi_{\sigma_{t}} (\bx - \mu_{t} \by ) \d P_{0} (\by).
$
Let  $\bff_0 (\bx, t) = \nabla \log p_t(\bx)$.
For each fixed $t$, the map $\bx \mapsto \bff_0(\bx, t)$ is the score function corresponding to the marginal density $p_t$.
By convention, we also refer to the map $(\bx,t) \mapsto \bff_0(\bx,t)$ as the score function.

For a non-random $\overline{T} > 0$, let $(\bY_{t})_{t \in [0, \overline{T}) }$ be the reverse-time process defined by $\bY_{t} = \bX_{\overline{T}-t}$.
It is well known \citep{anderson1982reverse} that, under mild assumption on $P_{0}$, the reverse-time process is also a diffusion process satisfying
\bean
\d \bY_{t} = \Big[ \bY_t + 2  \bff_0 \big (\bY_t, \overline T - t \big) \Big] \d t + \sqrt{2} \d \bB_t,\quad \bY_0 \sim P_{\overline{T}}.
\eean
Note that the Brownian motion in this SDE is not the same as that appearing in the SDE \eqref{eq:OU}.
However, for notational simplicity, we use the same notation $\bB_t$ throughout the paper.

Diffusion models exploit this reverse process to construct an implicit estimator of $P_0$.
Specifically, once an estimator $\widehat{\bff}$ of the score function $\bff_{0}$ is available, one can simulate the reverse process from the standard Gaussian distribution to generate samples from the estimated distribution.
This initialization is justified because $P_t$ converges rapidly to the standard Gaussian distribution as $t \to \infty$ \citep{bakry2014analysis}.

The score function can be estimated via score matching \citep{hyvarinen2005estimation, vincent2011connection, song2021scorebased}.
Let $\cF$ be a class of Borel-measurable functions $(\bx,t) \mapsto \bff(\bx,t)$ used to model the score function $\bff_0$.
In practice, $\cF$ is typically chosen as a class of (deep) neural networks.
To avoid potential singularity issues, we estimate $\bff_0(\cdot,t)$ only for $t \in [\underline{T},\overline{T}]$, where $\underline{T} > 0$ is sufficiently small.
By a well-known identity due to \cite{vincent2011connection}, the objective function
$
\bff \mapsto \int_{\underline{T}}^{\overline{T}} \bbE [ \| \bff( \bX_{t} , t ) - \bff_{0} (\bX_{t},t) \|_2^2 ] \d t
$
has the same minimizer over $\cF$ as
$
\bff \mapsto \bbE [  \ell_{\bff}( \bX_{0}) ],
$
where
$
\ell_{\bff}(\bx) = \bbE [ \| \bff( \mu_{T} \bx + \sigma_{T} \bZ , T ) + \sigma_{T}^{-1} \bZ \|_2^2  ]
$
denotes the loss function.
Here, $\bZ$ is a $D$-dimensional standard Gaussian random vector, $T$ is uniformly distributed on $[\underline{T},\overline{T}]$, and $\bZ$ and $T$ are independent.

This equivalence naturally leads to the following empirical risk minimization (ERM) estimator based on $n$ observations $\bX^1,\ldots,\bX^n$:
\be
\widehat \bff \in \argmin_{\bff \in \cF}
\frac{1}{n} \sum_{i=1}^n \ell_{\bff}(\bX^{i}).
\label{eq:erm}
\ee
Note that the loss function $\ell_{\bff}(\cdot)$ is not directly tractable because it involves expectation with respect to $T$ and $\bZ$.
In practice, the solution to \eqref{eq:erm} is approximated by stochastic gradient descent.
Specifically, since the objective $\bbE[\ell_{\bff}(\bX_{0})]$ can be written as an expectation over the independent random variables $\bX_0$, $T$, and $\bZ$, one forms stochastic approximations by sampling $\bX_0$ from the empirical distribution and drawing $T$ and $\bZ$ independently from their respective distributions.

Let $(\widehat \bY_t)_{t \in [0,\overline{T}-\underline{T}]}$ be the solution to the SDE
\be
\d \widehat \bY_t
=
\Big[\widehat \bY_t + 2 \widehat \bff \big(\widehat \bY_t,\overline{T}-t \big)\Big]\d t
+ \sqrt{2}\,\d \bB_t,
\quad
\widehat \bY_0 \sim \cN(\mathbf{0}_D,\bI_D),
\label{eq:diffusion}
\ee
and set $\widehat \bX_t = \widehat \bY_{\overline{T}-t}$ for $t \in [\underline{T},\overline{T}]$.
Then, for each $t \in [\underline{T},\overline{T}]$, the distribution of $\widehat \bX_t$ estimates $P_t$.
Moreover, since $P_{\underline{T}}$ is close to $P_0$ when $\underline{T}$ is small, the distribution of $\widehat \bX_{\underline{T}}$ serves as an implicit estimator of $P_0$.

In particular, one can obtain samples from the (marginal) distribution of $\widehat \bX_{\underline{T}}$ by solving the SDE~\eqref{eq:diffusion} up to $t = \overline{T} - \underline{T}$, for example via the Euler--Maruyama discretization \citep{kloeden2011numerical, song2021scorebased}.
More importantly, for any $(\bx,t) \in \bbR^{D} \times (\underline{T}, \overline{T}]$, sampling from the conditional distribution of $\widehat \bX_{\underline{T}}$ given $\widehat \bX_{t} = \bx$ is straightforward by the Markov property of the SDE~\eqref{eq:diffusion}. This conditional sampling is the key property exploited in the following section to define the graph estimator.


\section{Diffusion-based graphical model selection} \label{sec:method}

In this section, we introduce a diffusion-based method for nonparametric undirected graphical model selection.
Our method estimates the Hessian matrix $\nabla^2 \log p_t(\bx)$ from samples of the conditional distribution of $\widehat \bX_{\underline{T}}$ given $\widehat \bX_{t} = \bx$.
Notably, the method requires no additional training given the score function estimator $\widehat \bff$ defined in~\eqref{eq:erm}.
We begin with the relationship between $\nabla^2 \log p_t(\bx)$ and the conditional independence graph $G_{0}$.

Throughout the paper, we assume that $p_{0}$ is strictly positive and twice continuously differentiable on $\cX = [-1,1]^{D}$.
For each $i,j \in [D]$, let $H_{ij}(\bx,t) = \partial^2 \log p_t (\bx) / (\partial x_i \partial x_j)$.

The equivalence \eqref{eq:iff_cond} implies that, in principle, the graph $G_0$ can be recovered by analyzing $H_{ij}(\bx,t)$ at $t = 0$.
However, this characterization does not directly extend to $t>0$, since conditional independence is generally not preserved under Gaussian perturbation.
Nevertheless, for each fixed $\bx$, the map $t \mapsto H_{ij}(\bx,t)$ is continuous, and hence $H_{ij}(\bx,t)$ converges to $H_{ij}(\bx,0)$ as $t \to 0$.
One may therefore expect that the condition $(i,j) \notin E_{0}$ is nearly equivalent to $\vert H_{ij}(\bx,t) \vert$ being sufficiently small for all $\bx$.
At the same time, $t$ should not be taken too large; as $t \rightarrow \infty$, $\bX_{t}$ rapidly approaches the standard Gaussian, and hence $\vert H_{ij}(\bx,t) \vert$ converges to zero for every $(i,j)$ and $\bx$.
An illustrative example is provided in Section~\ref{ssec:supp_illu}.


For $t > 0$, $H_{ij}(\bx,t)$ can be estimated from the score function estimator $\widehat \bff$, without any additional training.
Specifically, the second-order Tweedie's formula states that
\be
 \nabla^2 \log p_t (\bx)
= \sigma_{t}^{-4} \mu_{t}^2 \text{Cov} \big[ \bX_{0} \mid \bX_{t} = \bx \big]  - \sigma_{t}^{-2} \bI_{D} \qquad \forall t > 0,
\label{eq:tweedie}
\ee
see Lemma~1 of \cite{wainwright2025score}.

The advantage of this formula is that $H_{ij}(\bx,t)$ can be estimated by estimating the conditional covariance of $\bX_{0}$ given $\bX_{t} = \bx$.
Since the distribution of $\widehat \bX_{t}$ estimates that of $\bX_{t}$ for $t \in [\underline{T}, \overline{T}]$, and $\underline T$ is sufficiently small, $H_{ij} (\bx,t)$ can be estimated by
\bean
\widehat H_{ij} (\bx,t) 
= \sigma_{t}^{-4} \mu_{t}^2 \text{Cov}_{n} \big[ \widehat X_{\underline{T}, i} , \widehat X_{\underline{T} ,j} \mid \widehat \bX_{t} = \bx \big]
\qquad \forall i \neq j \in [D],
\eean
where $\text{Cov}_{n} [ \cdot, \cdot] = \text{Cov}[\cdot,\cdot \mid \bX^{1}, \ldots, \bX^{n}]$ denotes the covariance conditional on the $n$ observations. This conditional covariance can be readily estimated via Monte Carlo sampling using trajectories simulated from the reverse process (\cf \eqref{eq:diffusion}).



We have introduced the pointwise estimator $\widehat H_{ij}(\bx, t)$ of $H_{ij}(\bx, t)$ for each $(\bx, t)$.
Recall that, when $t$ is small, one expect that the condition $(i, j) \notin E_{0}$ is nearly equivalent to $|H_{ij}(\bx, t)|$ being sufficiently small for all $\bx$.
This suggests graph estimation by thresholding $\bbE_n [\vert  \widehat H_{ij}(\bU_{t}, t) \vert ]$ for a suitable random vector $\bU_{t}$, where $\bbE_{n}[\cdot] = \bbE[\cdot \mid \bX^{1}, \ldots, \bX^{n}]$ denotes the expectation conditional on the $n$ observations.

For $t \in [\underline{T}, \overline{T}]$ and $ \tau > 0$, we define the undirected graph estimator $\widehat G_{t, \tau} = ([D], \widehat E_{t, \tau})$ by
\be
(i,j) \notin \widehat E_{t,\tau}
\quad\Longleftrightarrow\quad 
\bbE_{n} \big[ \big| \widehat H_{ij}(\bU_{t}, t) \big| \big] \leq \tau,
\qquad i \neq j \in [D].
\label{eq:estimator}
\ee
Here, $\tau$ and $t$ are allowed to depend on $n$, and the distribution of $\bU_t$ will be specified later.
The overall procedure is summarized in Algorithm~\ref{alg:simple}.

\begin{algorithm}[!t]
\caption{Diffusion-based graph estimation}
\label{alg:simple}
\begin{algorithmic}[1]
\State \textbf{Input:} Data $\bX^1,\ldots,\bX^n$, parameters $t,\tau$, and model class $\cF$.
\State Solve the problem \eqref{eq:erm} and obtain the score function estimator $\widehat \bff \in \cF$.
\State Compute $\bbE_n\big[|\widehat H_{ij}(\bU_{t},t)|\big]$ for all $i \neq j$.
\State Set $(i,j) \in \widehat E_{t,\tau}$ if and only if $\bbE_n\big[|\widehat H_{ij}(\bU_{t},t)|\big] > \tau$.
\State Return $\widehat G_{t,\tau} = ([D],\widehat E_{t,\tau})$.
\end{algorithmic}
\end{algorithm}

Note that evaluating the target estimator $\widehat G_{t,\tau}$ involves numerical errors from two sources.
First, sampling from the conditional distribution of $\widehat \bX_{\underline{T}}$ given $\widehat \bX_t=\bx$ requires discretizing the underlying SDE, for example via the Euler--Maruyama method \citep{kloeden2011numerical, song2021scorebased}, which introduces discretization error.
Second, both the covariance appearing in the definition of $\widehat H_{ij}(\bx,t)$ and the outer expectation in $\bbE_{n} [ | \widehat H_{ij}(\bU_{t}, t) | ]$ are approximated from finitely many simulated samples, which introduces Monte Carlo error.
Various works study discretization errors in diffusion models \citep{oko2023diffusion, benton2024nearly, chen2023sampling}.
Moreover, the errors induced by approximating the covariance and the expectation can be controlled by standard Monte Carlo bounds.
We assume these numerical errors to be negligible in our theoretical analysis.

\section{Selection consistency}
\label{sec:consistency}

\subsection{Assumptions}
\label{ssec:assm}

For a multi-index $\balpha = (\alpha_1,\ldots,\alpha_D)^{\top} \in (\bbZ_{\geq 0})^{D}$, let $\D^{\balpha}$ denote the mixed partial derivative operator $\partial^{\balpha_{\cdot}} / ( \partial x_1^{\alpha_1} \cdots \partial x_{D}^{\alpha_D})$, where $\alpha{_{\cdot}} = \sum_{i=1}^{D} \alpha_i$.
For any $\beta, K > 0$, let $\cH^{\beta, K}(A)$ denote the class of $\beta$-H\"older functions on $A \subseteq \bbR^{D}$, consisting of all functions $g : A \to \bbR$ such that
\bean
\sum_{\alpha_{\cdot} \leq \lfloor \beta \rfloor} \sup_{\bx \in A} \vert (\D^{\balpha} g) (\bx) \vert + \sum_{\alpha_{\cdot} = \lfloor \beta \rfloor} \sup_{\substack{\bx,\by \in A \\ \bx \neq \by}} \frac{\vert (\D^{\balpha} g)(\bx) - (\D^{\balpha} g)(\by) \vert}{\| \bx-\by\|_{\infty}^{\beta - \lfloor \beta \rfloor}} \leq K,
\eean
where $\lfloor \beta \rfloor$ denotes the largest integer strictly smaller than $\beta$.
Throughout the paper, we will impose the following assumption on  $p_{0}$:
\bed
\item[] (\bS) There exist constants $\beta > 2$ and $K > 0$ such that $\log p_0 \in \cH^{\beta, K}([-1,1]^{D})$.
\eed

This condition arises naturally because our characterization of the graph structure relies on the equivalence \eqref{eq:iff_cond}, which in turn requires the second-order smoothness of $\log p_{0}$.
To the best of our knowledge, existing methods for nonparametric undirected graphical model selection rely on substantially more restrictive assumptions \citep{liu2009nonparanormal,liu2012copula,xue2012regularized,liu2011forest,baptista2024learning,zheng2023generalized,liaw2025learning}.
Detailed reviews of these works are provided in Section~\ref{ssec:supp_related}.

\subsection{Edge identification via Hessian components}
\label{ssec:hessian}

In this subsection, we present one of the main results of the paper, showing that the edges of $G_{0}$ can be identified through the Hessian of $\log p_{t}$.
Given an undirected graph $G$, 
let $d_{G}(i,j)$ denote the minimum length of a path from vertex $i$ to vertex $j$ in $G$.
By convention, we set $d_{G}(i, i) = 0$ and $d_{G}(i, j) = \infty$ if there is no path between $i$ and $j$ in $G$.
For simplicity, we write $d(i,j) = d_{G_{0}}(i,j)$.
Note that for all $i \neq j \in [D]$,
$ (i,j) \notin E_{0} $
if and only if
$ d(i,j) > 1 $.
Theorem \ref{thm1} shows that, ${\rm{Cov}} [ X_{0,i}, X_{0,j} \mid \bX_{t} = \bx ]$ decays with a rate depending on $d(i, j)$ as $t \to 0$.

\begin{theorem}
\label{thm1}
Let $\gamma \in [0,1)$, and suppose that the density $p_{0}$ satisfies assumption (\bS). 
Then, for every $t \leq C_{1}$ and $\| \bx \|_{\infty} \leq \mu_t \gamma$, we have $\sigma_{t} / \mu_{t} \leq 1$,
\bean
\Big\vert  {\rm{Cov}} \big[ X_{0,i}, X_{0,j} \mid \bX_{t} = \bx \big] \Big\vert
\leq  C_{2}  \sigma_{t}^{ ( 2 d(i,j) + 2 ) \wedge (\beta + 2) \wedge 5  },
\eean
for all $i,j \in [D]$, and
\bean
\Big\vert  {\rm{Cov}} \big[ X_{0,i}, X_{0,j} \mid \bX_{t} = \bx \big] \Big\vert
\geq
\Big\vert \big( \bH^{d(i,j)} \big)_{ij} \Big\vert \big( \sigma_{t} / \mu_{t} \big)^{2 d (i,j) + 2 }
 - C_{2} \sigma_{t}^{ ( 2 d(i,j) + 4 ) \wedge (\beta + 2) \wedge 5 },
\eean
for all $i,j \in [D]$ with $d(i,j) < \infty$, where $C_{1} = C_{1} ( \beta,K,D, \gamma) $ and $C_{2} = C_{2} (\beta,K,D) $ are positive constants, and
$
\bH = ( H_{ij} ( \bx / \mu_{t} , 0 ) ) \in \bbR^{D \times D}.
$
\end{theorem}

Since $\sigma_{t} = \sqrt{1-e^{-2t}}$ and $\mu_{t} = e^{-t}$, we have $\sigma_{t} \asymp \sqrt{t}$ and $\mu_{t} \asymp 1$ for sufficiently small $t$.
Hence, Theorem~\ref{thm1} implies that
\bean
\Big\vert  {\rm{Cov}} \big[ X_{0,i}, X_{0,j} \mid \bX_{t} = \bx \big] \Big\vert
\asymp t^{2} \qquad \forall (i,j) \in E_{0},
\eean
provided that $H_{ij}(\bx / \mu_{t}, 0)$ is nonzero, whereas
\bean
\Big\vert  {\rm{Cov}} \big[ X_{0,i}, X_{0,j} \mid \bX_{t} = \bx \big] \Big\vert
\lesssim t^{\frac{ (\beta+2) \wedge 5}{2}} \qquad \forall (i,j) \notin E_{0} \text{ and } i \neq j.
\eean
Consequently, the conditional covariance decays at a faster rate for non-edges $(i,j) \notin E_{0}$ than for edges $(i,j) \in E_{0}$.
Theorem~\ref{thm1} immediately yields the following corollary, which characterizes the decay rate of $H_{ij}(\bx,t)$ according to whether $(i,j) \in E_{0}$.

\begin{corollary}
\label{cor1}
Suppose that the density function $p_{0}$ satisfies assumption (\bS).
Let $\gamma, C_{1}, C_{2}$ be the constants defined in Theorem~\ref{thm1}.
Then, for every $t \leq C_{1}, \| \bx \|_{\infty} \leq \mu_{t} \gamma $, and $i,j \in [D]$, we have
\bean
 \big\vert H_{ij}  ( \bx / \mu_t, 0  ) \big\vert - C_{2} \sigma_{t}^{\tilde \beta}
\leq \big\vert H_{ij} (\bx, t) \big\vert 
\leq C_{2},
\quad \text{if } d(i,j) = 1,
\eean
and
$
\vert  H_{ij} (\bx, t)  \vert \leq C_{2} \sigma_{t}^{\tilde \beta}$ if $d(i,j) \geq 2$,
where $\tilde \beta = (\beta-2) \wedge 1$.
\end{corollary}

Corollary~\ref{cor1} implies that, for small $t$ and suitable $\bx \in [-\mu_{t}\gamma, \mu_{t}\gamma]^{D}$, the pointwise value $|H_{ij}(\bx, t)|$ can be used to characterize conditional independence.
Specifically, we have $|H_{ij}(\bx, t)| \lesssim t^{\tilde \beta / 2}$ for $(i, j) \notin E_{0}$ with $i \neq j$, whereas $|H_{ij}(\bx, t)| \asymp 1$ for $(i, j) \in E_{0}$.
This pointwise characterization, however, requires that $|H_{ij}(\bx / \mu_{t}, 0)| > 0$ for every $(i, j) \in E_{0}$.
One can instead characterize conditional independence through the averaged Hessian $\bbE[|H_{ij}(\bU_{t}, t)|]$, provided that $\bbE[|H_{ij}(\bU_{t} / \mu_{t}, 0)|] > 0$ for every $(i, j) \in E_{0}$.

Throughout this section, we take $\bU_{t}$ to be the uniform random vector on $[-\mu_{t}\gamma, \mu_{t}\gamma]^{D}$ for each $t \geq 0$.
Note that this choice is flexible: our main consistency result (Theorem~\ref{thm2}) continues to hold if the uniform distribution is replaced by any continuous distribution on $[-\mu_{t}\gamma, \mu_{t}\gamma]^{D}$ whose density is uniformly bounded from above for $t \lesssim 1$.
Practical choices of $\bU_{t}$ used in our experiments are discussed in Section~\ref{sec:supp_details}.

Note that $\bU_{0}$ is the uniform distribution on $[-1,1]^D$.
The equivalence \eqref{eq:iff_cond} then implies that $\bbE[|H_{ij}(\bU_{0}, 0)|] > 0$ for every $(i, j) \in E_{0}$, provided that $\gamma$ is sufficiently close to $1$.
Consequently, $C_{\gamma} > 0$ can always be ensured by such a choice of $\gamma$, where $C_{\gamma} = \min_{(i,j) \in E_{0}} \bbE[\vert H_{ij}(\bU_{0}, 0) \vert]$.
We therefore fix $\gamma \in (0, 1)$ such that $C_{\gamma} > 0$.
When $E_{0}$ is empty, we set $C_{\gamma} = 1$; any fixed positive constant suffices for our main results.


Let $\widetilde {T} = {\widetilde T} (\beta, \underline{T}, C_{1}, C_{2}, C_{\gamma}) > 0$ be a (small enough) constant such that $2 C_{2} \sigma_{{\widetilde T}}^{\tilde \beta} < C_{\gamma}$, ${\widetilde T} \leq C_{1}$ and ${\widetilde T} \geq \underline{T}$.
Then, for every $t \leq {\widetilde T}$ and $i \neq j \in [D]$, Corollary~\ref{cor1} implies that
\be
(i,j) \notin E_{0}
\quad\Longleftrightarrow\quad
\bbE \big[ \big\vert H_{ij} (\bU_{t}, t) \big\vert \big] \leq C_{2} \sigma_{t}^{\tilde \beta}.
\label{eq:iff_cond_t}
\ee
This novel characterization is the key to choosing the threshold $\tau$ in the estimator $\widehat G_{t,\tau}$ and to proving graph selection consistency in Section~\ref{ssec:consistency}.

A closely related analysis of the relationship between $H_{ij}(\bx,t)$ and $d(i,j)$ was given by \cite{gottwald2025localized}.
Under stronger assumptions than those of Corollary~\ref{cor1}, they derived a sharper upper bound on $\vert H_{ij}(\bx,t) \vert$.
Specifically, Theorem~2.1 of \cite{gottwald2025localized} states that, if $p_{0}$ is log-concave and twice continuously differentiable, then $\vert H_{ij}(\bx,t) \vert \lesssim \sigma_{t}^{2d(i,j)-2}$.
This bound matches with Corollary~\ref{cor1} when $d(i,j) = 1$ and is strictly stronger if $d(i,j) > 1$.
In contrast, our result relaxes the log-concavity assumption and also provides a lower bound on the Hessian entry, leading to the characterization of conditional independence in \eqref{eq:iff_cond_t}.

\subsection{Estimator consistency}
\label{ssec:consistency}

For two undirected graphs $G_{1} = ([D],E_{1})$ and $G_{2} = ([D],E_{2})$, we write $G_{1} = G_{2}$ if and only if $E_{1} = E_{2}$.
We say that the estimator $\widehat G_{t,\tau}$ is consistent if 
$
\bbP ( \widehat G_{t,\tau} \neq G_{0} ) = o(1).
$
With appropriate choices of $\tau$ and $t$, the following theorem provides a non-asymptotic upper bound on $\bbP(\widehat G_{t,\tau} \neq G_{0})$.

\begin{theorem}
\label{thm2}
Suppose that $p_{0}$ satisfies (\bS), and let $\tilde \beta, \gamma, C_{\gamma}, C_{2}, {\widetilde T}$ be the constants defined in Section~\ref{ssec:hessian}.
For each $t \geq 0$, let $\bU_{t}$ be the uniform random vector on $[-\mu_{t}\gamma, \mu_{t}\gamma]^{D}$.
Let $\epsilon_{n} > 0$ be given with $\epsilon_{n} < 1/e$, and suppose that
\be
\bbE \left[ \left( \int_{\underline{T}}^{{\widetilde T}} \int_{\bbR^{D}} \big\| \widehat \bff(\bx,t) - \bff_0(\bx, t)  \big\|_2^2 p_{t}(\bx) \d \bx \d t \right)^{1/2} \right]
 \leq \epsilon_{n},
 \label{thm2:assmp1}
\ee
and
\be
 \sup_{\bff \in \cF } \sup_{ \bx \in \bbR^{D}} \big\| \bff( \bx, t) \big\|_{\infty} \leq  \sigma_{t}^{-1} \sqrt{\log (1/\epsilon_n)},
 \quad \forall t \in [\underline{T}, {\widetilde T}].
  \label{thm2:assmp2}
\ee
For each $t \leq {\widetilde T}$, let $\tau > 0$ be chosen to satisfy 
$
C_{2} \sigma_{t}^{\tilde \beta} < \tau < C_{\gamma} - C_{2} \sigma_{t}^{\tilde \beta}.
$
Then, for every $t \in [\underline{T}, {\widetilde T}]$, we have 
\bean
\bbP \Big( \widehat G_{t,\tau} \neq G_{0} \Big)
\leq C_{3} \left[
\frac{\underline{T}}{ (  t^{2 -\tilde \beta / 2} \eta_{t,\tau,1} ) \wedge ( t^2 \eta_{t,\tau,2} ) }
+
\frac{ \epsilon_n  \{ \log (1/\epsilon_n) \}^{ 3/2 } }{ t^2 ( \eta_{t,\tau,1} \wedge \eta_{t,\tau,2} ) } 
\right],
\eean
where $C_{3} = C_{3} (\beta, K, D, \gamma )$, $\eta_{t,\tau,1} = \tau - C_{2} \sigma_{t}^{\tilde \beta}$ and $\eta_{t,\tau,2} = C_{\gamma} - C_{2} \sigma_{t}^{\tilde \beta} - \tau$.

\end{theorem}

The condition \eqref{thm2:assmp1} is a key assumption for establishing consistency; it specifies the convergence rate of $\widehat \bff$ to the true score function $\bff_{0}$, which in turn governs the convergence rate of the graph estimator $\widehat G_{t,\tau}$ to $G_{0}$.
Several works on diffusion models establish \eqref{thm2:assmp1} under (\bS), together with additional technical assumptions.
For example, when $\cF$ is taken to be a class of sparse neural networks, the rate $\epsilon_{n} \asymp n^{-\beta/(2\beta+D)}$ is achievable up to a logarithmic factor~\citep{oko2023diffusion}.
Moreover, the improved rate $\epsilon_{n} \asymp n^{-\beta/(2\beta+d)}$ can be attained up to a logarithm factor by taking $\cF$ to be a class of sparse weight-sharing networks \citep{kwon2026nonparametric} or fully connected networks \citep{fan2025optimal}.
Here, $d = \max_{C \in \cC_{0}} \vert C \vert$ denotes the maximum clique size of $G_{0}$.
The condition~\eqref{thm2:assmp2} is a technical assumption that controls the tail probability of the solution to the SDE~\eqref{eq:diffusion}, which is also imposed in the aforementioned works \citep{oko2023diffusion, kwon2026nonparametric, fan2025optimal}.

By Theorem \ref{thm2}, once we take $\underline{T} \asymp n^{-c_{0}}$ for a sufficiently large constant $c_{0} > 0$, the estimator $\widehat G_{t, \tau}$ is consistent whenever
\be
\frac{\epsilon_n \{\log(1/\epsilon_n)\}^{3/2}}{t^2 ( \eta_{t,\tau,1} \wedge \eta_{t,\tau,2} )}=o(1).
\label{eq:consistency}
\ee
When $\epsilon_{n} = o(1)$, condition \eqref{eq:consistency} is easily verified by taking $t$ and $\tau$ to be constants, for example, $t = {\widetilde T}$ and $\tau = C_{\gamma}/2$.
As another example, if $\beta \geq 3$ and $\epsilon_{n} \asymp n^{-c_{1}}$, then the choice $t = n^{-c_{1}/4}$ and $\tau = 2n^{-c_{1}/32}$ satisfies \eqref{eq:consistency}.


Note that the constant $C_{\gamma}$ represents the minimum signal of dependence between variables.
We have implicitly assumed that this signal is of constant order; however, it is natural to consider settings in which two variables are conditionally dependent, but the strength of their dependence is weak. Such weak dependence can be accommodated by allowing $C_{\gamma} \to 0$ as $n \rightarrow \infty$.
Building on the previous example, let $c_{1} = \beta/(2\beta+d)$, $\epsilon_{n} = n^{-c_{1}}$, $t = n^{-c_{2}}$ and $\tau = 2 C_{2} \sigma_{t}^{\widetilde \beta}$, where $c_{2} = 2 c_{1}/(4 + \widetilde \beta) - \delta/\widetilde \beta$ for a sufficiently small $\delta > 0$.
Then, condition \eqref{eq:consistency} is satisfied if $ C_{\gamma} 
\gtrsim n^{-\frac{\tilde \beta \beta}{(\tilde \beta + 4)(2\beta + d)} + \delta} $, which implies that the graph estimator $\widehat G_{t,\tau}$ remains consistent even in this weak dependence scenario.


\section{Practical tuning strategy} \label{sec:tuning}

While Theorem \ref{thm2} provides sufficient conditions for $(t, \tau)$ to guarantee the consistency of $\widehat G_{t, \tau}$, these parameters depend on unknown quantities such as $d$ and $\beta$. In particular, selecting an appropriate threshold $\tau$ poses a significant challenge in practical applications. To address this inherent limitation, this section discusses practical procedures for aggregating information across a range of timescales $t$ and employing a heuristic to determine the threshold for edge selection. A detailed discussion accompanied by an illustrative example is provided in Section~\ref{sec:supp_details}.

For each fixed $t$, let $\widetilde H_{ij}(t)$ be the standardized value of $\bbE_n[|\widehat H_{ij}(\bU_{t}, t)|]$, where the standardization is performed across all pairs $\{ (i,j) : i,j \in [D], \, i < j \}$. Instead of directly thresholding $\bbE_n[|\widehat H_{ij}(\bU_{t}, t)|]$ for a single $t$, we leverage information across multiple time points $t_1 < \cdots < t_M$, chosen to be sufficiently small. Empirically, we have observed that it suffices to select $t_i$'s such that $\sigma_{t_i} / \mu_{t_i} \leq 0.5$. We set $M=30$ in all our experiments. We then apply $K$-means clustering with $K = 2$ to the $M$-dimensional vectors
$
\{ \big(\widetilde H_{ij}(t_1), \ldots, \widetilde H_{ij}(t_M)\big) : i, j \in [D], \, i < j \}.
$
For each $i < j$, we let $(i, j) \in \widehat E$ (and consequently $(j, i) \in \widehat E$) whenever the cluster containing $(i, j)$ corresponds to the one with the larger centroid. As illustrated in Section \ref{sec:simulation}, this simple rule works surprisingly well across all our numerical experiments.




\section{Simulation studies} \label{sec:simulation}

With the tuning strategies developed in Section~\ref{sec:tuning}, we have conducted extensive simulation studies to empirically demonstrate the performance of the proposed method and compare it with existing parametric and nonparametric approaches. We consider four data distributions, comprising two Gaussian and two non-Gaussian distributions. Remarkably, across all experiments, the proposed method performs competitively with the correctly specified parametric models and outperforms all nonparametric baselines by a significant margin. The full experimental setup, implementation details, and a detailed discussion of these results are provided in Section~\ref{sec:supp_simul}; here, we focus on an illustrative case.

Figure~\ref{fig:simul_examples} presents the results for one non-Gaussian example (Gaussian copula), showing that the values $\widetilde H_{ij}(t)$ align well with the graph distance $d(i,j)$. As can be seen, two clear clusters emerge based on whether $d(i,j)=1$ or $d(i, j) > 1$, and their separation becomes increasingly distinct as the sample size increases. This behavior is consistently observed throughout all our experiments.


\begin{figure*}[!t]
    \centering
        \includegraphics[width=\linewidth]{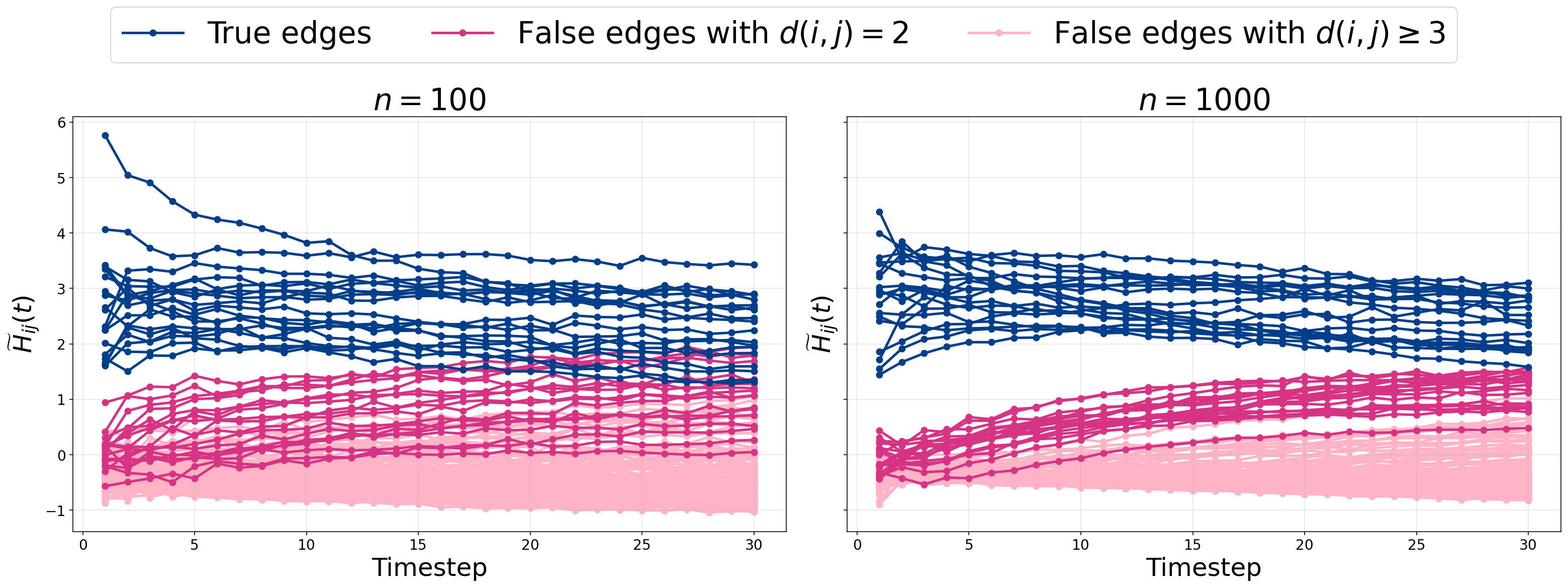}
        \caption{$\widetilde H_{ij}(t)$ for a non-Gaussian example with $n = 100$ (left) and $n = 1000$ (right).} \label{fig:simul_examples}
\end{figure*}

\section{Real data analysis} \label{sec:realdata}

In this section, we apply our diffusion-based method for undirected graphical model selection to two real datasets that have not been explored in the existing literature.
We first apply our method to image data, where the recovered graph has a transparent interpretation, and then to stock price data, where the dependence structure among companies is less obvious and of independent interest.

\subsection{Image analysis}
\label{ssec:image}

A long-standing tradition in image analysis is to represent an image as a local undirected graph, in which nearby pixels are directly connected, and distant pixels are connected only through their neighbours.
This locality has long been exploited in classical image processing \citep{li2009markov} and underlies the design of modern architectures such as convolutional neural networks \citep{krizhevsky2012imagenet, goodfellow2016deep}.
Although recent empirical work suggests that conditioning on neighbouring pixels weakens the dependence between distant pixels \citep{vandermeulen2024breaking, vandermeulen2025dimension}, the conditional independence graph of image data has not been formally investigated, to the best of our knowledge.
We apply the proposed method to the MNIST dataset \citep{lecun1989handwritten}, which consists of $50000$ training images, each of size $28 \times 28 = 784$ pixels; see Section~\ref{ssec:supp_mnist} for further details.

The left panel of Figure~\ref{figure:mnist_example} displays the resulting estimated graph.
Each vertex is either isolated or connected only to its immediate neighbours, which aligns with the locality principle described above.
Edges are concentrated near the center of the image, where most non-zero pixel values reside.
Out of the $\binom{784}{2} = 306936$ possible edges, only $640$ are recovered, indicating that the conditional independence graph of MNIST is highly sparse.

A more striking feature is the emergence of diagonal edges. In addition to horizontal and vertical neighbors, the estimated graph contains a substantial number of diagonal edges connecting upper-right to lower-left pixels, but markedly fewer connecting upper-left to lower-right pixels.

To illustrate this phenomenon, we examine the estimated local Hessian values around a fixed anchor pixel.
Since each pixel of the $28 \times 28$ grid can be identified by its position in the $784$-dimensional vector, we fix the anchor at $i = (15, 16)$ and report the values of $\widetilde H_{ij}(t)$ for all $j \in [28] \times [28]$ with $j \neq i$ in the right panel of Figure~\ref{figure:mnist_example}.
Here, the upper-left pixel corresponds to $(1, 1)$ and the lower-right pixel to $(28, 28)$.
The four pixels at $\ell^{1}$-distance one from $i$, namely $(14, 16)$, $(16, 16)$, $(15, 15)$, and $(15, 17)$, are clearly identified as strong candidates for inclusion in the edge set.
Among the pixels at $\ell^{1}$-distance two from $i$, the next strongest candidates are $(16, 15)$ and $(14, 17)$, which correspond precisely to the upper-right-to-lower-left diagonal direction.
This pattern aligns with the asymmetry of diagonal edges observed in the estimated conditional independence graph.

We conjecture that this asymmetry reflects the predominant slant of MNIST digits, which tend to lean from the upper right toward the lower left. This slant arises mainly because the digits are handwritten, and the majority of writers are right-handed, typically writing digits from the upper right toward the lower left. While the horizontal and vertical neighbors can be readily explained by the general properties of natural images, these diagonal edges reflect a property specific to MNIST images themselves. This example shows that our model captures even such an unanticipated structure, suggesting its potential for broader applications.

We have shown that the MNIST data set admits a low-dimensional, sparse graphical structure.
Although the ambient dimension is $D = 784$, Figure~\ref{figure:mnist_example} reveals that the maximum clique size in the estimated graph is only $d = 4$.
This provides empirical support for the sparse graphical model assumptions employed in \cite{kwon2026nonparametric, fan2025optimal} to explain the success of diffusion models on high-dimensional structured data.

\begin{figure*}[!t]
    \centering
    \includegraphics[width=\linewidth]{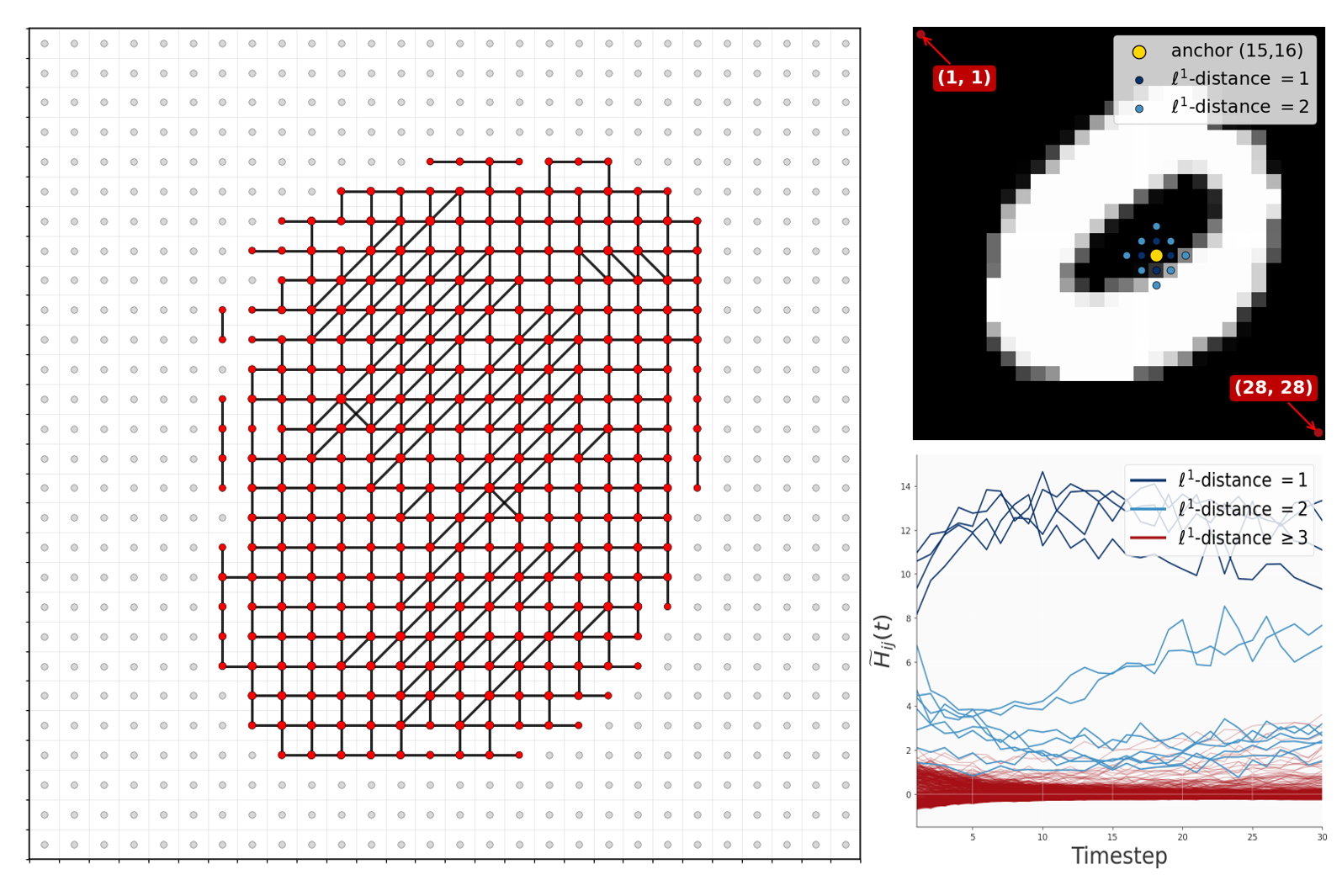}

    \caption{Estimated conditional independence graph of the MNIST data (left), and an example MNIST image with the values $\widetilde H_{ij}(t)$ for a fixed anchor pixel $i = (15, 16)$ (right). The values cluster according to the $\ell^{1}$-distance from the anchor.}
    \label{figure:mnist_example}
\end{figure*}



\subsection{Network analysis}
\label{ssec:network}

We analyze the stock prices of $D = 28$ companies and compare the recovered edges with supplementary knowledge about inter-company relationships. Through the \texttt{yfinance} package, we collect the daily closing stock prices of each company from January 1st, 2019 to December 31st, 2019, yielding 251 daily prices $\bR^{1}, \ldots, \bR^{251}$ per company. The inter-company relationships are collected from \textit{Relato}\footnote{\url{https://data.world/datasyndrome/relato-business-graph-database}} and pre-processed in the same way as \cite{xu2023covariate}. We focus on the ``competitor" and ``customer/supplier" relationships, as they are key drivers of the joint dynamics of stock prices. We restrict attention to companies in the Standard \& Poor's $500$ that belong to the industrials sector as of March $2026$ and appear in the Relato dataset. The two relationship graphs are shown in Figure~\ref{figure:network_true}.

Our goal is to estimate the conditional independence graph from the stock prices $\bR$. 
To eliminate the time dependence across stock prices, we work on the daily returns instead of the raw prices. Define $\bX^t = (\bR^{t+1} - \bR^{t})/\bR^t$ for $t \in [250]$, where $\bX^t \in \mathcal{R}^{28}$ contains the daily returns of all 28 companies on the $t$-th day. We apply our Algorithm \ref{alg:simple} to $\bX^{1}, \bX^2, \cdots, \bX^{250}$, with tuning parameters chosen as in Section \ref{sec:tuning}. 
Implementation details are provided in Section \ref{ssec:supp_network}. 

The estimated graph $\widehat{G}$ is summarized in Figure \ref{figure:network_est}, which recovers 15 edges out of $\binom{28}{2} = 378$ possible edges. It indicates the sparsity of conditional independence. 
We compare the estimated graph with the competitor and the customer/supplier relationships in Table~\ref{tab:graph_metrics_with_boeing}, in terms of the true positive rate (TPR), false discovery rate (FDR), and Jaccard distance (JD). In the relationship networks, Boeing is connected to many more companies than any other single company. It indicates the impact of Boeing on the whole industry. Hence, we consider both cases where Boeing is included and excluded from the true network in the comparison. 
The estimated graph recovered 7 out of 12 edges in the competitor graph and 4 out of 18 edges in the customer/supplier graph, excluding Boeing. 
Figure~\ref{fig:network_hessian} displays the values $\widetilde H_{ij}(t)$ for all pairs $(i,j)$, with colors indicating whether each pair is labeled as a competitor relationship, included in the estimated graph, or both.
The full list of mismatches between the competitor graph and the estimated graph is provided in Table~\ref{tab:competitor_mismatch} in Section~\ref{ssec:supp_network}.
These results validate the effectiveness of our method.  

\begin{table}[!t]
\caption{Graph estimation metrics for competitor and customer/supplier relationships in the stock price data, with and without Boeing-related relationships.}
\label{tab:graph_metrics_with_boeing}
\footnotesize
\centering
\begin{tabular}{l|c|cccc}
\toprule
Relationship &  & Total & TPR & FDR & JD \\
\midrule
Competitor &Boeing Included&  17  & 7/17 & 8/15 & {0.72} \\
 & Boeing Excluded & 12 & \textbf{7/12} & 8/15 & \textbf{0.65} \\
 \hline
Customer/Supplier &Boeing Included & 25 & 4/25 & 11/15 & 0.89 \\
 & Boeing Excluded  &  18 & 4/18 & 11/15 & 0.86 \\
\bottomrule
\end{tabular}
\end{table}

A closer inspection further suggests that $\widehat{G}$ recovers economically meaningful relationships, even though they are not indicated in the competitor graph. 
For example, the top two strongest false positive edges in $\widehat{G}$ are (\textit{Parker Hannifin}, \textit{Rockwell Automation}) and (\textit{Eaton}, \textit{Rockwell Automation}). 
Eaton and Rockwell Automation have a publicly disclosed technology-partner relationship, where Eaton's power-management technologies complement Rockwell Automation's industrial automation solutions by improving power continuity and system reliability for manufacturing customers\footnote{\url{https://www.eaton.com/us/en-us/catalog/backup-power-ups-surge-it-power-distribution/eaton-intelligent-power-manager/power-management-alliance-partners/rockwell-automation.html}}. 
Parker Hannifin and Rockwell Automation serve overlapping industrial and manufacturing customers, at different layers of the architecture. Parker provides motion/control components and subsystems while Rockwell provides the automation control and software platform. In the customer/supplier graph, we can observe a common neighbour of both Parker Hannifin and Rockwell Automation, validating this relationship. 
This prediction is economically plausible.


Conversely, the pair (\textit{ADP}, \textit{FedEx}) is labeled as a competitor in the Relato dataset, but absent and even disconnected in $\widehat G$.
This pair is unlikely to represent a meaningful competitive relationship: ADP operates in payroll, human capital management, HR outsourcing, and PEO services, whereas FedEx operates in transportation, parcel delivery, freight, and logistics, so the two firms' products are not substitutes from the customer's perspective. 
However, both firms target small-business customers, so they may bid on overlapping keywords related to small-business services, which causes a ``competitor" relationship in the Relato dataset. 
This competitor relationship does not affect the stock price.
Together, these examples suggest that stock-price-based graph estimation may provide complementary information on economically meaningful relationships beyond those recorded in the relationship dataset.


\begin{figure*}[!t]
    \centering
    \includegraphics[width=0.48\linewidth]{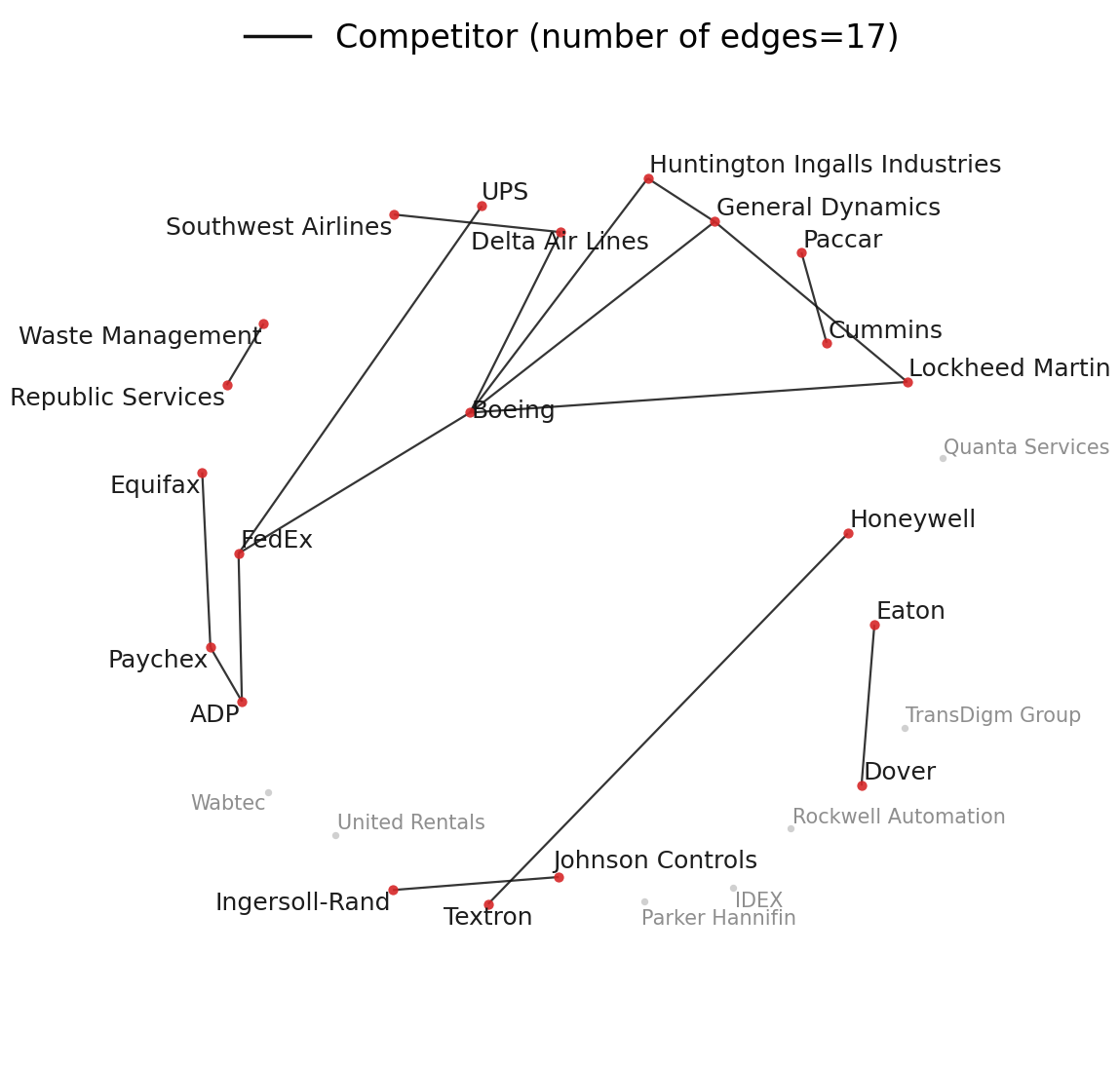}
    \includegraphics[width=0.48\linewidth]{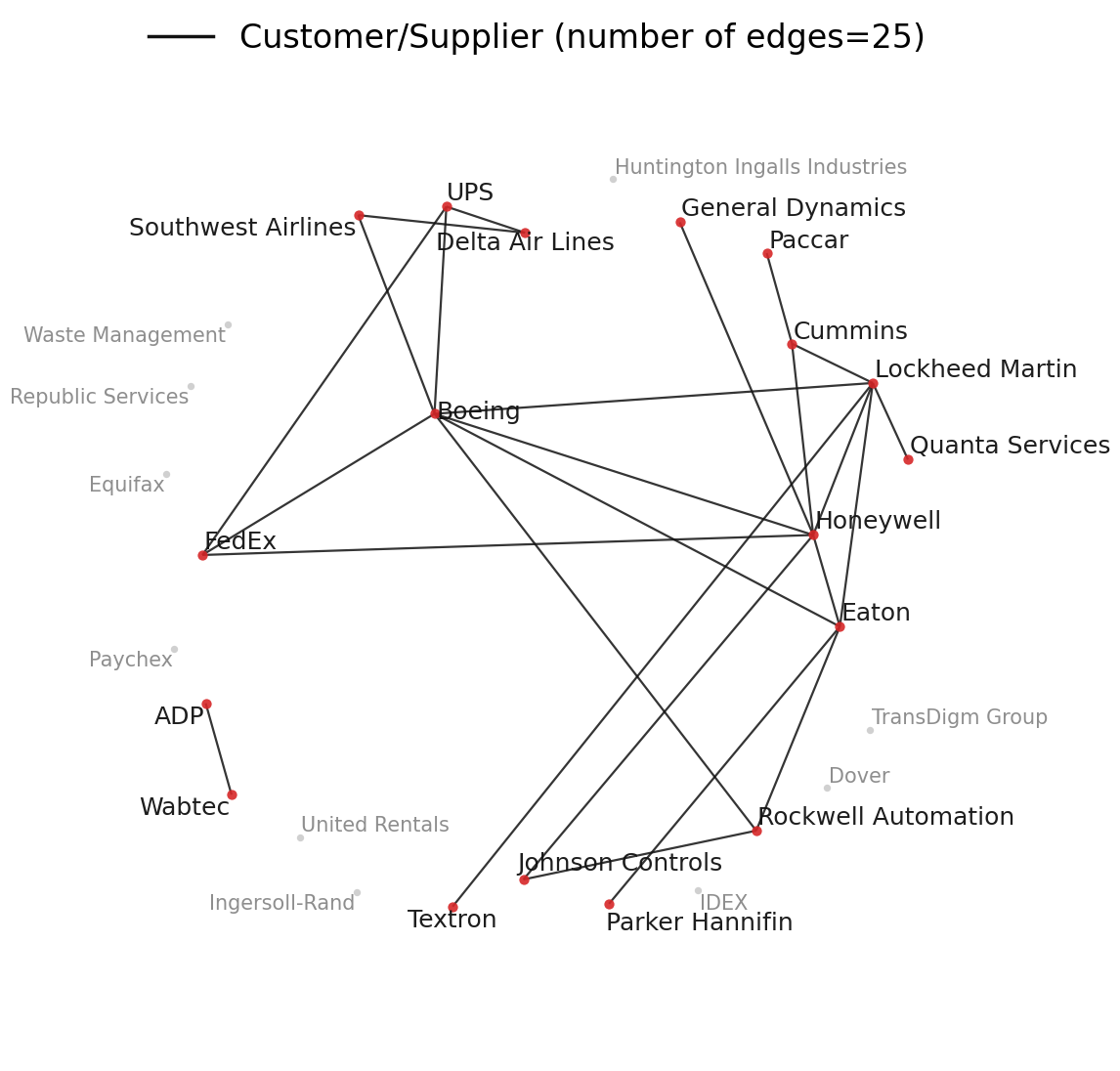}
    \caption{Relationship graphs from \textit{Relato} for the $28$ industrials companies in the Standard \& Poor's $500$: competitor relationships (left) and customer/supplier relationships (right).}
    \label{figure:network_true}
\end{figure*}

\begin{figure*}[!t]
    \centering
    \subfigure[Estimated conditional independence graph $\widehat G$]{
    \includegraphics[width=0.48\linewidth]{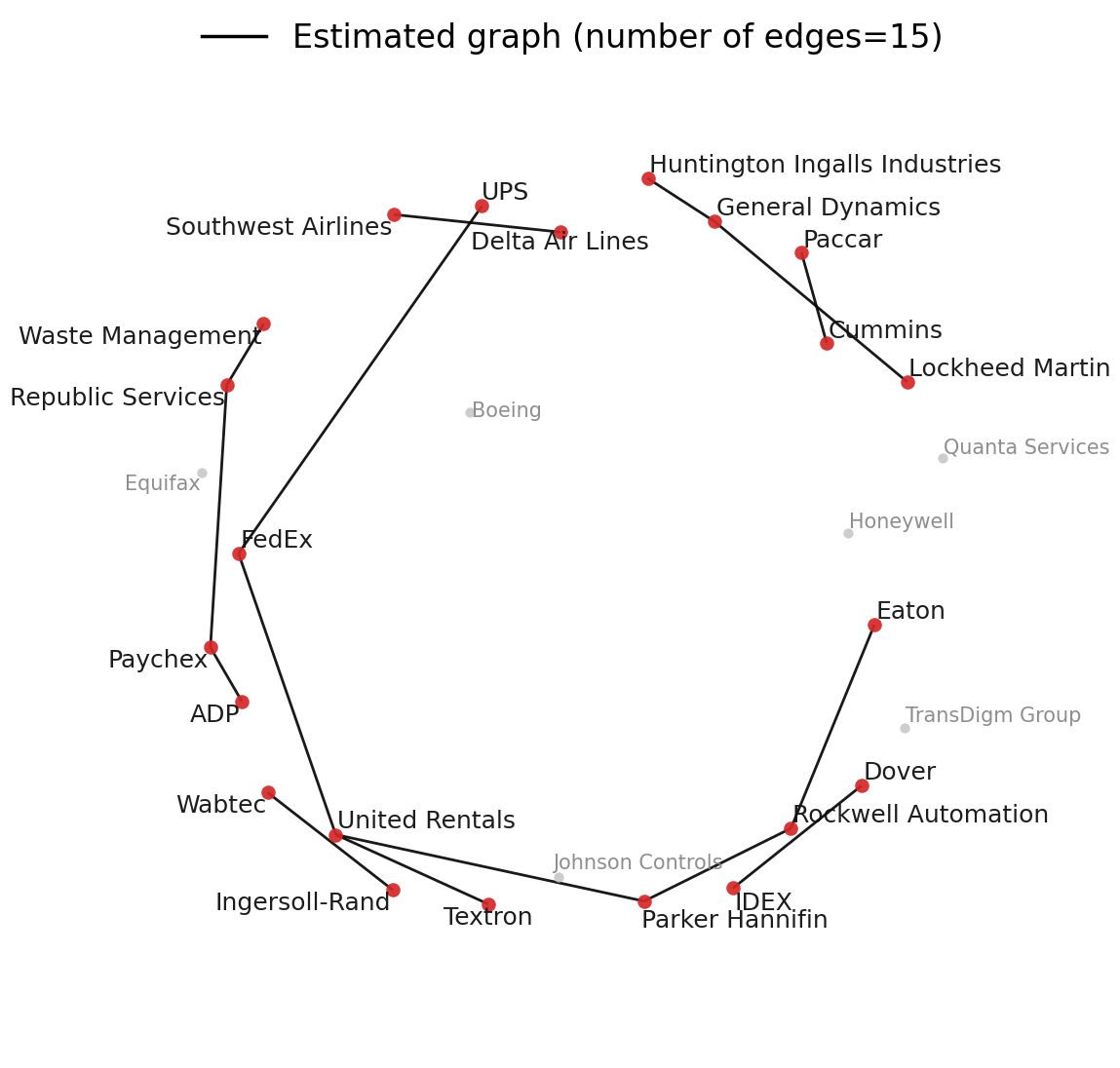}
    \label{figure:network_est}
    }
    \subfigure[$\widetilde H_{ij}(t)$ for all pairs $(i,j)$, colored by competitor label and by inclusion in $\widehat G$]{
    \includegraphics[width=0.465\linewidth]{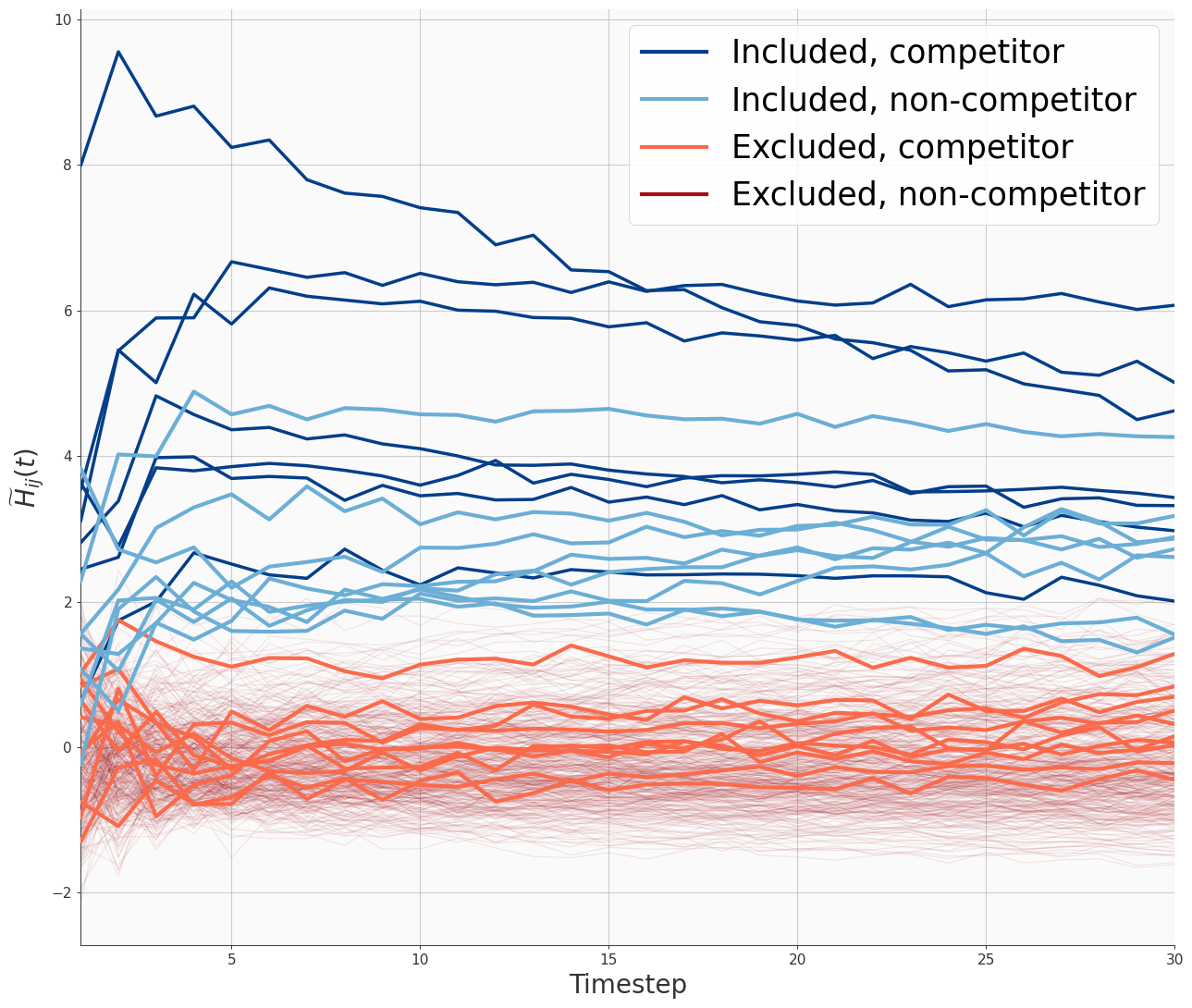}
    \label{fig:network_hessian}
    }        
    \caption{Graph estimation results for the $28$ industrials companies.}
\end{figure*}

\section{Conclusion} \label{sec:conclusion}

We have proposed a diffusion-based estimator of the conditional independence graph and shown that it is theoretically consistent and empirically competitive across a variety of data distributions.
We close by discussing two directions for future work.

First, as briefly noted in Section~\ref{sec:introduction}, a natural next step is to extend our analysis to high-dimensional settings in which the dimension $D$ grows with the sample size $n$.
Given that diffusion models have demonstrated strong empirical performance even when $D \gg n$, we expect that our framework can be extended to such scenarios with more delicate analysis.


Second, it would be interesting to extend our framework to discrete distributions, for which diffusion models based on continuous-time Markov chains (CTMCs) have recently achieved remarkable success in modeling high-dimensional discrete data such as natural language \citep{lou2024discrete, khanna2025mercury}.
Since our framework relies on the connection between the SDE solution and conditional independence via Tweedie's formula, investigating an analogous connection in the CTMC setting is a natural direction.

\section*{Data availability}
All code and data for reproducing the numerical experiments are publicly available at \url{https://github.com/hkkwon0922/nonparametric_graphical_diffusion}.

\section*{Acknowledgments}
This work was supported by Samsung Science and Technology Foundation under Project Number SSTF-BA2101-03, and the Singapore MOE Tier 1 A-8003581-00-00.
We disclose that AI-based tools were used to assist with English grammar editing and code development for the numerical experiments.

\clearpage
\pagebreak
\bibliographystyle{apalike}
\bibliography{bib-short}

\clearpage
\pagebreak
\appendix
\renewcommand{\contentsname}{Appendix}

\addtocontents{toc}{\protect\setcounter{tocdepth}{3}}
\renewcommand{\thesection}{\Alph{section}}
\tableofcontents

\setcounter{figure}{0}
\renewcommand{\thefigure}{A\arabic{figure}}

\setcounter{table}{0}
\renewcommand{\thetable}{A\arabic{table}}

\section{Related work and illustrative example}
\label{sec:supp_rel}

\subsection{Related work} \label{ssec:supp_related}

In this subsection, we review existing methods for undirected graphical model selection.
We also discuss existing work on nonparametric density estimation relevant to our method.

The factorization \eqref{eq:factorization} often yields convenient characterizations of the graph structure in certain distribution families.
For example, in the Gaussian case with covariance matrix $\Sigma \in \bbR^{D \times D}$, the condition $(i,j) \notin E_{0}$ is equivalent to $(\Sigma^{-1})_{ij} = 0$ for all $i \neq j \in [D]$; see Proposition~5.2 of \cite{lauritzen1996graphical}.
Thus, Gaussian graphical models recover the graph $G_{0}$ by estimating the precision matrix \citep{dempster1972covariance, meinshausen2006high, yuan2007model, friedman2008sparse}.
This idea extends to semiparametric settings, since conditional independence is invariant under componentwise monotone transformations \citep{liu2009nonparanormal}.
Accordingly, the precision matrix of the transformed Gaussian random vector is estimated after suitable marginal transformations~\citep{liu2009nonparanormal, liu2012copula, xue2012regularized}; these methods are known as nonparanormal graphical models.

While nonparanormal models offer greater flexibility, they still fundamentally rely on a latent Gaussian structure. To move beyond such constraints, it is natural to investigate whether the graph structure can be characterized in a fully nonparametric manner.
To the best of our knowledge, however, existing methods rely on restrictive assumptions.
For example, $G_{0}$ is restricted to be an undirected acyclic graph (forest), which in particular implies $d \leq 2$, where $d$ denotes the maximum clique size of $G_{0}$.
They exploit the fact that, for this graph class, the factorization \eqref{eq:factorization} admits an explicit representation in terms of the univariate and bivariate marginal densities of $p_{0}$.

The methods most closely related to ours are those that exploit the equivalence \eqref{eq:iff_cond} to recover $G_{0}$ through the Hessian of $\log p_{0}$ \citep{baptista2024learning, zheng2023generalized, liaw2025learning}.
These methods first estimate $\log p_{0}$ or $\nabla \log p_{0}$ under suitable regularization and then differentiate to obtain a pointwise estimator of $\nabla^2 \log p_0(\bx)$.
The graph estimator is then constructed by thresholding the estimated entries of the matrix $\bbE [ \vert \nabla^2 \log p_{0}(\bU) \vert ]$ for a suitable random vector $\bU$, where $\vert \cdot \vert$ is applied componentwise to the matrix.
The SING model \citep{baptista2024learning} assumes that $p_{0}$ is the density of the pushforward of a standard normal random vector under a monotone lower-triangular map parameterized by tensor-product Hermite polynomials.
It establishes consistency by estimating the transport map within the same model class and then taking second-order derivatives of the logarithm of the induced pushforward density.
While selection consistency is not established in~\cite{zheng2023generalized}, the graph is estimated by first learning $\nabla \log p_{0}(\bx)$ with a neural network via score matching~\citep{hyvarinen2005estimation}, and then differentiating the learned score to estimate $\nabla^2 \log p_0$.
The L-SING model \citep{liaw2025learning} estimates the conditional density of $X_{0,i}$ given the remaining variables for each $i \in [D]$ and uses the Hessian of the corresponding log-conditional density to construct the graph.

A common limitation of these works is that they all require pointwise estimation of the Hessian of $\log p_{0}$, which is intrinsically more difficult than estimating the density or its score; see Section~6 of \cite{kwon2026nonparametric} and the references therein.
Our method circumvents this difficulty by using diffusion models to estimate the Hessian of $\log p_{t}$ instead of that of $\log p_{0}$.
This is made possible by the minimax optimality of diffusion models in nonparametric density estimation.

The minimax optimal rate for estimating a standard smooth density in the $\beta$-\Holder \ class is well known to be $n^{-\beta/(2\beta + D)}$ \citep{tsybakov2008introduction, gine2016mathematical}, for example, under the total variation distance.
Various modern deep generative models also attain this rate up to logarithmic factors \citep{kwon2024minimax, oko2023diffusion, puchkin2024rates}.
Moreover, diffusion models can attain the rate $n^{-\beta/(2\beta+d)}$ for every $\beta > 0$ and $d \leq D$ \citep{kwon2026nonparametric, fan2025optimal}, which is faster than the standard rate $n^{-\beta/(2\beta+D)}$.
In fact, for any constants $\beta > 0$ and $d_0 \leq D$, the rate $n^{-\beta/(2\beta+ d_0)}$ is minimax optimal over densities in the $\beta$-\Holder \ class whose conditional independence graph has maximum clique size at most $d_0$ \citep{kwon2026nonparametric}.
This suggests that diffusion models are natural candidates for density estimation regardless of whether the underlying graph $G_{0}$ is sparse.

\subsection{Illustrative example}
\label{ssec:supp_illu}

In this subsection, we provide an illustrative example motivating our method.
The definition of $H_{ij}(\bx,t)$ is given in Section~\ref{sec:method}.
Here, we let $\bX_0 = (X_{0,1}, X_{0,2}, X_{0,3})$ follow a Gaussian distribution in which $X_{0,1}$ and $X_{0,3}$ are conditionally independent given $X_{0,2}$.
Figure~\ref{fig:gaussian} presents the absolute values of $H_{ij}(t) = H_{ij}(\bx,t)$.
We suppress the dependence on $\bx$, since $p_t$ remains Gaussian for all $t \geq 0$, and hence $H_{ij}(t)$ is simply the negative of the $(i,j)$ entry of the precision matrix of $\bX_t$.
As shown in the figure, $\vert H_{13}(t) \vert$ converges to zero as $t \to 0$, whereas the other entries remain bounded away from zero, indicating that the edge $(1,3)$ is absent from the graph.
As $t \rightarrow \infty$, however, $\bX_t$ rapidly approaches the standard Gaussian, and all $\vert H_{ij}(t) \vert$ converge to zero.
Thus, to characterize the graph structure through the values of $\vert H_{ij}(t) \vert$, it is necessary to take $t$ sufficiently small.

\begin{figure}[!t]
    \centering
    \begin{minipage}[t]{0.35\linewidth}
        \centering
        \vspace{20pt}
        \begin{tikzpicture}[every node/.style={circle, draw, inner sep=2.0pt}]
            \node (x1) at (0,0) {$X_{0,1}$};
            \node (x2) at (2.0,0) {$X_{0,2}$};
            \node (x3) at (4.0,0) {$X_{0,3}$};
            \draw (x1) -- (x2) -- (x3);
        \end{tikzpicture}
        \\
        \bean
        \Sigma_0^{-1} =
        \begin{pmatrix}
        1 & \ 0.7 & \ 0\\
        0.7 & \ 1 & \ 0.5\\
        0 & \ 0.5 & \ 1
        \end{pmatrix}
        \eean
    \end{minipage}
    \hfill
    \begin{minipage}[t]{0.6\linewidth}
        \centering
        \vspace{0pt}
        \includegraphics[width=\linewidth]{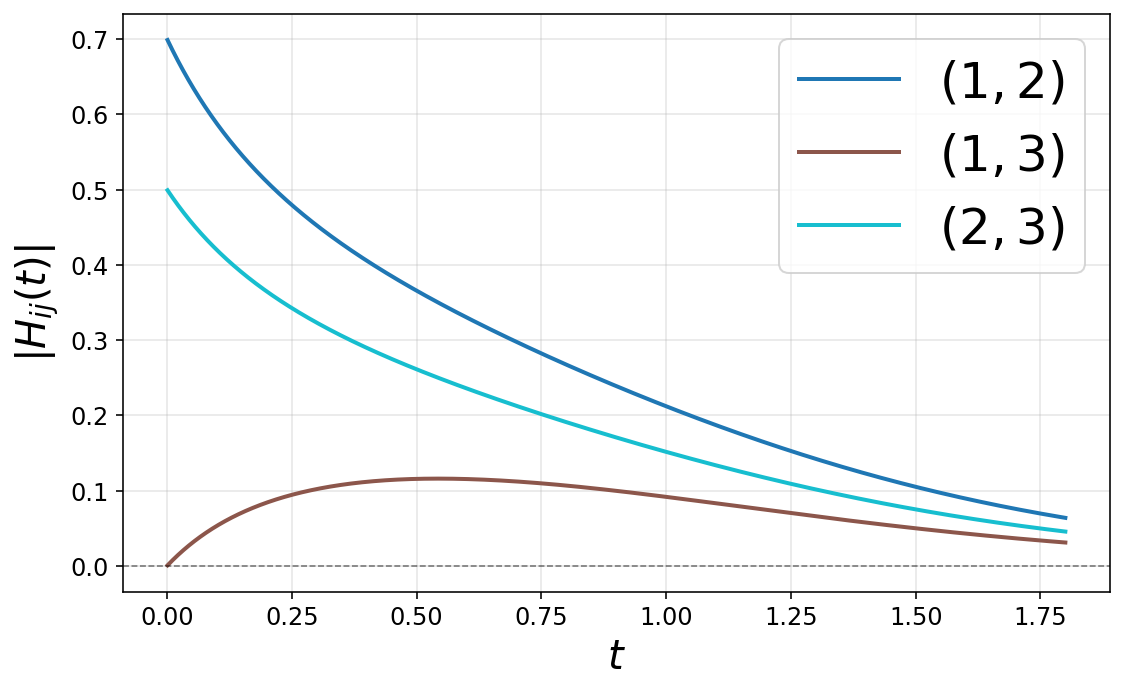}
    \end{minipage}
    \caption{ $\vert H_{ij}(t) \vert$ for 3-dimensional Gaussian distribution with covariance matrix $\Sigma_{0}$ and the corresponding conditional independence graph.}
    \label{fig:gaussian}
\end{figure}

\section{Proofs}
\label{sec:supp_proofs}

In this section, we provide the proofs of the theorems and the corollary in Section~\ref{sec:consistency}.
We first present an overview of the key ideas behind the proofs in Section~\ref{ssec:proof_outlines}, and then give the complete proofs in Sections~\ref{ssec:proof_thm1}--\ref{ssec:proof_thm2}.

\subsection{Proof outlines}
\label{ssec:proof_outlines}

In this subsection, we outline the key ideas behind the proofs of Theorems~\ref{thm1} and~\ref{thm2}.
We begin with Theorem~\ref{thm1}.

\subsubsection{Theorem~\ref{thm1}}

Note that
\bean
{\rm{Cov}} \big[ \bX_{0} \mid \bX_{t} = \bx \big] =  \left( \frac{\sigma_t}{\mu_t} \right)^2 \Cov \big[ \bZ_{0} \big],
\label{eq:pf1_change}
\eean
where $\bZ_{0}$ is the conditional random vector $(\mu_{t}\bX_{0} -\bx)/\sigma_{t}$ given $\bX_{t} = \bx$.
Here, we suppress the dependence on $\bx$ and $t$ by treating them as fixed.
Using the above identity, it suffices to derive upper and lower bounds for the covariance of $\bZ_{0}$ instead.

We first approximate $\bZ_{0}$ by a random vector $\widetilde \bZ$ whose covariance can be characterized explicitly.
Since the conditional density of $\bX_{0}$ given $\bX_{t} = \bx$ is proportional, up to normalization, to the map $\bx_{0} \mapsto p_{0}(\bx_{0}) \phi_{\sigma_{t}} (\bx - \mu_{t} \bx_{0})$, the Lebesgue density of $\bZ_{0}$ is proportional to the map
\bean
\bz \mapsto \exp \big( u_{0} (\bz) \big),
\qquad 
u_{0} (\bz) =  \log p_{0} \left( \frac{\bx + \sigma_t \bz}{\mu_t} \right) -  \frac{\| \bz \|_2^2}{2},
\eean
with support $S = \{ \bz : \| \bx + \sigma_{t} \bz \|_{\infty} \leq \mu_{t} \}$.
By assumption (\bS), $u_{0}$ is twice continuously differentiable on $S$.
Since $\mathbf{0}_{D} \in S$, let $\widetilde u : \bbR^{D} \to \bbR$ be the second-order polynomial given by the second-order Taylor expansion of $u_{0}(\bz)$ at $\bz = \mathbf{0}_{D}$. Then,
\be
\big\vert u_{0}(\bz) - \widetilde u(\bz) \big\vert \lesssim \sigma_{t}^{\beta \wedge 3} \| \bz \|_{\infty}^{\beta \wedge 3} \qquad \forall \bz \in S.
\label{eq:pf1_taylor}
\ee

Define $\widetilde \bZ = (\widetilde Z_{1},\ldots,\widetilde Z_{D})$ to be the Gaussian random vector with density proportional to the map $\bz \mapsto \exp(\widetilde u(\bz))$.
Then, the difference between the covariances of $\bZ_{0}$ and $\widetilde \bZ$ can be bounded as
\be
\Big\vert {\rm Cov} \big[ Z_{0,i}, Z_{0,j} \big] - {\rm Cov} \big[ \widetilde Z_{i}, \widetilde Z_{j} \big] \Big\vert
\lesssim \sigma_{t}^{\beta \wedge 3} \qquad \forall i,j \in [D].
\label{eq:pf1_gau_err}
\ee
Here, the error is mainly driven by the Taylor expansion error in \eqref{eq:pf1_taylor}, while the mismatch between the supports of $\widetilde \bZ$ and $\bZ_{0}$, namely $\bbR^{D}$ and $S$, is negligible.
Specifically, $\bbP (\widetilde \bZ \notin S)$ is sufficiently small, since $\| \bx \|_{\infty} \leq \mu_{t}\gamma < \mu_{t}$ and the set $S$ is sufficiently large for small $t$.

We now derive upper and lower bounds for the covariance of $\widetilde \bZ$.
Since $\widetilde \bZ$ is Gaussian, its covariance matrix is given by the inverse of the negative Hessian of its log-density.
By construction, $\nabla^2 \widetilde u(\bz) = \nabla^2 u_{0}(\mathbf{0}_{D}) = (\sigma_{t}/\mu_{t})^2 \bH - \bI_{D}$ for every $\bz \in \bbR^{D}$ and hence
\bean
\Cov\big[ \widetilde \bZ \big]
=  \bigg( \bI_{D} - \frac{\sigma_{t}^2}{\mu_{t}^2} \bH \bigg)^{-1}
= \sum_{k \geq 0} \frac{\sigma_{t}^{2k}}{\mu_{t}^{2k}} \bH^{k},
\eean
where the last equality holds for sufficiently small $t$ by the Neumann series expansion.

The key connection between the graph distance and the covariance appears here.
By the equivalence \eqref{eq:iff_cond}, the condition
$H_{ij}(\bx / \mu_{t},0) \neq 0$ implies $(i,j) \in E_{0}$ for $i \neq j \in [D]$.
The essential point is that $(\bH^{k})_{ij} = 0$ for every $k < d(i,j)$; see Lemma~\ref{lem1}.
The proof proceeds by contradiction: if $(\bH^{k})_{ij} \neq 0$ for some $k < d(i,j)$, then there exists a path of length $k$ from vertex $i$ to $j$, contradicting the definition of $d(i,j)$.
Therefore, all terms with $k < d(i,j)$ vanish in the $(i,j)$th entry of the series expansion, which yields
\bean
\Cov\big[ \widetilde Z_{i}, \widetilde Z_{j} \big]
= \sum_{k \geq d(i,j)} \big( \bH^{k} \big)_{ij} \big( \sigma_{t}/ \mu_{t} \big)^{2k}
\eean
and
\bean
\Big\vert {\rm Cov} \big[ \widetilde Z_{i}, \widetilde Z_{j} \big] - \big( \bH^{d(i,j)} \big)_{ij} \big( \sigma_{t} / \mu_{t} \big)^{2d(i,j)}
\Big\vert \lesssim \sigma_{t}^{2d(i,j)+2}.
\eean
Combining this with \eqref{eq:pf1_gau_err} and the triangle inequality, we have
\bean
\Big\vert {\rm Cov} \big[ Z_{0,i}, Z_{0,j} \big] - \big( \bH^{d(i,j)} \big)_{ij} \big( \sigma_{t}/\mu_{t} \big)^{2d(i,j)}
\Big\vert \lesssim \sigma_{t}^{(2d(i,j)+2)\wedge \beta \wedge 3}.
\eean
The assertion then follows from the relation $\Cov[X_{0,i}, X_{0,j} \mid \bX_{t} = \bx] \asymp \sigma_{t}^2 \Cov[Z_{0,i}, Z_{0,j}]$.

\subsubsection{Theorem~\ref{thm2}}

Here, we provide a sketch of the proof of Theorem~\ref{thm2}.
For every $t \leq {\widetilde T}$, Corollary~\ref{cor1} implies that
\bean
\bbE \big[ \big\vert H_{ij} (\bU_{t},t) \big\vert \big] -\tau \geq C_{\gamma} - C_{2} \sigma_{t}^{\tilde \beta} - \tau > 0
\qquad  \forall (i,j) \in E_{0}.
\eean
Combining this with Markov's inequality, the false negative rate is bounded as
\bean
\bbP  \Big( \bbE_{n} \big[  \big\vert \widehat H_{ij} \big( \bU_{t} , t \big)  \big\vert \big] \leq \tau \Big)
&& \leq  \bbP \Big( \Big\vert \bbE_{n} \big[  \big\vert \widehat H_{ij} \big( \bU_{t} , t \big)  \big\vert \big] - \bbE \big[  \big\vert  H_{ij} \big(  \bU_{t} , t \big)  \big\vert \big] \Big\vert
\geq \bbE \big[  \big\vert  H_{ij} \big( \bU_{t}, t \big)  \big\vert \big]  - \tau
\Big)
\\
&& \leq \frac{\cE_{ij,t}}{ C_{\gamma} - C_{2} \sigma_{t}^{\tilde \beta} - \tau } \qquad \forall (i,j) \in E_{0},
\eean
where
\bean
\cE_{ij,t}
= \bbE \left[ \Big\vert \bbE_{n} \big[  \big\vert \widehat H_{ij} \big( \bU_{t} , t \big)  \big\vert \big]
- \bbE \big[  \big\vert  H_{ij} \big(  \bU_{t} , t \big)  \big\vert \big] \Big\vert \right],
\qquad i \neq j \in [D],
\eean
is the estimation error for the $(i,j)$th entry of the averaged Hessian at time $t$.
Similarly, the false positive rate is bounded as
\bean
\bbP  \Big( \bbE_{n} \big[  \big\vert \widehat H_{ij} \big( \bU_{t} , t \big)  \big\vert \big] > \tau \Big)
\leq \frac{\cE_{ij,t}}{ \tau - C_{2} \sigma_{t}^{\tilde \beta} }
\qquad \forall (i,j) \notin E_{0} \text{ with } i \neq j.
\eean
Therefore, it suffices to derive an upper bound on $\cE_{ij,t}$.

Recall that $H_{ij}(\bx,t)$ and $\widehat H_{ij}(\bx,t)$ are defined through the conditional covariances of $\bX_{0}$ given $\bX_{t} = \bx$ and of $\widehat \bX_{\underline{T}}$ given $\widehat \bX_{t} = \bx$, respectively, scaled by $\sigma_{t}^{-4}\mu_{t}^2$.
By the triangle inequality,
\be
\cE_{ij,t}
\leq 
 \Big\vert  \bbE \big[  \big\vert H_{ij} \big( \bU_{t}, t \big)  \big\vert  \big] -  \bbE \big[   \big\vert \widetilde H_{ij} \big( \bU_{t}, t \big)  \big\vert \big]  \Big\vert 
 +
 \bbE \Big[  \Big\vert \bbE \big[  \big\vert \widetilde H_{ij} \big( \bU_{t}, t \big)  \big\vert \big] - \bbE_{n} \big[  \big\vert \widehat H_{ij} \big( \bU_{t}, t \big)  \big\vert \big] \Big\vert  \Big],
 \label{eq:pf2_triangle}
\ee
where $\widetilde H_{ij} (\bx,t) = \sigma_{t}^{-4} \mu_{t}^2 \text{Cov} [ X_{\underline{T}, i}, X_{\underline{T}, j} \mid \bX_{t} = \bx ]$.
The first term on the right-hand side of \eqref{eq:pf2_triangle} measures the truncation error of the conditional covariance at time $\underline{T}$, scaled by $\sigma_{t}^{-4} \mu_{t}^2$.
More precisely, for $i \neq j \in [D]$, we have
\bean
\Big\vert \text{Cov} \big[ X_{\underline{T},i}, X_{\underline{T},j} \mid \bX_{t} = \bx \big] - \text{Cov} \big[ X_{0,i}, X_{0,j} \mid \bX_{t} = \bx \big] \Big\vert
\lesssim 
\frac{\underline{T}}{t^{2}} \Big\vert \text{Cov} \big[ X_{0,i}, X_{0,j} \mid \bX_{t} = \bx \big] \Big\vert.
\eean
Combining this with Theorem~\ref{thm1}, the first term on the right-hand side of \eqref{eq:pf2_triangle} is bounded above, up to a constant, by $\underline{T}/t^{2}$ for $(i,j) \in E_{0}$ and by $\underline{T}/t^{2-\tilde \beta / 2}$ for $(i,j) \notin E_{0}$ with $i \neq j$.


The second term on the right-hand side of \eqref{eq:pf2_triangle} measures the difference between the averaged conditional covariances of $\widehat \bX_{\underline{T}}$ given $\widehat \bX_{t}$ and $\bX_{\underline{T}}$ given $\bX_{t}$.
Given $\bx \in [-\mu_{t}\gamma, \mu_{t}\gamma]^{D}$, we first consider the pointwise difference.
For $s \in [\underline{T}, t]$, let $\widehat p_{s \mid t} (\cdot \mid \bx)$ and $p_{s \mid t}(\cdot \mid \bx)$ denote the Lebesgue densities of the conditional random vectors $\widehat \bX_{s}$ given $\widehat \bX_{t} = \bx$ and $\bX_{s}$ given $\bX_{t} = \bx$, respectively.
Then, based on well-known results \citep{bogachev2016distances, le2016brownian, chen2023sampling}, the total variation distance between the solutions of the two SDEs can be bounded by controlling the difference between their drift terms, as
\bean
\Big\{ d_{\rm TV} \big( p_{\underline{T} \mid t} (\cdot \mid  \bx), \widehat p_{\underline{T} \mid t} (\cdot \mid \bx) \big) \Big\}^2
\leq  \int_{\underline{T}}^{t} \int_{\bbR^{D}}  \big\| \widehat \bff( \bz, s ) - \bff_{0} (\bz ,s) \big\|_2^{2} p_{s \mid t}(\bz \mid \bx) \d \bz \d s.
\eean

This requires control of the \textit{conditional} $L^2$-error between $\widehat \bff(\cdot, s)$ and $\bff_{0}(\cdot,s)$, where the expectation is taken with respect to the conditional distribution of $\bX_{s}$ given $\bX_{t} = \bx$ for $s \in [\underline{T},t]$.
However, condition \eqref{thm2:assmp1} only guarantees the \textit{marginal} $L^2$-error, where the expectation is taken with respect to the marginal distribution of $\bX_{s}$ for $s \in [\underline{T},t]$.
Note that $p_{t}(\bx) \gtrsim 1$ whenever $\| \bx \|_{\infty} \leq \mu_{t}$, since $p_{0}$ is bounded from below on its support; see Lemma~6 of \cite{kwon2026nonparametric}.
Combining this with Bayes' rule yields
\bean
p_{s \mid t} (\bz \mid \bx) \lesssim \phi_{\sigma_{t-s}}(\bx - \mu_{t-s} \bz)\, p_{s} (\bz),
\eean
which in turn implies that $\bbE[ p_{s \mid t}(\bz \mid \bU_{t}) ] \lesssim p_{s}(\bz)$.
Thus, after taking expectation with respect to $\bU_{t}$, the conditional $L^2$-error can be controlled by the marginal $L^2$-error, yielding
\bean
\bbE \Big[ d_{\rm TV} \big( p_{\underline{T} \mid t} (\cdot \mid  \bU_{t}), \widehat p_{\underline{T} \mid t} (\cdot \mid \bU_{t}) \big) \Big]
\lesssim \epsilon_{n}.
\eean
The difference between the covariances can then be controlled by the total variation distance, provided that the tail probabilities of both densities $p_{\underline{T} \mid t}(\cdot \mid \bx)$ and $\widehat p_{\underline{T} \mid t}(\cdot \mid \bx)$ are sufficiently small.
Combining this with condition \eqref{thm2:assmp2} and several basic lemmas based on It\^{o}'s formula \citep{le2016brownian}, the tail bounds introduce an additional logarithmic factor $\{\log (1/\epsilon_{n})\}^{3/2}$ in the convergence rate.

\subsection{Proof of Theorem~\ref{thm1}}
\label{ssec:proof_thm1} 

We begin by stating and proving a key lemma used in the proof of Theorem~\ref{thm1}.
This lemma shows that the truncation error of the Neumann series expansion can be controlled by the graph-induced sparsity pattern of the matrix.

\begin{lemma}\label{lem1}
Let $K > 0$ be given, and consider an undirected graph $G$ with vertex set $\{1,\ldots,D\}$.
Suppose that a symmetric matrix $\bH  \in \bbR^{D \times D}$ satisfies $\vert H_{ij} \vert \leq K$ for all $i,j \in [D]$ and $H_{ij} = 0$ when $d_{G}(i,j) > 1$.
Then, for every $0 \leq \epsilon < (DK)^{-1}$,
\bean
\left\vert \left\{ \left( \bI_{D} - \epsilon \bH \right)^{-1} \right\}_{ij} - \epsilon^{d_{G}(i,j)} \left( \bH^{d_{G}(i,j)} \right)_{ij} \right\vert 
\leq \frac{ (\epsilon D K)^{d_{G}(i,j)+1} }{1-\epsilon D K}
, \quad \forall i, j \in [D] \text{ with } d_{G}(i,j) < \infty,
\eean
and $\{ (\bI_{D} - \epsilon \bH)^{-1}\}_{ij} = 0$ when $d_{G}(i,j) = \infty$.
\end{lemma}

\begin{proof}
By the Cauchy–Schwarz inequality, we have $(\sum_{i=1}^{D} \vert x_i \vert)^2 \leq D \| \bx \|_2^2$ for any $\bx = (x_1,\ldots,x_{D}) \in \bbR^{D}.$
A simple calculation yields that
\bean
\bx^{\top} \bH \bx  = \sum_{i,j \in [D]} H_{ij} x_{i} x_{j}
\leq K  \sum_{i,j \in [D]} \vert x_i \vert \vert x_j \vert = K \| \bx \|_1^2 \leq D K \| \bx \|_2^2.
\eean
It follows that, for any $\epsilon > 0$,
\bean
\bx^{\top} \left( \bI_{D} - \epsilon \bH \right) \bx 
= \| \bx \|_2^2 - \epsilon \left( \sum_{i,j \in [D]} H_{ij} x_i x_j \right)
\geq (1-\epsilon D K)  \| \bx \|_2^2.
\eean
Hence, for every $0 \leq \epsilon < (DK)^{-1}$, the matrix $\bI_{D}-\epsilon \bH$ is positive-definite.
Let
\bean
\bSigma
=
\left( \bI_{D} - \epsilon \bH \right)^{-1}.
\eean
Then, for every $N \in \bbZ_{\geq 0}$, a simple calculation yields that
\bean
\bSigma - \sum_{k=0}^{N} \epsilon^{k} \bH^{k} = \epsilon^{N+1} \bSigma \bH^{N+1}
\eean
because $ \bSigma^{-1}  (\sum_{k=0}^{N} \epsilon^{k} \bH^{k} ) = ( \bI_{D} - \epsilon \bH ) (\sum_{k=0}^{N} \epsilon^{k} \bH^{k} )
= \bI_{D} - \epsilon^{N+1} \bH^{N+1}.$
For any $i,j \in [D]$,
\bean
\left\vert \left( \bSigma \bH^{N+1} \right)_{ij} \right\vert
= \left\vert \bfe_{i}^{\top} \bSigma \bH^{N+1} \bfe_{j} \right\vert
\leq \left\| \bfe_{i}^{\top} \bSigma \right\|_2
\left\| \bH^{N+1} \bfe_{j} \right\|_2
\leq \Tnorm{  \bSigma }_{2} \Tnorm{ \bH }_2^{N+1},
\eean
where $\bfe_{i} \in \{0,1\}^{D}$ denotes the standard basis vector whose $i$-th element is one and all other elements are zero.
Here, $\tnorm{\cdot}_2$ denotes the spectral norm  (or $\ell^2$ operator norm).
Since $\bx^{\top} \bSigma^{-1} \bx \geq (1-\epsilon KD) \| \bx \|_2^2$ for any $\bx \in \bbR^{D}$, the smallest eigenvalue of $\bSigma^{-1}$ is larger than $1-\epsilon  K D$; therefore, $ \tnorm{\bSigma}_2 \leq 1/(1-\epsilon KD)$.
Moreover, 
\bean
&& \Tnorm{\bH}_2 = \sup_{\| \bx \|_2 \leq 1} \| \bH \bx \|_2 = \sup_{\| \bx \|_2 \leq 1} \sqrt{ \sum_{i=1}^{D} \left( \sum_{j=1}^{D} H_{ij} x_j \right)^2 } \leq K \sup_{\| \bx \|_2 \leq 1} \sqrt{ \sum_{i=1}^{D} \left( \sum_{j=1}^{D} \vert x_j \vert \right)^2 }
\\
&& \leq D K \sup_{\| \bx \|_2 \leq 1} \| \bx \|_2 = DK.
\eean
Combining the last three displays, we have
\be
\left\vert \Sigma_{ij} - \sum_{k=0}^{N} \epsilon^{k} (\bH^{k})_{ij} \right\vert
\leq \frac{(\epsilon D K)^{N+1}}{1-\epsilon D K}
\label{eqlem1:errorbd}
\ee
for every $0 \leq \epsilon < (DK)^{-1}$ and $i,j \in [D]$.

Note that
\bean
(\bH^{k})_{ij} = \sum_{ \substack {i_0,\ldots,i_{k} \in [D] \\ i_0 = i, i_k = j }}  \prod_{l=0}^{k-1} H_{i_{l} i_{l+1} }
\eean
for $k \geq 1$, and $(\bH^{0})_{ij} = \delta_{ij}$, where $\delta_{ij}$ denotes the Kronecker delta.
For $i,j \in [D]$ with $i \neq j$ and $k \in \bbZ_{\geq 0}$ with $k < d_{G}(i,j)$, we have $(\bH^{k})_{ij} = 0$.
The proof proceeds by contradiction.
Suppose that $(\bH^k)_{ij} \neq 0$.
Then, there exists a sequence $(i^*_{0},\ldots,i^*_{k})$ such that $\prod_{l=0}^{k-1} H_{i^*_l i^*_{l+1}} \neq 0$ with $i_0^{*} = i$ and $i_{k}^{*} = j$.
This implies the existence of a path of length $k$ from vertex $i$ to $j$, which contradicts the assumption $k < d_{G}(i,j)$. 

For $i,j \in [D]$ with $d_{G}(i,j) < \infty$, we have
\bean
\sum_{k=0}^{d_{G}(i,j)} \epsilon^{k} (\bH^{k})_{ij} = \epsilon^{d_{G}(i,j)} \left( \bH^{d_{G}(i,j)} \right)_{ij}.
\eean
Combining with (\ref{eqlem1:errorbd}), we have
\bean
\left\vert \Sigma_{ij} -  \epsilon^{d_{G}(i,j)} \left( \bH^{d_{G}(i,j)} \right)_{ij}  \right\vert
\leq
 \frac{(\epsilon D K)^{d_{G}(i,j)+1}}{1-\epsilon D K}
\eean
for every $0 \leq \epsilon < (DK)^{-1}$ and $i,j \in [D]$ with $d_{G}(i,j) < \infty$.
Moreover, $(\bH^{k})_{ij} = 0$ for any $k \in \bbZ_{\geq 0}$ with $d_{G}(i,j) = \infty$.
Since $0 \leq \epsilon D K < 1$, combining with (\ref{eqlem1:errorbd}), we have $\Sigma_{ij} = 0$ for $d_{G}(i,j) = \infty$.

\end{proof}

\medskip

\textit{Proof of Theorem~\ref{thm1}.}
Let $p_{0 \mid t}(\cdot \mid \bx)$ be the conditional density of $\bX_{0}$ given $ \bX_{t} = \bx$, given as 
\bean
p_{0 \mid t}(\bx_0 \mid \bx) = p_{0}(\bx_0) \exp \left( -\frac{ \|\bx - \mu_t \bx_0 \|_2^2 }{2\sigma_{t}^2} \right)  \left\{ \int_{[-1,1]^{D}} p_{0}(\by) \exp \left( -\frac{ \|\bx - \mu_t \by \|_2^2 }{2\sigma_{t}^2} \right) \d \by \right\}^{-1}
\eean
for $\bx_0 \in [-1,1]^{D}$.
Throughout this proof, we fix $\bx \in [-\mu_t \gamma, \mu_t \gamma]^{D}$ and work conditional on $\bX_t = \bx$. For simplicity, we often suppress this conditioning in the notation.

Let $u_0(\cdot) = \log p_0(\cdot)$ and $\bZ_0 = (\mu_t \bX_0 - \bx)/\sigma_t$.
Then, the density of $\bZ_0$ (conditional on $\bX_t = \bx$) is given as
\bean
p_{\bZ_0}(\bz) = \exp \left( u_0 \left( \frac{\bx + \sigma_t \bz}{\mu_t} \right) -  \frac{\| \bz \|_2^2}{2}  \right) 
\left\{ \int_{S} \exp \left( u_0 \left( \frac{\bx + \sigma_t \by}{\mu_t} \right) -  \frac{\| \by \|_2^2}{2}  \right) \d \by  \right\}^{-1},
\quad \bz \in S,
\eean
where $ S = \{ \bz : \|\bx + \sigma_t \bz \|_{\infty} \leq \mu_t \}$. 
Note also that
\bean
{\rm{Cov}} \big[ X_{0,i}, X_{0,j} \mid \bX_{t} = \bx \big] =  \left( \frac{\sigma_t}{\mu_t} \right)^2 \text{Cov} \big[ Z_{0,i}, Z_{0,j} \big],
\eean
where $\bZ_{0} = ( Z_{0,1}, \ldots, Z_{0,D} )$.

We approximate $u_0$ by its second-order Taylor expansion.
Specifically, for every $\bz \in S$,
\be
 u_0 \left( \frac{\bx + \sigma_t \bz}{\mu_t} \right)
=  u_0 \left( \frac{\bx}{\mu_t} \right)
 +  \frac{\sigma_t}{\mu_t}  \left( \nabla u_0 \left( \frac{\bx}{\mu_t} \right) \right)^{\top} \bz
 \quad + \left( \frac{\sigma_t^2}{2 \mu_t^2} \right)  \bz^{\top}  \nabla^2 u_0 \left( \frac{\bx}{\mu_t} \right) \bz
 + r(\bz),
 \label{eqthm1:qd_approx}
\ee
where
\bean
r(\bz) = \left( \frac{\sigma_t^2}{\mu_t^2} \right) \sum_{ \alpha. = 2} \left\{ \left( \D^{\balpha} u_0 \right) \left( \frac{\bx + \xi \sigma_t \bz}{\mu_t} \right) - \left( \D^{\balpha} u_0 \right) \left( \frac{\bx}{\mu_t} \right) \right\} \frac{ \bz^{\balpha} }{\balpha !}
\eean
for a suitable $\xi \in [0,1]$.
Here, $\balpha = (\alpha_1,\ldots,\alpha_{D}) \in \bbZ_{\geq 0}^{D}$ is a multi-index with $\balpha! = \prod_{k=1}^{D} ( \alpha_k! ) $ and $\bz^{\balpha} = \prod_{k=1}^{D} z_k^{\alpha_k}$.
Since $u_0 \in  \cH^{\beta,K}_{D}([-1,1]^{D})$, there exists a positive constant $D_{1} = D_{1} (D,K)$ such that
\be
\vert r(\bz) \vert \leq D_{1} \left( \frac{\sigma_t}{\mu_t} \right)^{\beta \wedge 3} \| \bz \|_{\infty}^{\beta \wedge 3}.
\label{eqthm1:rem}
\ee

Let $ \bH = (H_{ij}) = \nabla^2 u_0 ( \bx / \mu_t) \in \bbR^{D \times D}$.
Since $u_0 \in  \cH^{\beta,K}_{D}([-1,1]^{D})$, we have $\vert H_{ij} \vert \leq K$ for all $i,j \in [D]$.
Hence, for any $\by \in \bbR^{D}$, 
\bean
\by^{\top} \bH \by
= \sum_{i,j \in [D]} H_{ij} y_i y_j
\leq \sum_{i,j \in [D]} K \vert y_i \vert \vert y_j \vert = K \| \by\|_1^2 \leq DK \|\by \|_2^2,
\eean
where the last inequality holds by the Cauchy--Schwarz inequality.
This bound further implies that
\be
\by^{\top} \left\{ \bI_{D} -  \left( \frac{\sigma_t}{\mu_t} \right)^2 \bH \right\} \by
\geq \| \by \|_2^2 - D K \left( \frac{\sigma_t}{\mu_t} \right)^2 \| \by \|_2^2
= 
\kappa_{t} \| \by \|_2^2, \quad \forall \by \in \bbR^{D},
\label{eqthm1:matrix_pd}
\ee
where $\kappa_{t} = 1- D K (\sigma_t /\mu_t)^2 $. 
Let $D_{2} = D_{2}(D,K)$ be a positive constant such that $ (\sigma_{t} / \mu_{t})^2 < (DK)^{-1} $ for every $t \leq D_{2}$.
Then, for every $t \leq D_{2}$, we have $0 < \kappa_{t} \leq 1$, and consequently, the matrix $\bI_{D} - (\sigma_t / \mu_t)^{2} \bH  $ is positive-definite.
Combining this with Lemma~\ref{lem1}, for every $t \leq D_{2}$ and $i,j \in [D]$, we have
\be
\left\vert \widetilde \Sigma_{ij} -  \left( \frac{\sigma_{t}}{\mu_{t}} \right)^{2 d(i,j) } \left( \bH^{d(i,j)} \right)_{ij}  \right\vert
\leq   \frac{ (1 - \kappa_t)  ( D K )^{ d(i,j)}  }{ \kappa_t } \left( \frac{\sigma_{t}}{\mu_{t}} \right)^{2 d(i,j)}
, \quad \text{if } d(i,j) < \infty,
\label{eqthm1:prop}
 \ee
and $\widetilde \Sigma_{ij} = 0$ if $d(i,j) = \infty$,  where 
\bean
\widetilde \bSigma = (\widetilde \Sigma_{ij}) = \left\{ \bI_{D} -  \left( \frac{\sigma_t}{\mu_t} \right)^2 \bH \right\}^{-1}.
\eean
Motivated by the quadratic approximation of $u_0$ in (\ref{eqthm1:qd_approx}), we consider the Gaussian random vector $\widetilde \bZ = (\widetilde Z_{1},\ldots, \widetilde Z_{D}) \sim \cN(\widetilde \bmu, \widetilde \bSigma)$, where 
\bean
\widetilde \bmu
= ( \widetilde \mu_1, \ldots, \widetilde \mu_{D})^{\top}
= \left( \frac{\sigma_t}{\mu_t} \right)  \widetilde \bSigma \nabla u_0 \left( \frac{\bx}{\mu_{t}} \right).
\eean
Our goal is to derive both upper and lower bounds on $\vert \text{Cov}[Z_{0,i} , Z_{0,j}] \vert $.
Accordingly, we decompose the covariance as follows:
\bean
&& \Big\vert \text{Cov} \big[ Z_{0,i}, Z_{0,j} \big] \Big\vert
\leq \Big\vert \widetilde \Sigma_{ij} \Big\vert  + \Big\vert \text{Cov} 
\big[ Z_{0,i}, Z_{0,j} \big] - \widetilde \Sigma_{ij} \Big\vert 
\quad \text{and}
\\
&& \Big\vert \text{Cov} \big[ Z_{0,i}, Z_{0,j} \big] \Big\vert
\geq \Big\vert \widetilde \Sigma_{ij} \Big\vert  - \Big\vert \text{Cov} \big[ Z_{0,i}, Z_{0,j} \big] - \widetilde \Sigma_{ij} \Big\vert.
\eean
We proceed by deriving upper and lower bounds for $\vert \widetilde \Sigma_{ij} \vert$, 
and an upper bound for $\vert \text{Cov}[Z_{0,i}, Z_{0,j}] - \widetilde \Sigma_{ij} \vert$; 
see (\ref{eqthm1:cov_bd}), (\ref{eqthm1:cov_lower_bd}), and (\ref{eqthm1:cov_gau_bd}).

For any $i,j \in [D]$ with $d(i,j) < \infty$,
\bean
\left\vert \left( \bH^{d(i,j)} \right)_{ij} \right\vert
= \left\vert \bfe_{i}^{\top} \bH^{d(i,j)} \bfe_{j} \right\vert 
\leq \Tnorm{\bH^{d(i,j)}}_2
\leq \Tnorm{\bH}_2^{d(i,j)},
\eean
where $\bfe_{i} \in \{0,1\}^{D}$ denotes the $i$-th unit vector and $\tnorm{\cdot}_2$ denotes the spectral norm (or $\ell^2$ operator norm).
Moreover, a simple calculation yields that
\bean
\Tnorm{\bH }_2
&&  = \sup_{\| \by \|_2 \leq 1} \| \bH \by \|_2 
= \sup_{\| \by \|_2 \leq 1} \sqrt{ \sum_{i=1}^{D} \left( \sum_{j=1}^{D} H_{ij} y_j \right)^2 }
\\
&& \leq K \sup_{\| \by \|_2 \leq 1} \sqrt{ \sum_{i=1}^{D} \left( \sum_{j=1}^{D} \vert y_j \vert \right)^2 }
= K  \sup_{\| \by \|_2 \leq 1} \sqrt{ D \| \by \|_1^2 }
\leq D K \sup_{\| \by \|_2 \leq 1} \| \by \|_2 = DK.
\eean
Combining the last two displays yields that
\bean
\left\vert \left( \bH^{d(i,j)} \right)_{ij} \right\vert
\leq (DK)^{d(i,j)},
\quad \forall i,j \in [D] \text{ with }d(i,j) < \infty.
\eean
Let $D_{3} = D_{3}(D_{2}, D,K)$ be a positive constant such that $D_{3} \leq D_{2}$ and
$1/2 \leq \kappa_{t} \leq 1$ for every $t \leq D_{3}$.
Then, $0 \leq ( 1- \kappa_{t}) / \kappa_t \leq 1$.
Combining the last display with (\ref{eqthm1:prop}), for all $t \leq D_{3}$ and $i,j \in [D]$ with $d(i,j) < \infty$, we have
\be
\begin{split}
 \left\vert \widetilde \Sigma_{ij} \right\vert
&  \leq \left( \frac{\sigma_t}{\mu_t} \right)^{2 d(i,j)} \left\vert \left( \bH^{d(i,j)} \right)_{ij} \right\vert 
+ \frac{ (1 - \kappa_t)  ( D K )^{ d(i,j)}  }{ \kappa_t } \left( \frac{\sigma_{t}}{\mu_{t}} \right)^{2 d(i,j)}
\\
& \leq 2  (DK)^{d(i,j)} \left( \frac{\sigma_t}{\mu_t} \right)^{2 d(i,j)}
\leq  2 \left\{  (DK)^{D-1} \vee 1 \right\} \left( \frac{\sigma_t}{\mu_t} \right)^{2 d(i,j)},
\label{eqthm1:cov_bd}
\end{split}
\ee
where the last inequality follows from the fact that $d(i,j) \leq D-1$ whenever $d(i,j) < \infty$.

To obtain a corresponding lower bound, we use the elementary inequality $\vert a \vert \geq \vert b \vert - \vert a-b \vert$ for any $a,b \in \bbR$.
Combining this with (\ref{eqthm1:prop}), for all $t \leq D_{3}$ and $i,j \in [D]$ with $d(i,j) < \infty$, we have
\bean
\left\vert \widetilde \Sigma_{ij} \right\vert
&& \geq \left\vert \left( \bH^{d(i,j)} \right)_{ij} \left( \frac{\sigma_{t}}{\mu_{t}} \right)^{2 d_{G} (i,j) } \right\vert
- \left\vert \widetilde \Sigma_{ij} -  \left( \frac{\sigma_{t}}{\mu_{t}} \right)^{2 d(i,j) } \left( \bH^{d(i,j)} \right)_{ij}  \right\vert 
\\
&& \geq \left\vert \left( \bH^{d(i,j)} \right)_{ij} \right\vert \left( \frac{\sigma_{t}}{\mu_{t}} \right)^{2 d_{G} (i,j) } 
 - \frac{ (1 - \kappa_t)  ( D K )^{ d(i,j)}  }{ \kappa_t } \left( \frac{\sigma_{t}}{\mu_{t}} \right)^{2 d(i,j)}.
\eean
Moreover, since $\kappa_t = 1 - DK (\sigma_{t}/\mu_{t})^2$ and 
 $\kappa_{t} \geq 1/2$ for every $t \leq D_{3}$, it follows that
$(1-\kappa_{t})/\kappa_{t} \leq 2 D K (\sigma_{t} / \mu_t)^2$.
Substituting this bound into the last display yields the lower bound
\be
\begin{split}
& \left\vert \left( \bH^{d(i,j)} \right)_{ij} \right\vert \left( \frac{\sigma_{t}}{\mu_{t}} \right)^{2 d_{G} (i,j) }
 - 2 (DK)^{d(i,j) + 1} \left( \frac{\sigma_{t}}{\mu_{t}} \right)^{ 2 d(i,j) + 2 }
 \\
 & \geq \left\vert \left( \bH^{d(i,j)} \right)_{ij} \right\vert \left( \frac{\sigma_{t}}{\mu_{t}} \right)^{2 d_{G} (i,j) }
 - 2 \left\{ ( D K )^{D} \vee 1 \right\} \left( \frac{\sigma_{t}}{\mu_{t}} \right)^{ 2 d(i,j) + 2 }.
\label{eqthm1:cov_lower_bd}
\end{split}
\ee

We now focus on deriving (\ref{eqthm1:cov_gau_bd}).
Let $\bZ_{1} = (  Z_{1,1}, \ldots, Z_{1,D} )$ denote the truncated Gaussian random vector obtained by truncating $\widetilde \bZ$ to the set $S$.
Specifically, its density is given by
\bean
p_{\bZ_1}(\bz) = p_{\widetilde \bZ}(\bz) \left( \int_{S} p_{\widetilde \bZ}(\by) \d \by \right)^{-1}, \quad \bz \in S,
\eean
where $p_{\widetilde \bZ}$ denotes the Lebesgue density of $\widetilde \bZ$.
For $i,j \in [D]$, recall that $\widetilde \Sigma_{ij} = \text{Cov} [\widetilde Z_{i}, \widetilde Z_{j} ]$.
Moreover, by the triangle inequality,
\bean
&& \Big\vert \text{Cov} \big[ Z_{0,i}, Z_{0,j} \big] - \widetilde \Sigma_{ij} \Big\vert
\leq \Big\vert \text{Cov} \big[ Z_{0,i}, Z_{0,j} \big] - 
\text{Cov} \big[  Z_{1,i},  Z_{1,j} \big] \Big\vert
+ \Big\vert \text{Cov} \big[ Z_{1,i}, Z_{1,j} \big] - 
\text{Cov} \big[ \widetilde Z_{i}, \widetilde Z_{j} \big] \Big\vert.
\eean
We proceed by bounding each term on the right-hand side separately; see (\ref{eqthm1:cov_taylor_bd}) and (\ref{eqthm1:cov_trunc_bd}).

We first focus on deriving (\ref{eqthm1:cov_taylor_bd}).
Since $ S = \{ \bz : \|\bx + \sigma_t \bz \|_{\infty} \leq \mu_t \}$ and $\| \bx \|_{\infty} \leq \mu_{t} \gamma \leq 1$, we have
\bean
\| \bz \|_{\infty} \leq \max_{i \in [D]} \left( \frac{\vert \mu_t + x_i \vert }{\sigma_t} \vee \frac{\vert \mu_t - x_i \vert }{\sigma_t} \right) \leq \frac{2 \mu_t }{\sigma_{t}},
\quad \forall \bz \in S.
\eean
Combining this bound with (\ref{eqthm1:rem}), we have $\vert r(\bz) \vert \leq D_{4}$ for all $\bz \in S$, where $D_{4} = 2^{\beta \wedge 3} D_{1}$.
Hence,
\be
e^{-D_{4}}
\leq e^{- (1-\theta) D_{4}}
\leq e^{ ( 1 - \theta ) r(\bz) }
\leq e^{ ( 1 - \theta ) D_{4}}
\leq e^{D_{4}}, \quad \forall \bz \in S, \ \forall \theta \in [0,1].
\label{eqthm1:exp_rem}
\ee

Recall that both $p_{\bZ_0}$ and $p_{\bZ_1}$ are supported on the set $S$.
Let  $\bZ_{\theta} = (Z_{\theta,1},\ldots,Z_{\theta,D}),$ for $\theta \in [0,1]$, denote a random vector interpolating between $\bZ_{0}$ and $\bZ_{1}$, supported on the set $S$, whose Lebesgue density is defined by
\bean
p_{\bZ_{\theta}}(\bz)
= \frac{e^{ ( 1- \theta) r(\bz)} p_{ \bZ_{1} }(\bz) }{\int_{S}  e^{ ( 1- \theta) r(\by)} p_{\bZ_{1}}(\by) \d \by },
\quad \bz \in S.
\eean
For any $i \in [D]$ and $\theta \in [0,1]$, we have
\bean
\bbE \big[ Z_{\theta,i} \big]
 = \frac{\int_{S} z_i e^{ ( 1- \theta) r(\bz)} p_{ \bZ_{1} }(\bz) \d \bz }{\int_{S}  e^{ ( 1- \theta) r(\bz)} p_{\bZ_{1}}(\bz) \d \bz  }
= \frac{ \bbE [ Z_{1,i} e^{ ( 1- \theta) r(\bZ_1) } ] }{ \bbE[ e^{ ( 1- \theta) r( \bZ_{1} ) } ]}
\leq e^{2D_{4}} \bbE \big[ \vert Z_{1,i} \vert \big] < \infty,
\eean
where $\bz = (z_1, \ldots, z_{D})$, and the first inequality holds by (\ref{eqthm1:exp_rem}).
Then, the partial derivative of $\bbE[Z_{\theta,i}]$ with respect to $\theta$ is given by
\bean
\frac{\partial}{\partial {\theta}} \bbE \big[ Z_{\theta,i} \big]
&& = 
\left( \frac{ \bbE[ r(\bZ_{1}) e^{ ( 1- \theta) r(\bZ_1) }  ]  }{\bbE[ e^{ ( 1- \theta) r(\bZ_1) }  ]} \right) \left( \frac{ \bbE[ Z_{1,i} e^{ ( 1- \theta) r(\bZ_1) }  ]  }{\bbE[ e^{ ( 1- \theta) r(\bZ_1) }  ]} \right)
- \frac{ \bbE[ r(\bZ_{1}) Z_{1,i} e^{ ( 1- \theta) r(\bZ_1) }  ]  }{\bbE[ e^{ ( 1- \theta) r(\bZ_1) }  ]}
\\
&& = \bbE \big[ r(\bZ_{\theta}) \big] \bbE \big[ Z_{\theta,i} \big] - \bbE \big[ r(\bZ_{\theta})  Z_{\theta,i} \big]
=
- \text{Cov} \big[Z_{\theta,i}, r(\bZ_{\theta}) \big],
\quad \forall i \in [D].
\eean
Similarly, $\bbE[ Z_{\theta,i} Z_{\theta,j}] < \infty$, and the partial derivative of $\bbE[ Z_{\theta,i} Z_{\theta,j}]$ with respect to $\theta$ is given by
\bean
\frac{\partial}{\partial {\theta}} \bbE \big[ Z_{\theta,i} Z_{\theta,j} \big] = - \text{Cov} \big[ Z_{\theta,i} Z_{\theta, j}, r(\bZ_{\theta}) \big], \quad \forall i,j \in [D].
\eean
The two identities above imply that, for any $i \in [D]$,
\bean
\Big\vert \bbE \big[ Z_{1,i} \big] - \bbE \big[ Z_{0,i} \big] \Big\vert 
\leq \Big\vert \int_{0}^{1} \text{Cov} \big[Z_{\theta,i}, r(\bZ_{\theta}) \big] \d \theta \Big\vert
\leq  \int_{0}^{1}  \Big\vert\text{Cov} \big[Z_{\theta,i}, r(\bZ_{\theta}) \big] \Big\vert \d \theta ,
\eean
and, for any $i,j \in [D]$,
\bean
\Big\vert \bbE \big[ Z_{1,i} Z_{1,j} \big] - \bbE \big[ Z_{0,i} Z_{0,j} \big] \Big\vert 
\leq \Big\vert \int_{0}^{1} \text{Cov} \big[ Z_{\theta,i} Z_{\theta,j} , r(\bZ_{\theta}) \big] \d \theta \Big\vert
\leq  \int_{0}^{1}  \Big\vert \text{Cov} \big[ Z_{\theta,i} Z_{\theta, j} , r(\bZ_{\theta}) \big] \Big\vert \d \theta.
\eean

For $\theta \in [0,1]$, the Cauchy--Schwarz inequality yields that
\bean
\Big\vert \text{Cov} \big[ Z_{\theta,i}, r(\bZ_{\theta}) \big] \Big\vert 
\leq \sqrt{ \text{Var}[Z_{\theta,i}] \text{Var}[r(\bZ_{\theta})]  }
\leq \sqrt{ \bbE[ Z_{\theta,i}^2] \bbE [ \{ r(\bZ_{\theta}) \}^2 ] },
\quad \forall i \in [D],
\eean
and
\bean
\Big\vert\text{Cov} \big[ Z_{\theta,i} Z_{\theta, j}, r(\bZ_{\theta}) \big] \Big\vert 
\leq \sqrt{ \text{Var}[Z_{\theta,i} Z_{\theta,j}] \text{Var}[r(\bZ_{\theta})]  }
\leq \sqrt{ \bbE[ Z_{\theta,i}^2 Z_{\theta,j}^2 ] \bbE [ \{ r(\bZ_{\theta}) \}^2 ] },
\quad \forall i,j \in [D].
\eean
Combining the last four displays, for all $i,j \in [D]$, we have
\bean
&& \Big\vert \text{Cov} \big[ Z_{1,i}, Z_{1,j} \big] - \text{Cov} \big[ Z_{0,i}, Z_{0,j} \big] \Big\vert
\\
&& \leq \Big\vert \bbE \big[ Z_{1,i} Z_{1,j} \big] - \bbE \big[ Z_{0,i} Z_{0,j} \big] \Big\vert 
+ \Big\vert \bbE \big[ Z_{1,i} \big] \bbE \big[ Z_{1,j} \big] - \bbE \big[ Z_{0,i} \big]\bbE \big[ Z_{0,j} \big] \Big\vert
\\
&& \leq \Big\vert \bbE \big[ Z_{1,i} Z_{1,j} \big] - \bbE \big[ Z_{0,i} Z_{0,j} \big] \Big\vert 
+ \Big\vert \bbE \big[ Z_{1,i} \big] \Big\vert 
\Big\vert \bbE \big[ Z_{1,j} \big]  - \bbE \big[ Z_{0,j} \big]  \Big\vert
+ \Big\vert \bbE \big[ Z_{0,j} \big] \Big\vert 
\Big\vert \bbE \big[ Z_{1,i} \big]  - \bbE \big[ Z_{0,i} \big]  \Big\vert
\\
&& \leq  \int_{0}^{1} 
\sqrt{ \bbE \big[ Z_{\theta,i}^2 Z_{\theta,j}^2 \big] \bbE \big[ \{ r(\bZ_{\theta}) \}^2 \big] }
+ \Big\vert \bbE \big[ Z_{1,i} \big] \Big\vert \sqrt{ \bbE \big[ Z_{\theta,j}^2 \big] \bbE \big[ \{ r(\bZ_{\theta}) \}^2 \big] }
+ \Big\vert \bbE \big[ Z_{0,j} \big] \Big\vert \sqrt{ \bbE \big[ Z_{\theta,i}^2 \big] \bbE \big[ \{ r(\bZ_{\theta}) \}^2 \big] }
\d \theta
\\
&& =  \int_{0}^{1} \sqrt{ \bbE \big[ \{ r(\bZ_{\theta}) \}^2 \big] }
\left( \sqrt{ \bbE \big[  Z_{\theta,i}^2 Z_{\theta, j}^2 \big] }
+  \Big\vert \bbE \big[ Z_{1,i} \big] \Big\vert \sqrt{ \bbE \big[  Z_{\theta,j}^2  \big] }
+ \Big\vert \bbE \big[ Z_{0,j} \big] \Big\vert \sqrt{ \bbE \big[  Z_{\theta,i}^2  \big] }
\right) \d \theta,
\\
&& \leq \int_{0}^{1} \sqrt{ \bbE \big[ \{ r(\bZ_{\theta}) \}^2 \big] }
\left( \sqrt{ \bbE \big[ \| \bZ_{\theta} \|_{\infty}^4 \big] }
+   \bbE \big[ \| \bZ_{1} \|_{\infty} \big] \sqrt{ \bbE \big[ \| \bZ_{\theta} \|_{\infty}^2   \big] }
+ \bbE \big[ \| \bZ_{0} \|_{\infty} \big] \sqrt{ \bbE \big[ \| \bZ_{\theta} \|_{\infty }^2 \big] }
\right) \d \theta.
\eean
By the definition of $\bZ_{\theta}$ and (\ref{eqthm1:exp_rem}),
the moments of $\| \bZ_{\theta} \|_{\infty}$ satisfy
\bean
\bbE \left[ \|\bZ_{\theta}\|_{\infty}^{k} \right]
= \frac{ \bbE \left[ \| \bZ_{1} \|_{\infty}^k e^{ ( 1- \theta) r(\bZ_1) } \right] }{ \bbE[ e^{ ( 1- \theta) r( \bZ_{1} ) } ]}
\leq e^{2D_{4}}  \bbE\left[ \| \bZ_{1} \|_{\infty}^{k} \right],
\quad \forall k > 0, \forall \theta \in [0,1].
\eean
Moreover, by (\ref{eqthm1:rem}),
\bean
 \bbE \left[ \left\{ r(\bZ_{\theta}) \right\}^2 \right]
 \leq D_{1}^2 \left( \frac{\sigma_t}{\mu_t} \right)^{(2\beta) \wedge 6} \bbE \left[ \| \bZ_{\theta} \|_{\infty}^{ (2\beta) \wedge 6} \right]
 \leq D_{1}^2 e^{2D_{4}} \left( \frac{\sigma_t}{\mu_t} \right)^{(2\beta) \wedge 6} \bbE \left[ \| \bZ_{1} \|_{\infty}^{ (2\beta) \wedge 6} \right],
 \quad \forall \theta \in [0,1].
\eean
Combining the last three displays, for all $i,j \in [D]$, we have
\be
\begin{split}
& \Big\vert \text{Cov}[Z_{1,i}, Z_{1,j}] - \text{Cov}[Z_{0,i}, Z_{0,j}] \Big\vert
\\
& \leq D_{1} e^{D_{4}} \left( \frac{\sigma_t}{\mu_t} \right)^{\beta \wedge 3}  
\sqrt{ \bbE \big[ \| \bZ_{1} \|_{\infty}^{ (2\beta) \wedge 6} \big] }
\\
& \qquad \times \int_{0}^{1} 
e^{D_{4}} \sqrt{ \bbE \big[  \left\| \bZ_{1} \right\|_{\infty}^{4}   \big] }
+ e^{D_{4}} \bbE \big[ \left\| \bZ_{1} \right\|_{\infty} \big] \sqrt{ \bbE \big[ \left\| \bZ_{1} \right\|_{\infty}^2  \big] }
+ e^{3 D_{4} } \bbE \big[ \left\| \bZ_{1} \right\|_{\infty} \big] \sqrt{ \bbE \big[ \left\| \bZ_{1} \right\|_{\infty}^2 \big] }
\d \theta
\\
& \leq D_{5}  \left( \frac{\sigma_t}{\mu_t} \right)^{\beta \wedge 3}
\sqrt{ \bbE \big[ \left\| \bZ_{1} \right\|_{\infty}^{(2\beta) \wedge 6} \big] }
\left(
\sqrt{ \bbE \big[ \left\| \bZ_{1} \right\|_{\infty}^{4} \big] }
+  \bbE \big[ \left\| \bZ_{1} \right\|_{\infty} \big] 
\sqrt{ \bbE \big[ \left\| \bZ_{1} \right\|_{\infty}^2 \big]  }
\right),
\label{eqthm1:cov_err_bd}
\end{split}
\ee
where $D_{5} = D_{5} ( D_{1}, D_{4}) > 0$. 

To derive (\ref{eqthm1:cov_taylor_bd}), we next bound the moments of $\| \bZ_{1}\|_{\infty}$; see (\ref{eqthm1:z1_mom_bd}).
By the definition of $\bZ_{1}$, for any $k > 0$,
\be
\bbE \left[ \left\| \bZ_{1} \right\|_{\infty}^{k} \right]
= \frac{ \int_{S} \| \bz \|_{\infty}^{k} p_{\widetilde \bZ} (\bz) \d \bz }{\bbP ( \widetilde \bZ \in S )}
\leq \frac{ \bbE [ \| \widetilde \bZ \|_{\infty}^{k} ] }{\bbP ( \widetilde \bZ \in S )}.
\label{eqthm1:mean_trunc}
\ee
We then derive an upper bound on $\bbP( \widetilde \bZ \notin S ) = 1 - \bbP ( \widetilde \bZ \in S )$ and on $\bbE[ \| \widetilde \bZ \|_{\infty}^{k}]$; see (\ref{eqthm1:tail_bd}) and (\ref{eqthm1:ztilde_mom_bd}).

Since $S = \{ \bz : \| \bx +  \sigma_{t} \bz \|_{\infty} \leq \mu_t \}$ and $\widetilde \bZ \sim \cN (\widetilde \bmu, \widetilde \bSigma)$, 
a simple calculation yields that
\bean
&& \bbP \left( \widetilde \bZ \notin S \right)
= \bbP \left( \big\| \bx + \sigma_t \widetilde \bZ \big\|_{\infty} \geq \mu_t \right)
\leq \sum_{i=1}^{D} \bbP \left( \big\vert x_i + \sigma_t \widetilde Z_i \big\vert \geq \mu_{t} \right)
\\
&& \leq \sum_{i=1}^{D} \left\{ \bbP \left( \widetilde Z_i \geq \frac{\mu_{t} - x_i}{\sigma_{t}}  \right) + 
\bbP \left( \widetilde Z_i \leq \frac{-\mu_{t} - x_i}{\sigma_{t}}  \right) 
\right\}
\\
&& = \sum_{i=1}^{D} \left[  \bbP \left( Z \geq \big( \widetilde \Sigma_{ii} \big)^{-\frac{1}{2}} \left\{ \frac{\mu_{t} - x_i}{\sigma_{t} } - \widetilde \mu_i \right\} \right) 
+ \bbP \left( Z \leq \big( \widetilde \Sigma_{ii} \big)^{-\frac{1}{2}} \left\{ \frac{- \mu_{t} - x_i}{\sigma_{t} } - \widetilde \mu_i \right\} \right) 
\right],
\eean
where $Z$ denotes the one-dimensional standard normal random variable.

Note that $\tnorm{ \widetilde \bSigma }_2$ equals the inverse of the smallest eigenvalue of $\widetilde \bSigma^{-1}$.
Combining  (\ref{eqthm1:matrix_pd}) with the definition of $\widetilde \bSigma$,
$\tnorm{ \widetilde \bSigma}_2 \leq 1/ \kappa_{t}$ and $ \kappa_{t} \geq 1/2 $ for every $t \leq D_{3}$.
Moreover, $\| \nabla u_0 (\bx / \mu_t) \|_2 \leq \sqrt{D} \| \nabla u_0 (\bx / \mu_t ) \|_{\infty} \leq K \sqrt{D} $ because $u_0 \in \cH^{\beta,K}_{D}([-1,1]^{D})$.
Hence, for every $i \in [D]$ and $t \leq D_{3}$,
\bean
\vert \widetilde \mu_i \vert \leq \| \widetilde \bmu \|_2
= \left( \frac{\sigma_{t}}{\mu_t} \right)
\left\| \widetilde \bSigma \nabla u_0 \left( \frac{\bx}{\mu_{t}} \right)  \right\|_{2}
\leq 
\left( \frac{\sigma_{t}}{\mu_t} \right)
\Tnorm{\widetilde \bSigma}_{2}
\left\| \nabla u_0 \left( \frac{\bx}{\mu_{t}} \right) \right\|_{2}
 \leq  \frac{ 2  K \sqrt{ D }  \sigma_{t}  }{\mu_t} 
\eean
and
\bean
\widetilde \Sigma_{ii} = \bfe_{i}^{\top} \widetilde \bSigma \bfe_{i}
\leq \Tnorm{\widetilde \bSigma}_{2} \leq 2.
\eean
Since $\| \bx \|_{\infty} \leq \mu_t \gamma $, we have that
\bean
 \frac{\mu_{t} - x_i}{\sigma_{t} } - \widetilde \mu_i 
\geq \frac{\mu_t ( 1-\gamma) }{\sigma_{t}} - \frac{2 K \sqrt{ D }  \sigma_t}{\mu_t}
\quad \text{and} \quad
 \frac{-\mu_{t} - x_i}{\sigma_{t} } - \widetilde \mu_i 
\leq - \frac{\mu_t (1-\gamma) }{\sigma_{t}} + \frac{2 K \sqrt{ D }  \sigma_t}{\mu_t}.
\eean
Let $D_{6} = D_{6} ( D_{3}, D, K, \gamma ) $ be a positive constant such that $D_{6} \leq D_{3}$ and $ 2 K \sqrt{ D } \sigma_{t} / \mu_t \leq \mu_{t} (1-\gamma) / (2 \sigma_{t}) $ for every $t \leq D_{6}$.
Then, combining the last three displays yields
\bean
 \left( \widetilde \Sigma_{ii} \right)^{-\frac{1}{2}} \left\{ \frac{\mu_{t} - x_i}{ \sigma_{t} } - \widetilde \mu_i \right\}
\geq \frac{\mu_{t} ( 1- \gamma) }{ 4 \sigma_{t}} > 0
\quad \text{and} \quad
 \left( \widetilde \Sigma_{ii} \right)^{-\frac{1}{2}} \left\{ \frac{-\mu_{t} - x_i}{\sigma_{t} } - \widetilde \mu_i \right\} 
\leq  -\frac{\mu_{t} ( 1 - \gamma) }{ 4 \sigma_{t}} < 0
\eean
for every $t \leq D_{6}$.
Hence,
\be
\begin{split}
\bbP \left( \widetilde \bZ \notin S \right) 
& \leq \sum_{i=1}^{D} \left[  \bbP \left( Z \geq \big( \widetilde \Sigma_{ii} \big)^{-\frac{1}{2}} \left\{ \frac{\mu_{t} - x_i}{\sigma_{t} } - \widetilde \mu_i \right\} \right) 
+ \bbP \left( Z \leq \big( \widetilde \Sigma_{ii} \big)^{-\frac{1}{2}} \left\{ \frac{- \mu_{t} - x_i}{\sigma_{t} } - \widetilde \mu_i \right\} \right) 
\right]
\\
& \leq D \bbP \left( \big\vert  Z  \big\vert \geq \frac{\mu_{t} (1-\gamma) }{ 4 \sigma_{t} } \right)
\leq 2 D \exp \left( -\frac{\mu_{t}^2 (1-\gamma)^2 }{32 \sigma_{t}^2} \right),
\label{eqthm1:tail_bd}
\end{split}
\ee
where the last inequality holds by the tail probability of the standard normal distribution.

Combining (\ref{eqthm1:mean_trunc}) with (\ref{eqthm1:tail_bd}), for all $t \leq D_{6}$ and $k > 0$,
\bean
\bbE \left[ \big\| \bZ_{1} \big\|_{\infty}^{k} \right]
\leq \frac{ \bbE [ \| \widetilde \bZ \|_{\infty}^{k} ] }{1 - \bbP ( \widetilde \bZ \notin S )}
\leq \frac{ \bbE [ \| \widetilde \bZ \|_{\infty}^{k} ] }{1 - 2D \exp ( - \mu_{t}^2 (1-\gamma)^2 / ( 32 \sigma_{t}^2) )}.
\eean
For any $k \geq 1$ and $m \in \bbN$, note that $ \vert \sum_{i=1}^{m} a_i \vert ^{k} \leq m^{k-1}  \sum_{i=1}^{m} \vert a_i \vert ^{k}$, $a_1,\ldots,a_{m} \in \bbR$.
Since $\vert \widetilde \mu_i \vert \leq 2 K \sqrt{ D }  (\sigma_{t} / \mu_t)$ and $\widetilde \Sigma_{ii} \leq 2$ for all $i \in [D]$ and $ t \leq D_{6}$, we have that
\bean
&& \bbE \left[ \big\vert \widetilde Z_i \big\vert^{k} \right]
= \bbE \left[ \big\vert ( \widetilde \Sigma_{ii} )^{\frac{1}{2}} Z + \widetilde \mu_i \big\vert^{k} \right]
\leq 2^{k-1} \left\{ (\widetilde \Sigma_{ii})^{\frac{k}{2}} \bbE \left[ \big\vert Z \big\vert^{k} \right] + \big\vert \widetilde \mu_i \big\vert^{k} \right\}
\\
&& \leq 2^{\frac{3k-2}{2}} \bbE \left[ \big\vert Z \big\vert^{k} \right]
+ 2^{2k-1} D^{\frac{k}{2}} K^k \left( \frac{\sigma_{t}}{\mu_t} \right)^{k}
\leq 2^{\frac{3k-2}{2}} \bbE \left[ \big\vert Z \big\vert^{k} \right] + 2^{2k-1} K^{\frac{k}{2}}
\defeq \widetilde  D_{k},
\quad \forall k \geq 1,
\eean
where the last inequality holds because $\sigma_{t} /\mu_{t} \leq 1/\sqrt{DK}$.

Let $D_{7} = D_{7} (D_{6}, D, \gamma )$ be a positive constant such that $D_{7} \leq D_{6}$ and $1 - 2D \exp ( - \mu_{t}^2 ( 1- \gamma )^2 / (32 \sigma_{t}^2) ) \geq 1/2$ for every $t \leq D_{7}$.
Then, for every $t \leq D_{7}$ and $k \geq 1$,
\be
\bbE \left[ \big\| \widetilde \bZ \big\|_{\infty}^{k} \right]
\leq \bbE \left[ \left( \sum_{i=1}^{D} \left\vert \widetilde Z_i \right\vert \right)^{k} \right]
\leq D^{k-1} \sum_{i=1}^{D} \bbE \left[ \big\vert \widetilde Z_i \big\vert^k \right]
\leq D^{k} \widetilde   D_{k}
\label{eqthm1:ztilde_mom_bd}
\ee
and hence
\be
\bbE \left[ \big\| \bZ_{1} \big\|_{\infty}^{k} \right] 
\leq \frac{ \bbE [ \| \widetilde \bZ \|_{\infty}^{k} ] }{ 1 - 2D \exp ( - \mu_{t}^2 (1-\gamma)^2 / (32 \sigma_{t}^2) )}
\leq 2 \bbE \left[ \big\| \widetilde \bZ \big\|_{\infty}^{k} \right]
\leq 2 D^{k} \widetilde   D_{k}.
\label{eqthm1:z1_mom_bd}
\ee
Therefore, combining with (\ref{eqthm1:cov_err_bd}), we have
\be
\Big\vert \text{Cov} \big[ Z_{1,i}, Z_{1,j} \big] - \text{Cov} \big[ Z_{0,i}, Z_{0,j} \big] \Big\vert
\leq D_{8} \left( \frac{\sigma_{t}}{\mu_{t}} \right)^{\beta \wedge 3},
\quad \forall t \leq D_{7}, \forall i,j \in [D],
\label{eqthm1:cov_taylor_bd}
\ee
where $D_{8} = D_{8} ( D_{5}, \beta, K ) > 0$. 

We now focus on deriving (\ref{eqthm1:cov_trunc_bd}).
Recall that $\widetilde \Sigma_{ij} = \text{Cov}[\widetilde Z_i, \widetilde Z_j]$ for all $i,j \in [D]$.
By the definition of $\bZ_{1}$,
for all $i \in [D]$, we have
\bean
\Big\vert \bbE \big[ Z_{1,i} \big] - \bbE \big[ \widetilde Z_{i} \big] \Big\vert
&& = \left\vert \frac{ { \int_{S} z_i p_{\widetilde \bZ}(\bz) \d \bz } }{ \bbP (\widetilde \bZ \in S) } 
- \bbE \big[ \widetilde Z_i \big]  \right\vert
\\
&& \leq \left\vert \frac{ \int_{S} z_i p_{\widetilde \bZ}(\bz) \d \bz }{ \bbP (\widetilde \bZ \in S) } 
- \frac{ \bbE [ \widetilde Z_i ] }{ \bbP (\widetilde \bZ \in S) }  \right\vert
+ \left\vert  \frac{ \bbE [ \widetilde Z_i ] }{ \bbP (\widetilde \bZ \in S) } - \bbE \big[ \widetilde Z_i \big]  \right\vert
\\
&& = \frac{ \left\vert \int_{\bbR^{D} \backslash S } z_i p_{\widetilde \bZ}(\bz) \d \bz  \right\vert }{\bbP ( \widetilde \bZ \in S) } + \Big\vert \bbE \big[ \widetilde Z_i \big] \Big\vert \frac{ \bbP (\widetilde \bZ \notin S) }{ \bbP (\widetilde \bZ \in S) }.
\eean
By the Cauchy--Schwarz inequality, 
\bean
\left\vert \int_{\bbR^{D} \backslash S } z_i p_{\widetilde \bZ}(\bz) \d \bz \right\vert
\leq 
 \int_{\bbR^{D} \backslash S } \vert z_i \vert p_{\widetilde \bZ}(\bz) \d \bz
 \leq \sqrt{ \bbE \big[ \widetilde Z_i^2 \big] \bbP \big(\widetilde \bZ \notin S \big) },
 \quad \forall i \in [D].
\eean
Combining the last two displays, we have
\bean
\Big\vert \bbE \big[ Z_{1,i} \big] - \bbE \big[ \widetilde Z_{i} \big] \Big\vert
&& \leq \frac{ \sqrt{ \bbE [ \widetilde Z_i^2 ] \bbP (\widetilde \bZ \notin S) }  }{ \bbP (\widetilde \bZ \in S) }
+  \frac{ \bbE[ \vert \widetilde Z_i \vert ]  \bbP (\widetilde \bZ \notin S) }{ \bbP (\widetilde \bZ \in S) }
\\
&& \leq  \frac{ \sqrt{  \bbP (\widetilde \bZ \notin S)  } }{  \bbP (\widetilde \bZ \in S)  } 
\left\{ \sqrt{\bbE \left[ \big\|  \widetilde \bZ \big\|_{\infty}^2 \right] } + \bbE \left[ \big\| \widetilde \bZ \big\|_{\infty} \right]  \right\}
\\
&& \leq  \frac{ \sqrt{  \bbP (\widetilde \bZ \notin S)  } }{  \bbP (\widetilde \bZ \in S)  }  \left( D \sqrt{ \widetilde D_{2} } + D \widetilde D_{1} \right), \quad  \forall i \in [D],
\eean
where the last inequality holds by (\ref{eqthm1:ztilde_mom_bd}).
Similarly, we have
\bean
\Big\vert \bbE \big[ Z_{1,i} Z_{1,j} \big] - \bbE \big[ \widetilde Z_{i} \widetilde Z_{j} \big] \Big\vert
&&  \leq \frac{ \sqrt{ \bbE [ \widetilde Z_i^2 \widetilde Z_{j}^2 ] \bbP (\widetilde \bZ \notin S) }  }{ \bbP (\widetilde \bZ \in S) }
+  \frac{ \bbE[ \vert \widetilde Z_i \widetilde Z_{j} \vert ]  \bbP (\widetilde \bZ \notin S) }{ \bbP (\widetilde \bZ \in S) }
\\
&& \leq  \frac{ \sqrt{  \bbP (\widetilde \bZ \notin S)  } }{  \bbP (\widetilde \bZ \in S)  } \left\{ \sqrt{\bbE \left[ \big\|  \widetilde \bZ \big\|_{\infty}^4 \right] } + \bbE \left[ \big\| \widetilde \bZ \big\|_{\infty}^2 \right]  \right\}
\\
&& \leq  \frac{ \sqrt{  \bbP (\widetilde \bZ \notin S)  } }{  \bbP (\widetilde \bZ \in S)  } 
\left( D^2 \sqrt{ \widetilde D_{4} } + D^2 \widetilde D_{2} \right), \quad \forall i,j \in [D].
\eean
Combining the last two displays yields that, for all $i,j \in [D]$,
\bean
&& \Big\vert \text{Cov} \big[ Z_{1,i}, Z_{1,j} \big]  - \widetilde \Sigma_{ij} \Big\vert
= \Big\vert \text{Cov} \big[ Z_{1,i}, Z_{1,j} \big]   - \text{Cov} \big[ \widetilde Z_{i}, \widetilde Z_{j} \big] \Big\vert
\\
&& \leq \Big\vert \bbE \big[ Z_{1,i} Z_{1,j} \big] - \bbE \big[ \widetilde Z_{i} \widetilde Z_{j} \big] \Big\vert
+  \Big\vert  \bbE \big[ Z_{1,i} \big]  \Big\vert \Big\vert \bbE \big[ Z_{1,j} \big]  - \bbE \big[ \widetilde Z_{j} \big]  \Big\vert
+ \Big\vert \bbE \big[ \widetilde Z_{j} \big] \Big\vert  \Big\vert \bbE \big[ Z_{1,i} \big]  - \bbE \big[ \widetilde Z_{i} \big]  \Big\vert
\\
&& \leq  \frac{ \sqrt{  \bbP (\widetilde \bZ \notin S)  } }{  \bbP (\widetilde \bZ \in S)  }
\left\{ 
 D^2 \sqrt{ \widetilde D_{4} } + D^2 \widetilde D_{2} 
+
\left( \bbE \big[ \| \bZ_1 \|_{\infty} \big] +  \bbE \big[ \big\| \widetilde \bZ  \big\|_{\infty} \big] \right)
\left(  D \sqrt{ \widetilde D_{2} } + D \widetilde D_{1}  \right)
\right\}
\\
&& \leq \frac{ D_{9} \sqrt{  \bbP (\widetilde \bZ \notin S)  } }{  \bbP (\widetilde \bZ \in S)  },
\eean
where $D_{9} = D_{9} (D, K) > 0$ and the last inequality holds by  (\ref{eqthm1:ztilde_mom_bd}) and
(\ref{eqthm1:z1_mom_bd}).

Recall that $ 1- 2D \exp (- \mu_{t}^2 (1-\gamma)^2 / ( 32\sigma_{t}^2 )) \geq 1/2 $ for every $t \leq D_{7}$.
Thus, by (\ref{eqthm1:tail_bd}),
\bean
\frac{ \sqrt{  \bbP (\widetilde \bZ \notin S)  } }{  \bbP (\widetilde \bZ \in S)  }
=
\frac{ \sqrt{  \bbP (\widetilde \bZ \notin S)  } }{ 1 -  \bbP (\widetilde \bZ \notin S)  }
\leq
 2 \sqrt{2D} \exp \left( -\frac{\mu_{t}^2 (1-\gamma)^2 }{64 \sigma_{t}^2} \right),
 \quad \forall t \leq D_{7}.
\eean
Let $D_{10} = D_{10} ( D_{7}, \beta, \gamma )$ be a positive constant such that $D_{10} \leq D_{7}$ and both $\exp(-\mu_{t}^2 ( 1- \gamma )^2 / ( 64 \sigma_{t}^2)) \leq (\sigma_{t}/\mu_{t})^{\beta \wedge 3}$ and $\sigma_{t}/\mu_{t} \leq 1$ hold for every $t \leq D_{10}$.
Then, combining the last two displays yields that
\be
\left\vert \text{Cov} \left[ Z_{1,i}, Z_{1,j} \right]  - \widetilde \Sigma_{ij} \right\vert
\leq 2 D_{9} \sqrt{2D}  \left(  \frac{ \sigma_{t} }{ \mu_{t} } \right)^{\beta \wedge 3}, 
\quad \forall t \leq D_{10}, \forall i,j \in [D].
\label{eqthm1:cov_trunc_bd}
\ee
Moreover, for every $t \leq D_{10}$, 
\be
\begin{split}
\Big\vert  \text{Cov} \big[ Z_{0,i}, Z_{0,j} \big] - \widetilde \Sigma_{ij} \Big\vert 
& \leq
\Big\vert \text{Cov} \big[ Z_{0,i}, Z_{0,j} \big] - \text{Cov} \big[Z_{1,i}, Z_{1,j} \big] \Big\vert
 + \left\vert \text{Cov} \big[ Z_{1,i}, Z_{1,j} \big]  - \widetilde \Sigma_{ij} \right\vert
 \\
 & \leq \left( D_{8} + 2 D_{9} \sqrt{2D} \right) \left( \frac{\sigma_{t}}{\mu_{t}} \right)^{\beta \wedge 3}, 
 \quad \forall i,j \in [D],
 \label{eqthm1:cov_gau_bd}
 \end{split}
\ee
where the second inequality holds by  (\ref{eqthm1:cov_taylor_bd}) and (\ref{eqthm1:cov_trunc_bd}).
Combining (\ref{eqthm1:cov_bd}) with the last  display, we have
\bean
\Big\vert \text{Cov} \big[ Z_{0,i}, Z_{0,j} \big]  \Big\vert
&& \leq \Big\vert  \text{Cov} \big[ Z_{0,i}, Z_{0,j} \big] - \widetilde \Sigma_{ij} \Big\vert 
 + \Big\vert \widetilde \Sigma_{ij} \Big\vert
 \\
 && \leq  D_{11} \left( \frac{\sigma_{t}}{\mu_{t}} \right)^{ ( 2 d(i,j) ) \wedge \beta \wedge 3},
 \quad \forall t \leq D_{10}, \forall i, j \in [D] \text{ with } d(i,j) < \infty,
\eean
where $D_{11} = D_{8} + 2 D_{9} \sqrt{2D} +  2 \{  (DK)^{D} \vee 1 \}$.
This bound also holds for $i,j \in [D]$ with $d(i,j) = \infty$, since $\widetilde \Sigma_{ij} = 0$ in that case.
Recall that $1 \leq \mu_{t}^{-1} \leq 2$ for $t \leq D_{10}$.
Hence, for all $i,j \in [D]$ and $t \leq D_{10}$,
\bean
\Big\vert {\rm{Cov}} \big[ X_{0,i}, X_{0,j} \mid \bX_{t} = \bx \big] \Big\vert
=  \left( \frac{\sigma_t}{\mu_t} \right)^2 \Big\vert \text{Cov} \big[ Z_{0,i}, Z_{0,j} \big] \Big\vert
\leq D_{12} \sigma_{t}^{( 2 d(i,j) +2) \wedge (\beta+2) \wedge 5},
\eean
where $D_{12} = 2^{5/2} D_{11}$.
Similarly, combining with (\ref{eqthm1:cov_lower_bd}), we have
\bean
&& \Big\vert \text{Cov} \big[ Z_{0,i}, Z_{0,j} \big]  \Big\vert
\\
&& \geq \Big\vert \widetilde \Sigma_{ij} \Big\vert - \Big\vert \text{Cov} \big[ Z_{0,i}, Z_{0,j} \big] - \widetilde \Sigma_{ij} \Big\vert
\\
 && \geq  \Big\vert \widetilde \Sigma_{ij} \Big\vert
 - \Big\vert \text{Cov} \big[ Z_{0,i}, Z_{0,j} \big] - \text{Cov} \big[ Z_{1,i}, Z_{1,j} \big] \Big\vert
 - \Big\vert \text{Cov} \big[ Z_{1,i}, Z_{1,j} \big] - \widetilde \Sigma_{ij} \Big\vert 
 \\
 && \geq \Big\vert \left( \bH^{d(i,j)} \right)_{ij} \Big\vert \left( \frac{\sigma_{t}}{\mu_{t}} \right)^{2 d_{G} (i,j) }
 -  D_{11} \left( \frac{\sigma_{t}}{\mu_{t}} \right)^{(  2 d(i,j) + 2 ) \wedge \beta \wedge 3}
, \quad \forall t \leq D_{10}, \forall i, j \in [D] \text{ with } d(i,j) < \infty.
\eean
Consequently, for all $i,j \in [D]$ with $d_{G} (i,j) < \infty$ and $t \leq D_{10}$, 
\bean
\Big\vert {\rm{Cov}} \big[ X_{0,i}, X_{0,j} \mid \bX_{t} = \bx \big] \Big\vert
&&  =  \left( \frac{\sigma_t}{\mu_t} \right)^2 \Big\vert \text{Cov} \big[ Z_{0,i}, Z_{0,j} \big] \Big\vert
\\
&& \geq \Big\vert \left( \bH^{d(i,j)} \right)_{ij} \Big\vert \sigma_{t}^{2 d_{G} (i,j) +2 }
 -  D_{12} \sigma_{t}^{ ( 2 d(i,j) + 4 ) \wedge ( \beta + 2) \wedge 5} .
\eean
The first and second assertions are followed by re-defining the constants.

\hfill\qedsymbol

\subsection{Proof of Corollary~\ref{cor1}}
\begin{proof}
By Tweedie's formula \eqref{eq:tweedie}, $H_{ij}(\bx,t) = \sigma_{t}^{-4} \mu_{t}^2 \text{Cov}[X_{0,i}, X_{0,j} \mid \bX_{t} = \bx]$ for every $i \neq j \in [D]$. 
Since $\mu_{t}^2 \leq 1$, the assertion follows by Theorem~\ref{thm1}.
\end{proof}

\subsection{Proof of Theorem~\ref{thm2}}
\label{ssec:proof_thm2}

\begin{proof}

Fix $t \in [\underline{T}, {\widetilde T} ]$.
From the definition of $\widehat G_{t,\tau}$ (see \eqref{eq:estimator}), we have
\be
\begin{split}
\bbP \Big( \widehat G_{t,\tau} \neq G_{0} \Big)
&\leq \sum_{ \substack {i,j \in [D]  \\ d(i,j) = 1} } \bbP \Big(  (i,j) \notin \widehat E_{t,\tau}  \Big) 
+ \sum_{ \substack {i,j \in [D]  \\ d(i,j) \geq 2} } \bbP \Big(  (i,j) \in \widehat E_{t,\tau} \Big) 
\\
& = \sum_{ \substack {i,j \in [D]  \\ d(i,j) = 1 } } \bbP  \Big(  \bbE_{n} \big[  \big\vert \widehat H_{ij} \big( \bU_{t}, t \big)  \big\vert \big]  \leq \tau \Big) 
+ \sum_{  \substack {i,j \in [D]  \\ d(i,j) \geq 2 } } \bbP  \Big(  \bbE_{n} \big[  \big\vert \widehat H_{ij} \big( \bU_{t}, t \big)  \big\vert \big]  > \tau \Big).
\label{eqthm2:t1_t2}
\end{split}
\ee
Since $\vert a-c \vert \geq c - a \geq c - b$ for $a,b,c \in \bbR$ with $a \leq b$, it follows that
\bean
\bbP  \left( \bbE_{n} \big[  \big\vert \widehat H_{ij} \big( \bU_{t} , t \big)  \big\vert \big] \leq \tau \right)
\leq  \bbP \left( \Big\vert \bbE_{n} \big[  \big\vert \widehat H_{ij} \big( \bU_{t} , t \big)  \big\vert \big] - \bbE \big[  \big\vert  H_{ij} \big(  \bU_{t} , t \big)  \big\vert \big] \Big\vert
\geq \bbE \big[  \big\vert  H_{ij} \big( \bU_{t}, t \big)  \big\vert \big]  - \tau
\right).
\eean
Similarly, since $\vert c-a \vert \geq a-c > b - c$ for $a,b,c \in \bbR$ with $a > b$, it follows that
\bean
\bbP  \left( \bbE_{n} \big[  \big\vert \widehat H_{ij} \big( \bU_{t}, t \big)  \big\vert \big] > \tau \right)
\leq  \bbP \left( \Big\vert \bbE_{n} \big[  \big\vert \widehat H_{ij} \big( \bU_{t}, t \big)  \big\vert \big] - \bbE \big[  \big\vert  H_{ij} \big(  \bU_{t}, t \big)  \big\vert \big] \Big\vert
> \tau - \bbE \big[  \big\vert  H_{ij} \big(  \bU_{t}, t \big)  \big\vert \big]
\right).
\eean
By Corollary~\ref{cor1}, we have $\bbE [ \vert H_{ij} (\bU_{t},t) \vert ] > \tau$ if $d(i,j) = 1$ and $\bbE [ \vert H_{ij} (\bU_{t},t) \vert ] < \tau$ if $d(i,j) \geq 2$. 
Specifically,
\bean
&& \bbE \big[  \big\vert  H_{ij} \big(  \bU_{t}, t \big)  \big\vert \big] \geq C_{\gamma} - C_{2} \sigma_{t}^{{\tilde \beta}} > \tau,
\qquad \forall i,j \in [D] \text{ with } d(i,j) = 1,
\\
&&  \bbE \big[  \big\vert  H_{ij} \big(  \bU_{t}, t \big)  \big\vert \big] \leq C_{2} \sigma_{t}^{{\tilde \beta}} < \tau,
\qquad  \forall i,j \in [D] \text{ with } d(i,j) \geq 2.
\eean

Combining these inequalities with Markov's inequality, we have
\bean
 \bbP  \left( \bbE_{n} \big[  \big\vert \widehat H_{ij} \big( \bU_{t} , t \big)  \big\vert \big] \leq \tau \right)
 \leq  \frac{C_{ij}}{ \bbE [ \vert H_{ij} (\bU_{t},t) \vert ] - \tau }
 \leq \frac{C_{ij}}{ C_{\gamma} - C_{2} \sigma_{t}^{{\tilde \beta}} - \tau },
 \eean
for all $i,j \in [D]$ with $d(i,j) = 1$, and
 \bean
 \bbP  \left( \bbE_{n} \big[  \big\vert \widehat H_{ij} \big( \bU_{t} , t \big)  \big\vert \big] > \tau \right)
 \leq \frac{C_{ij}}{ \tau  - \bbE [ \vert H_{ij} (\bU_{t},t) \vert ]  }
 \leq  \frac{C_{ij}}{ \tau  - C_{2} \sigma_{t}^{{\tilde \beta}}  },
\eean
for all $i,j \in [D]$ with $d(i,j) \geq 2$,
where 
\bean
C_{ij} = \bbE \big[  \Big\vert \bbE_{n} \Big[  \big\vert \widehat H_{ij} \big( \bU_{t}, t \big)  \big\vert \Big] - \bbE \Big[  \big\vert  H_{ij} \big(  \bU_{t}, t \big)  \big\vert \Big] \Big\vert  \big].
\eean
Combining this with (\ref{eqthm2:t1_t2}), we have
\be
\bbP \left( \widehat G_{t,\tau} \neq G_{0} \right)
\leq
\sum_{ \substack {i,j \in [D]  \\ d(i,j) = 1 } } \frac{C_{ij}}{ C_{\gamma} - C_{2} \sigma_{t}^{{\tilde \beta}} - \tau }
+ \sum_{  \substack {i,j \in [D]  \\ d(i,j) \geq 2 } } \frac{C_{ij}}{ \tau  - C_{2} \sigma_{t}^{{\tilde \beta}}  }.
\label{eqthm2:markov}
\ee
By the triangle inequality, 
\bean
C_{ij}
\leq 
 \Big\vert  \bbE \big[  \big\vert H_{ij} \big( \bU_{t}, t \big)  \big\vert  \big] -  \bbE \big[   \big\vert \widetilde H_{ij} \big( \bU_{t}, t \big)  \big\vert \big]  \Big\vert 
 +
 \bbE \Big[  \Big\vert \bbE \big[  \big\vert \widetilde H_{ij} \big( \bU_{t}, t \big)  \big\vert \big] - \bbE_{n} \big[  \big\vert \widehat H_{ij} \big( \bU_{t}, t \big)  \big\vert \big] \Big\vert  \Big],
\eean
where $\widetilde H_{ij} (\bx,t) = \sigma_{t}^{-4} \mu_{t}^2 \text{Cov} [ X_{\underline{T}, i}, X_{\underline{T}, j} \mid \bX_{t} = \bx ], i \neq j \in [D]$.
We proceed by bounding each term on the right-hand side separately; see (\ref{eqthm2:hess_trunc_bd}) and (\ref{eqthm2:hess_hat_bd}).

We first focus on deriving (\ref{eqthm2:hess_trunc_bd}).
For notational simplicity, let $s = \underline{T}$ and $S = [-\mu_{t}\gamma, \mu_{t}\gamma]^{D}$.
For all $i \neq j \in [D]$, we have
\be
\begin{split}
 & \Big\vert  \bbE \big[  \big\vert H_{ij} \big( \bU_{t}, t \big)  \big\vert  \big] -  \bbE \big[   \big\vert \widetilde H_{ij} \big( \bU_{t}, t \big)  \big\vert \big]  \Big\vert 
 \\
& =  \left\vert ( 2\mu_{t} \gamma)^{-D}  \frac{\mu_{t}^2}{\sigma_{t}^{4}} 
 \int_{S}   \Big\vert \text{Cov} \big[ X_{0,i}, X_{0,j} \mid \bX_{t} = \bx \big] \Big\vert  -
\Big\vert \text{Cov} \big[ X_{s,i}, X_{s,j} \mid \bX_{t} = \bx \big] \Big\vert  \d \bx \right\vert
 \\
 & \leq ( 2\mu_{t} \gamma)^{-D}  \frac{\mu_{t}^2}{\sigma_{t}^{4}} 
\int_{S}   \Big\vert \text{Cov} \big[ X_{0,i}, X_{0,j} \mid \bX_{t} = \bx \big] -
\text{Cov} \big[ X_{s,i}, X_{s,j} \mid \bX_{t} = \bx \big] \Big\vert
 \d \bx,
 \label{eqthm2:0_s_diff}
 \end{split}
\ee
where the last inequality holds by the triangle inequality.
We proceed by bounding the difference between the two conditional covariance terms; see (\ref{eqthm2:cov_trunc_bd}).

Fix $ \bx_{s} \in \bbR^{D}, \bx_{t} \in [-\mu_{t}\gamma, \mu_{t}\gamma]^{D}$, and $ \bx_{0}\in [-1,1]^{D}$.
A simple calculation yields that
\bean
p_{t \mid s,0} (\bx_{t} \mid \bx_{s}, \bx_{0})
= \frac{p_{t,s,0} (\bx_{t},\bx_{s},\bx_{0}) }{ p_{s,0} (\bx_{s}, \bx_{0}) }
= \frac{ p_{s \mid t,0} (\bx_{s} \mid \bx_{t}, \bx_{0}) p_{t,0} (\bx_{t}, \bx_{0})  }{ p_{s,0} (\bx_{s}, \bx_{0}) }
= \frac{p_{s \mid t,0} (\bx_{s} \mid \bx_{t}, \bx_{0}) p_{t \mid 0} (\bx_{t} \mid \bx_{0})  }{ p_{s \mid 0} (\bx_{s} \mid \bx_{0}) }.
\eean
Here, $p_{s,0}(\cdot,\cdot)$ denotes the joint density of $(\bX_{s},\bX_{0})$, and $p_{s \mid 0}(\cdot \mid \bx_{0})$ denotes the conditional density of $\bX_{s}$ given $\bX_{0} = \bx_{0}$.
Since the process $(\bX_{u})_{u \geq 0}$ is Markov, we have
\bean
p_{t \mid s,0} (\bx_{t} \mid \bx_{s}, \bx_{0}) 
= p_{t \mid s} (\bx_{t} \mid \bx_{s}).
\eean
Combining the last two displays yields that
\bean
p_{s \mid t,0} ( \bx_{s} \mid \bx_{t}, \bx_{0} ) 
=
\frac{p_{t \mid s} (\bx_{t} \mid \bx_{s})  p_{s \mid 0} (\bx_{s} \mid \bx_{0})}{p_{t \mid 0} (\bx_{t} \mid \bx_{0})}.
\eean
Note that the conditional distribution of $\bX_{t} $ given $\bX_{s} = \bx_{s}$ is Gaussian with $\cN ( \mu_{t-s} \bx_{s}, \sigma_{t-s}^2 \bI_{D} ) $, and that the conditional distribution of $\bX_{s}$ given $\bX_{0} = \bx_0$ is $\cN (\mu_{s} \bx_{0}, \sigma_{s}^2 \bI_{D})$.
A direct calculation yields that
\bean
&& \exp \left( - \frac{  \| \bx_{t} - \mu_{t-s} \bx_{s} \|_2^2 }{ 2 \sigma_{t-s}^2 }
- \frac{  \| \bx_{s} - \mu_{s} \bx_{0} \|_2^2 }{ 2 \sigma_{s}^2 } \right)
\\
&& = C(\bx_{0},\bx_{t}) \exp \left( -\frac{1}{2} \left\{ \left( \bx_{s}^{\top} \bx_{s} \right) \left( \frac{\mu_{t-s}^2}{\sigma_{t-s}^2} + \frac{1}{\sigma_{s}^2} \right)
- 2 \bx_{s}^{\top}  \left( \frac{\mu_{t-s} \bx_{t}} {\sigma_{t-s}^2} + \frac{\mu_{s} \bx_{0 }}{\sigma_{s}^2} \right) \right\} \right),
\eean
where $C(\bx_{0},\bx_{t})$ is a constant depending only on $\bx_{0}$ and $\bx_{t}$.
Recall that $\mu_{t} = e^{-t}$ and $\sigma_{t} = \sqrt{1 - e^{-2t}}$.
Therefore, the coefficient of the quadratic term is
\bean
\frac{\mu_{t-s}^2}{\sigma_{t-s}^2} + \frac{1}{\sigma_{s}^2} 
= \frac{\mu_{t-s}^2 \sigma_{s}^2 + \sigma_{t-s}^2 }{ \sigma_{t-s}^2 \sigma_{s}^2}
= \frac{e^{- 2t + 2s} (1-e^{- 2s}) + 1 - e^{- 2t+ 2s} }{ \sigma_{t-s}^2 \sigma_{s}^2}
= \frac{ 1 - e^{-2t} }{ \sigma_{t-s}^2 \sigma_{s}^2}
= \frac{ \sigma_{t}^2 }{ \sigma_{t-s}^2 \sigma_{s}^2}.
\eean
Hence, the conditional distribution of $\bX_{s}$ given $\bX_{t} = \bx_{t}$ and $\bX_{0} = \bx_{0}$ is Gaussian with
\be
\cN
\left( \frac{ \sigma_{s}^2 \mu_{t-s} \bx_{t} + \sigma_{t-s}^2 \mu_{s} \bx_{0 } }{\sigma_{t}^2}  , \frac{ \sigma_{t-s}^2 \sigma_{s}^2 }{\sigma_{t}^2} \bI_{D}
\right).
\label{eqthm2:gaussian}
\ee

To derive the conditional covariance of $\bX_{s}$ given $\bX_{t} = \bx_{t}$, we apply the law of total covariance:
\bean
\text{Cov} \big[ X_{s,i}, X_{s,j} \mid \bX_{t} = \bx_{t} \big]
&& = \bbE\Big[ \text{Cov} \big[ X_{s,i}, X_{s,j} \mid \bX_{t} = \bx_{t}, \bX_{0} \big]
\mid \bX_{t} = \bx_{t} \Big]
\\
&& \quad + \text{Cov} \Big[  \bbE \big[  X_{s,i} \mid \bX_{t} = \bx_{t}, \bX_{0} \big] , \bbE\big[  X_{s,j} \mid \bX_{t} = \bx_{t}, \bX_{0} \big] \mid \bX_{t} = \bx_{t} \Big].
\eean
By (\ref{eqthm2:gaussian}), we have
\be
\bbE \big[ X_{s,i} \mid \bX_{t} = \bx_{t}, \bX_{0} = \bx_{0} \big]
=  \frac{ \sigma_{s}^2 \mu_{t-s} x_{t,i} + \sigma_{t-s}^2 \mu_{s} x_{0,i} }{\sigma_{t}^2},
\quad \forall i \in [D],
\label{eqthm2:conditional_mean}
\ee
and
\bean
\text{Cov} \big[ X_{s,i}, X_{s,j} \mid \bX_{t} = \bx_{t}, \bX_{0} = \bx_0 \big]
=  \left(\frac{ \sigma_{t-s}^2 \sigma_{s}^2 }{\sigma_{t}^2} \right) \delta_{ij},
\quad \forall i,j \in [D],
\eean
where $\bx_{t} = (x_{t,i}), \bx_{0} = (x_{0,i})$, and $\delta_{ij}$ denotes the Kronecker delta.
Combining the last three displays, we have
\bean
\text{Cov} \big[ X_{s,i}, X_{s,j} \mid \bX_{t} = \bx_{t} \big]
= \left( \frac{ \sigma_{t-s}^2 \sigma_{s}^2 }{\sigma_{t}^2} \right) \delta_{ij}
+ \left( \frac{\sigma_{t-s}^4 \mu_{s}^2 }{\sigma_{t}^4} \right) \text{Cov} \big[ X_{0,i}, X_{0,j} \mid \bX_{t} = \bx_{t} \big], 
\quad \forall i,j \in [D].
\eean
Hence, for all $i \neq j \in [D]$,
\bean
\Big\vert \text{Cov} \big[ X_{s,i}, X_{s,j} \mid \bX_{t} = \bx_{t} \big] - 
\text{Cov} \big[ X_{0,i}, X_{0,j} \mid \bX_{t} = \bx_{t} \big] \Big\vert
=
\Big\vert 1 - \frac{\sigma_{t-s}^4 \mu_{s}^2 }{\sigma_{t}^4}  \Big\vert
\Big\vert 
\text{Cov} \big[ X_{0,i}, X_{0,j} \mid \bX_{t} = \bx_{t} \big]
\Big\vert.
\eean
A simple calculation yields that
\bean
\sigma_{t}^{4} - \sigma_{t-s}^4 \mu_{s}^2 
&&  = 
  1 - 2e^{-2t} + e^{-4t} - e^{-2s} ( 1 - 2e^{-2t + 2s} + e^{-4t + 4s} )
  \\
&&  =  1 - e^{-2s} + e^{-4t} - e^{-4t + 2s}  
 \\
 && =  (   1 - e^{-4t + 2s}) (1 - e^{-2s}) = \sigma_{s}^2(   1 - e^{-4t + 2s}).
\eean
Since $0 < s \leq t$ and $0 < e^{-4t + 2s} < 1$, we have $\vert 1 - \sigma_{t}^{-4} \sigma_{t-s}^4 \mu_{s}^2 \vert \leq \sigma_{t}^{-4}\sigma_{s}^2$.
Therefore, for all $i \neq j \in [D]$,
\bean
\Big\vert \text{Cov} \big[ X_{s,i}, X_{s,j} \mid \bX_{t} = \bx_{t} \big] - 
\text{Cov} \big[ X_{0,i}, X_{0,j} \mid \bX_{t} = \bx_{t} \big] \Big\vert
\leq
\frac{\sigma_{s}^2 }{\sigma_{t}^4}
\Big\vert 
\text{Cov} \big[ X_{0,i}, X_{0,j} \mid \bX_{t} = \bx_{t} \big]
\Big\vert.
\eean
Combining this with Theorem~\ref{thm1}, for all $i \neq j \in [D]$, we have
\be
\Big\vert \text{Cov} \big[ X_{s,i}, X_{s,j} \mid \bX_{t} = \bx_{t} \big] - 
\text{Cov} \big[ X_{0,i}, X_{0,j} \mid \bX_{t} = \bx_{t} \big] \Big\vert
 \leq D_{1} \sigma_{s}^2 \sigma_{t}^{ ( 2 d(i,j) - 2 ) \wedge {\tilde \beta} },
 \label{eqthm2:cov_trunc_bd}
\ee
where $D_{1} = D_{1} ( C_{2} ) > 0$.
Combining this with (\ref{eqthm2:0_s_diff}), and noting that $\mu_{t}^2 \leq 1$, we have
\be
 \left\vert  \bbE \big[  \big\vert H_{ij} \big( \bU_{t}, t \big)  \big\vert  \big] -  \bbE \big[   \big\vert \widetilde H_{ij} \big( \bU_{t}, t \big)  \big\vert \big]  \right\vert 
\leq 
D_{1} \sigma_{s}^2 \sigma_{t}^{  ( 2 d(i,j) -2)\wedge {\tilde \beta} - 4 },
\quad \forall i \neq j \in [D].
 \label{eqthm2:hess_trunc_bd}
\ee

We now focus on deriving (\ref{eqthm2:hess_hat_bd}).
For all $i \neq j \in [D]$, we have
\be
\begin{split}
 & \bbE \left[  \left\vert \bbE \big[  \big\vert \widetilde H_{ij} \big( \bU_{t}, t \big)  \big\vert \big] - \bbE_{n} \big[  \big\vert \widehat H_{ij} \big( \bU_{t}, t \big)  \big\vert \big] \right\vert  \right]
 \\
& =   ( 2\mu_{t} \gamma)^{-D} \frac{\mu_{t}^2}{\sigma_{t}^{4}}
\bbE \left[ \left\vert 
\int_{S} \Big\vert \text{Cov} \big[ X_{s,i}, X_{s,j} \mid \bX_{t} = \bx \big] \Big\vert  -
\Big\vert \text{Cov}_n \big[ \widehat X_{s,i}, \widehat X_{s,j} \mid  \widehat \bX_{t} = \bx \big] \Big\vert 
 \d \bx  \right\vert \right]
 \\
& \leq \frac{\mu_{t}^2}{\sigma_{t}^{4}}
\bbE \left[
\int_{S}   ( 2\mu_{t} \gamma)^{-D} \Big\vert \text{Cov} \big[ X_{s,i}, X_{s,j} \mid \bX_{t} = \bx \big]  -
\text{Cov}_n \big[ \widehat X_{s,i}, \widehat X_{s,j} \mid  \widehat \bX_{t} = \bx \big] \Big\vert
 \d \bx \right],
 \label{eqthm2:cov_hat_bd_toy}
 \end{split}
 \ee
 where the last inequality holds by the triangle inequality.
 We proceed by bounding the difference between the two conditional covariance terms.

By the triangle inequality, for all $i,j \in [D]$, we have
\be
\begin{split}
& \Big\vert \text{Cov} \big[ X_{s,i}, X_{s,j} \mid \bX_{t} = \bx_{t} \big]  -
\text{Cov}_n \big[ \widehat X_{s,i}, \widehat X_{s,j} \mid  \widehat \bX_{t} = \bx_{t} \big] \Big\vert
\\
& \leq \Big\vert \bbE \big[ X_{s,i} X_{s,j} \mid \bX_{t} = \bx_{t} \big]  -
\bbE_n \big[ \widehat X_{s,i} \widehat X_{s,j} \mid  \widehat \bX_{t} = \bx_{t} \big] \Big\vert
\\
& \quad +
\Big\vert \bbE \big[ X_{s,i} \mid \bX_{t} = \bx_{t} \big] \bbE \big[ X_{s,j} \mid \bX_{t} = \bx_{t} \big]   -
\bbE_n \big[ \widehat X_{s,i} \mid  \widehat \bX_{t} = \bx_{t} \big]
\bbE_n \big[ \widehat X_{s,j} \mid  \widehat \bX_{t} = \bx_{t} \big]
\Big\vert
\label{eqthm2:cov_toy_bd1}
\end{split}
\ee
and
\be
\begin{split}
& \Big\vert \bbE \big[ X_{s,i} \mid \bX_{t} = \bx_{t} \big] \bbE \big[ X_{s,j} \mid \bX_{t} = \bx_{t} \big]   -
\bbE_n \big[ \widehat X_{s,i} \mid  \widehat \bX_{t} = \bx_{t} \big]
\bbE_n \big[ \widehat X_{s,j} \mid  \widehat \bX_{t} = \bx_{t} \big]
\Big\vert
\\
& \leq \Big\vert \bbE \big[ X_{s,i} \mid \bX_{t} = \bx_{t} \big] \Big\vert
\cdot
\Big\vert \bbE \big[ X_{s,j} \mid \bX_{t} = \bx_{t} \big]   -
\bbE_n \big[ \widehat X_{s,j} \mid  \widehat \bX_{t} = \bx_{t} \big]
\Big\vert
\\
& \quad + \Big\vert 
\bbE_n \big[ \widehat X_{s,j} \mid  \widehat \bX_{t} = \bx_{t} \big]
\Big\vert
\cdot
\Big\vert \bbE \big[ X_{s,i} \mid \bX_{t} = \bx_{t} \big]    -
\bbE_n \big[ \widehat X_{s,i} \mid  \widehat \bX_{t} = \bx_{t} \big]
\Big\vert.
\label{eqthm2:cov_toy_bd2}
\end{split}
\ee
We proceed by bounding the differences in the conditional first and second moments; see (\ref{eqthm2:sec_hat_bd}) and (\ref{eqthm2:fir_hat_bd}).

Let
\be
L =  \left\{ 2 \sigma_{s} \sqrt{ \log (1/\epsilon_n) } + 2\right\}^{2}
 \vee \left\{ 4 e^{t-s} \sqrt{t \log (1/\epsilon_n) } + e^{t-s} + 2 e^{t-s}\sigma_{t-s} \sqrt{\log (1 / \epsilon_n) } \right\}^2.
\label{eqthm2:L_def}
\ee
Let $\widehat p_{s \mid t} (\cdot \mid \bx_{t} )$ denote conditional the density of $\widehat \bX_{s}$ given $\widehat \bX_{t} = \bx_{t} $ and the observations $\bX^1, \ldots, \bX^n$.
A simple calculation yields that
\bean
&& \Big\vert \bbE \big[  X_{s,i}  X_{s ,j} \mid  \bX_{t} = \bx_{t} \big] - \bbE_n \big[ \widehat X_{s,i} \widehat X_{s ,j} \mid \widehat \bX_{t} = \bx_{t} \big]  \Big\vert 
\\
&& =   \Big\vert \int_{\bbR^{D}} y_i y_j  \left\{ p_{s \mid t}(\by \mid \bx_{t} ) -  \widehat p_{s \mid t}(\by \mid \bx_{t} )  \right\} \d \by  \Big\vert 
\\
&& \leq \int_{\bbR^{D}}  \big\vert y_i y_j \big\vert   \big\vert p_{s \mid t}(\by \mid \bx_{t} ) -  \widehat p_{s \mid t}(\by \mid \bx_{t} ) \big\vert \d \by 
\\
&& \leq 2 L d_{\rm TV} \left( p_{s \mid t}(\cdot \mid \bx_{t} ), \widehat p_{s \mid t}(\cdot \mid \bx_{t} ) \right)
+  \int_{ \vert y_i y_j \vert > L}  \big\vert y_i y_j \big\vert \left\{   p_{s \mid t}(\by \mid \bx_{t} ) +  \widehat p_{s \mid t}(\by \mid \bx_{t} ) \right\} \d \by,
\quad \forall i,j \in [D],
\eean
where $d_{\rm TV} (p,q) = 2^{-1} \int_{\bbR^{D}} \vert p(\bx) -q(\bx) \vert \d \bx $ denotes the total variation distance between two probability density functions $p$ and $q$ on $\bbR^{D}$.
By the Cauchy--Schwarz inequality, we have
\bean
\int_{ \vert y_i y_j \vert > L}  \big\vert y_i y_j \big\vert   p_{s \mid t}(\by \mid \bx_{t} )\d \by
&& = \int_{ \bbR^{D} }  \big\vert y_i y_j \big\vert 1 \big( \vert y_i y_j \vert > L \big)  p_{s \mid t}(\by \mid \bx_{t} )\d \by
\\
&& \leq \sqrt{ \bbE \big[ X_{s,i}^2 X_{s,j}^2 \mid \bX_{t} = \bx_{t} \big] 
\bbP \big[ \vert X_{s,i} X_{s,j} \vert > L \mid \bX_{t} = \bx_{t} \big]  },
\eean
and similarly, 
\bean
\int_{ \vert y_i y_j \vert > L}  \big\vert y_i y_j \big\vert  \widehat p_{s \mid t}(\by \mid \bx_{t} )\d \by
\leq \sqrt{ \bbE_{n} \big[ \widehat X_{s,i}^2 \widehat X_{s,j}^2 \mid \widehat \bX_{t} = \bx_{t} \big] 
\bbP_{n} \big[ \vert \widehat X_{s,i} \widehat X_{s,j} \vert > L \mid \widehat \bX_{t} = \bx_{t} \big]  },
\eean
where $\bbP_{n}(\cdot)$ denotes the conditional probability given the $n$ observations $\bX^1,\ldots,\bX^n$.

Since $\{ (x,y) \in \bbR^{2} : \vert xy \vert > a \} \subseteq \{ (x,y) \in \bbR^2 : \vert x \vert > \sqrt{a} \} \cup \{ (x,y) \in \bbR^2 :  \vert y \vert > \sqrt{a} \}  $ for any $a > 0$, we have
\bean
\bbP \Big(
\big\vert   X_{s,i}  X_{s,j}   \big\vert > L \mid  \bX_{t} = \bx_{t} 
\Big) 
\leq
\bbP \Big(
\big\vert   X_{s,i}  \big\vert > \sqrt{L} \mid  \bX_{t} = \bx_{t} 
\Big)
+
\bbP \Big(
\big\vert   X_{s,j}  \big\vert > \sqrt{L} \mid  \bX_{t} = \bx_{t} 
\Big).
\eean
Combining the last three displays yields that for all $i,j \in [D]$, we have
\be
\begin{split}
& \Big\vert \bbE \big[  X_{s,i}  X_{s ,j} \mid  \bX_{t} = \bx_{t} \big] - \bbE_n \big[ \widehat X_{s,i} \widehat X_{s ,j} \mid \widehat \bX_{t} = \bx_{t} \big]  \Big\vert
\\
& \leq 2 L d_{\rm TV} \left( p_{s \mid t}(\cdot \mid \bx_{t} ), \widehat p_{s \mid t}(\cdot \mid \bx_{t} ) \right) 
\\
& + \sqrt{ \bbE \big[ X_{s,i}^2 X_{s,j}^2 \mid \bX_{t} = \bx_{t} \big]  \Big\{  \bbP \Big(
\big\vert   X_{s,i}  \big\vert > \sqrt{L} \mid  \bX_{t} = \bx_{t} 
\Big)
+ \bbP \Big(
\big\vert   X_{s,j}  \big\vert > \sqrt{L} \mid  \bX_{t} = \bx_{t}
\Big) \Big\} } 
\\
& \quad + \sqrt{ \bbE_{n} \big[ \widehat X_{s,i}^2 \widehat X_{s,j}^2 \mid \widehat \bX_{t} = \bx_{t} \big]  \Big\{  \bbP_{n} \Big(
\big\vert \widehat  X_{s,i}  \big\vert > \sqrt{L} \mid \widehat \bX_{t} = \bx_{t} 
\Big)
+ \bbP_{n} \Big(
\big\vert  \widehat X_{s,j}  \big\vert > \sqrt{L} \mid \widehat \bX_{t} = \bx_{t}
\Big) \Big\} }.
\label{eqthm2:tot_sec_mom_bd}
\end{split}
\ee
Similarly, for all $i\in [D]$, we have
\be
\begin{split}
& \Big\vert \bbE \big[  X_{s,i}  \mid  \bX_{t} = \bx_{t} \big] - \bbE_n \big[ \widehat X_{s,i}  \mid \widehat \bX_{t} = \bx_{t} \big]  \Big\vert
\\
& \leq 2 L d_{\rm TV} \left( p_{s \mid t}(\cdot \mid \bx_{t} ), \widehat p_{s \mid t}(\cdot \mid \bx_{t} ) \right) 
+ \sqrt{ \bbE \big[ X_{s,i}^2 \mid \bX_{t} = \bx_{t} \big]  \bbP \Big(
\big\vert   X_{s,i}  \big\vert > L \mid  \bX_{t} = \bx_{t} 
\Big) } 
\\
& \quad + \sqrt{ \bbE_{n} \big[  \widehat X_{s,i}^2 \mid \widehat \bX_{t} = \bx_{t} \big]  \bbP_{n} \Big(
\big\vert  \widehat  X_{s,i}  \big\vert > L \mid \widehat  \bX_{t} = \bx_{t} 
\Big) }.
\label{eqthm2:tot_fir_mom_bd}
\end{split}
\ee
We proceed by bounding each term on the right-hand side separately; see (\ref{eqthm2:mom_true_bd}), (\ref{eqthm2:tv_bd}), and (\ref{eqthm2:mom_hat_bd}).

We first focus on deriving (\ref{eqthm2:mom_true_bd}).
Recall that the conditional distribution of $\bX_{s}$ given $\bX_{t} = \bx_t$ and $\bX_{0} = \bx_{0}$ is Gaussian; see \eqref{eqthm2:gaussian}.
Specifically, it is distributed as $\cN ( \widetilde \bmu, \widetilde \sigma^2 \bI_{D})$, where
\bean
\widetilde \bmu = \frac{ \sigma_{s}^2 \mu_{t-s} \bx_{t} + \sigma_{t-s}^2 \mu_{s} \bx
_{0} }{\sigma_{t}^2} 
\quad \text{and} \quad 
\widetilde \sigma = \frac{ \sigma_{t-s} \sigma_{s} }{\sigma_{t}}.
\eean
Hence,
\bean
\bbP \Big(
\big\vert   X_{s,i}  \big\vert > \sqrt{L} \mid  \bX_{t} = \bx_{t}, \bX_{0} = \bx_{0}
\Big)
= \bbP \Big(
\big\vert  \widetilde \sigma Z + \widetilde \mu_{i}  \big\vert > \sqrt{L}  \Big), \quad \forall i \in [D],
\eean
where $\widetilde \bmu = (\widetilde \mu_1, \ldots, \widetilde \mu_{D})$ and $Z$ denotes a one-dimensional standard normal random variable.
A simple calculation yields that
\bean
\bbP \Big(
\big\vert  \widetilde \sigma Z + \widetilde \mu_{i}  \big\vert > \sqrt{L}  \Big)
&& \leq \bbP \Big( Z > \frac{ \sqrt{L}  - \widetilde \mu_i }{\widetilde \sigma} \Big)
+ \bbP \Big( Z < \frac{ - \sqrt{L}  - \widetilde \mu_i }{\widetilde \sigma} \Big)
\\
&& \leq 2 \bbP \Big( Z > \frac{ \sqrt{L}  - \| \widetilde \bmu \|_{\infty} }{\widetilde \sigma} \Big),
\quad \forall i \in [D].
\eean

Note that $\| \widetilde \bmu \|_{\infty} \leq \sigma_{t}^{-2} ( \sigma_{s}^2 + \sigma_{t-s}^2 ) \leq 2,$ since $\| \bx_{t} \|_{\infty} \leq \mu_{t}\gamma \leq 1$ and $\| \bx_0 \|_{\infty} \leq 1$.
Combining this with the definition of $L$ (see (\ref{eqthm2:L_def})), we have
\bean
\sqrt{L}
\geq 2 \sigma_{s} \sqrt{ \log (1/\epsilon_n) } + 2
\geq 2 \widetilde \sigma \sqrt{ \log (1/\epsilon_n) } + \| \widetilde \bmu \|_{\infty},
\eean
which implies that $\widetilde \sigma^{-1} (\sqrt{L}  - \| \widetilde \bmu \|_{\infty} ) \geq 2 \sqrt{ \log ( 1/\epsilon_n )} > 0$.
Combining this with the tail probability of the normal distribution, we have
\bean
\bbP \Big(
\big\vert   X_{s,i}  \big\vert > \sqrt{L}  \mid  \bX_{t} = \bx_{t}, \bX_{0} = \bx_{0}
\Big)
\leq
 2 \bbP \Big( Z > \frac{ \sqrt{L}  - \| \widetilde \bmu \|_{\infty} }{\widetilde \sigma} \Big)
\leq
2\epsilon_n^2, \quad \forall i \in [D].
\eean
Since $L \geq 1$, we have $L \geq \sqrt{L}$.
Therefore,
\bean
\bbP \Big(
\big\vert   X_{s,i}  \big\vert > L \mid  \bX_{t} = \bx_{t}, \bX_{0} = \bx_{0}
\Big)
\leq
\bbP \Big(
\big\vert   X_{s,i}  \big\vert > \sqrt{L} \mid  \bX_{t} = \bx_{t}, \bX_{0} = \bx_{0}
\Big)
\leq
2\epsilon_n^2, \quad \forall i \in [D].
\eean
Note that the last two displays hold uniformly for all $\bx_{0} \in [-1,1]^{D}$.
Taking the conditional expectation with respect to $\bX_0$ given $\bX_t=\bx_t$ (by the law of total expectation), for any $i \in [D]$, we have
\be
\bbP \Big(
\big\vert   X_{s,i}  \big\vert > \sqrt{L} \mid  \bX_{t} = \bx_{t}
\Big) \leq 2 \epsilon_n^2
\quad \text{and} \quad
\bbP \Big(
\big\vert   X_{s,i}  \big\vert > L \mid  \bX_{t} = \bx_{t}
\Big) \leq 2 \epsilon_n^2.
\label{eqthm2:tail_true_bd}
\ee

Since the conditional distribution of $\bX_{s}$ given $\bX_{t} = \bx_t$ and $\bX_{0} = \bx_{0}$ is Gaussian with covariance matrix $\widetilde \sigma^2 \bI_{D}$, its coordinates are conditionally independent.
Hence, for any $i \neq j \in [D]$,
\bean
 \bbE \big[ X_{s,i}^2 X_{s,j}^2 \mid \bX_{t} = \bx_{t}, \bX_{0} = \bx_{0} \big]
 && =  \bbE \big[ X_{s,i}^2  \mid \bX_{t} = \bx_{t}, \bX_{0} = \bx_{0} \big]
  \bbE \big[  X_{s,j}^2 \mid \bX_{t} = \bx_{t}, \bX_{0} = \bx_{0} \big]
  \\
  && = \big( \widetilde \mu_{i}^2 + \widetilde \sigma^2 \big) \big( \widetilde \mu_{j}^2 + \widetilde \sigma^2 \big)
 \leq \big( \| \widetilde \bmu \|_{\infty}^2 + \widetilde \sigma^2 \big)^{2}.
\eean
Moreover, $\| \widetilde \bmu \|_{\infty}  \leq \sigma_{t}^{-2}(\sigma_{s}^2 + \sigma_{t-s}^2) \leq 2$ and $\widetilde \sigma \leq  \sigma_{t-s} \leq 1$, since $\| \bx_{t} \|_{\infty} \leq \mu_{t}\gamma \leq 1$, $\| \bx_{0} \|_{\infty} \leq 1$ and $ s \leq t$.
Therefore, by the law of total expectation, we have
\bean
 \bbE \big[ X_{s,i}^2 X_{s,j}^2 \mid \bX_{t} = \bx_{t} \big]
\leq 25, \quad \forall i \neq j \in [D],
\quad \text{and} \quad
 \bbE \big[ X_{s,i}^2 \mid \bX_{t} = \bx_{t} \big]
\leq 5, \quad \forall i \in [D].
\eean
Combining the last display with (\ref{eqthm2:tail_true_bd}), we have
\be
\begin{split}
 & \sqrt{ \bbE \big[ X_{s,i}^2 X_{s,j}^2 \mid \bX_{t} = \bx_{t} \big]  \Big\{  \bbP \Big(
\big\vert   X_{s,i}  \big\vert > \sqrt{L} \mid  \bX_{t} = \bx_{t} 
\Big)
+ \bbP \Big(
\big\vert   X_{s,j}  \big\vert > \sqrt{L} \mid  \bX_{t} = \bx_{t}
\Big) \Big\} } 
\\
&  \quad \leq 10 \epsilon_{n}, \quad \forall i \neq j \in [D], 
\qquad \text{and}
\\
& \sqrt{ \bbE \big[ X_{s,i}^2  \mid \bX_{t} = \bx_{t} \big]  \bbP \Big(
\big\vert   X_{s,i}  \big\vert > L \mid  \bX_{t} = \bx_{t} 
\Big) } 
\leq \sqrt{10} \epsilon_{n},
\quad \forall i \in [D].
\label{eqthm2:mom_true_bd}
\end{split}
\ee

We now focus on deriving (\ref{eqthm2:tv_bd}).
We will apply Corollary~1.2 of \cite{bogachev2016distances} to obtain an upper bound for the total variation distance.
Consider two stochastic processes $( \bZ_{u})_{u \in [0, t - s] } $ and $( \widehat \bZ_{u})_{u \in [0, t - s ] } $ defined by the SDEs
\bean
&& \d  \bZ_{u} = \Big[  \bZ_{u} + 2  \bff_0 \big(  \bZ_{u}, t - u \big) \Big] \d u + \sqrt{2} \d \bB_u,\quad  \bZ_0 \sim \delta_{\bx_{t}},
\\
&& \d \widehat \bZ_{u} = \Big[ \widehat \bZ_{u} + 2 \widehat \bff \big( \widehat \bZ_{u}, t - u \big) \Big] \d u + \sqrt{2} \d \bB_u,\quad \widehat \bZ_0 \sim \delta_{\bx_{t}},
\eean
where $\delta_{\bx_{t}}$ denotes the $D$-dimensional Dirac measure at $\bx_{t}$.
Then, for $u \in (0,t-s]$, $\bZ_{u}$ and $\widehat \bZ_{u}$ have Lebesgue densities given by $p_{t-u \mid t} (\cdot \mid \bx_{t})$ and $\widehat p_{t-u \mid t}(\cdot \mid \bx_{t})$, respectively.
Define functions $q_1, q_2 : \bbR^{D} \times (0, t-s ] \rightarrow \bbR $ by $q_1(\bz,u) = p_{t-u \mid t}(\bz \mid \bx_{t} )$ and $q_2(\bz,u) = \widehat p_{t-u \mid t} (\bz \mid \bx_{t} )$.
Then, $q_1$ and $q_2$ satisfy the corresponding Fokker--Planck equations  \citep{lebris2008existence, bogachev2022fokker, pavliotis2014stochastic} :
\bean
&& \frac{\partial }{\partial u} q_1(\bz,u) = - \sum_{i=1}^{D} \frac{\partial}{\partial z_i} \big[ \bb_1 (\bz,u) q_1(\bz,u) \big] + \sum_{i=1}^{D} \sum_{j=1}^{D} \frac{\partial^2}{\partial z_i \partial z_j} \big[ \delta_{ij} q_1(\bz,u) \big],
\\
&& \frac{\partial }{\partial u} q_2(\bz,u) = - \sum_{i=1}^{D} \frac{\partial}{\partial z_i} \big[ \bb_2 (\bz,u) q_2(\bz,u) \big] + \sum_{i=1}^{D} \sum_{j=1}^{D} \frac{\partial^2}{\partial z_i \partial z_j} \big[ \delta_{ij} q_2(\bz,u) \big],
\eean
where
\bean
\bb_1 (\bz,u) =  \bz + 2 \bff_{0} (\bz, t-u),
\quad 
\bb_2 (\bz,u) =  \bz + 2  \widehat \bff( \bz, t-u).
\eean
By a change of variables,
\bean
\int_{0}^{t-s}
\int_{\bbR^{D}} 
\big\| \bb_{1}(\bz,u) - \bb_{2}(\bz,u) \big\|_2^2 q_1(\bz,u)
\d \bz \d u
&& = 4 \int_{s}^{t}
\int_{\bbR^{D}} 
\big\| \bff_{0}(\bz,u) - \widehat \bff(\bz,u) \big\|_2^2 p_{u\mid t}(\bz \mid \bx_{t})
\d \bz \d u
\\
&& = 4 \int_{s}^{t}
\bbE_{n} \Big[
\big\| \bff_{0}( \bX_{u} ,u) - \widehat \bff( \bX_{u} ,u ) \big\|_2^2 \mid \bX_{t} = \bx_{t} \Big] \d u.
\eean

To apply Corollary~1.2 of \cite{bogachev2016distances}, it suffices to show that the right-hand side is finite.
Note that $p_{0}$ is bounded away from zero on its support, since $\log p_{0} \in \cH^{\beta,K}_{D}([-1,1]^{D})$ and hence $p_{0}(\bx) \geq e^{-K}$ for all $\bx \in [-1,1]^{D}$.
By Lemma~7 of \cite{kwon2026nonparametric}, we have
\bean
\big\| \bff_{0} (\bz, u) \big\|_2
\leq \frac{D_{2}}{\sigma_{u}} \left( \frac{\| \bz \|_{\infty} - \mu_{u}}{\sigma_{u}} \vee 1  \right),
\quad \forall \bz \in \bbR^{D}, \forall u \geq 0,
\eean
where $D_{2} = D_{2} ( D, K) > 0$.
By assumption, $\| \widehat \bff(\cdot,u) \|_{\infty} \leq \sigma_{u}^{-1} \sqrt{ \log (1/\epsilon_n) }$ for all $u \in [s,t]$.
Since $\| \widehat \bff(\cdot,u) \|_{2}^2 \leq D \| \widehat \bff(\cdot,u) \|_{\infty}^2$, it follows that
\bean
\big\| \bff_{0} (\bz, u) - \widehat \bff (\bz, u) \big\|_2^2 
&& \leq 2 \big\| \bff_{0} (\bz, u) \big\|_2^2 + 2 \big\| \widehat \bff (\bz, u) \big\|_2^2 
\\
&& \leq  \frac{2 D_{2}^2}{\sigma_{u}^2} \left( \frac{\| \bz \|_{\infty} - \mu_{u}}{\sigma_{u}} \vee 1  \right)^2
+ \frac{2 D \log ( 1 / \epsilon_n ) }{\sigma_{u}^2},
\\
&& \leq 
\frac{4 D_{2}^2 }{\sigma_{u}^4} \big( \| \bz \|_{\infty}^2 +  \mu_{u}^2 \big)
+ \frac{2 D_{2}^2 + 2 D \log (1/\epsilon_n) }{\sigma_{u}^2},
\quad \forall \bz \in \bbR^{D}, \forall u \in [s,t],
\eean
where the last inequality holds because $(a \vee b)^2 \leq a^2 + b^2$ for any $a,b \in \bbR$.
Since $\mu_{u}^2 \leq 1$, we have
\bean
&& \bbE_{n} \Big[
\big\| \bff_{0}( \bX_{u} ,u) - \widehat \bff( \bX_{u} ,u ) \big\|_2^2 \mid \bX_{t} = \bx_{t} \Big]
\\
&& \leq \frac{4 D_{2}^2 }{\sigma_{u}^4} 
\bbE \Big[
\big\|  \bX_{u} \big\|_{\infty}^2 \mid \bX_{t} = \bx_{t} \Big]
+ \frac{4 D_{2}^2}{\sigma_{u}^4} + \frac{2 D_{2}^2 + 2 D \log (1/\epsilon_n) }{\sigma_{u}^2},
\quad \forall u \in [s,t].
\eean

By repeating the argument leading to \eqref{eqthm2:gaussian}, for $u \in [s,t]$, the conditional distribution of $\bX_{u}$ given $\bX_{t} = \bx_{t}$ and $\bX_{0} = \bx_{0}$ is Gaussian with  
\bean
\cN
\left( \frac{ \sigma_{u}^2 \mu_{t-u} \bx_{t} + \sigma_{t-u}^2 \mu_{u} \bx_{0 } }{\sigma_{t}^2}  , \frac{ \sigma_{t-u}^2 \sigma_{u}^2 }{\sigma_{t}^2} \bI_{D}
\right).
\eean
Combining this with the law of total expectation, we have
\bean
&& \bbE \Big[
\big\|  \bX_{u} \big\|_{\infty}^2 \mid \bX_{t} = \bx_{t} \Big]
=
\bbE \left[ \bbE \Big[
\big\|  \bX_{u} \big\|_{\infty}^2 \mid \bX_{t}, \bX_{0} \Big] \mid \bX_{t} = \bx_{t} \right]
\leq \bbE \left[ \bbE \Big[
\big\|  \bX_{u} \big\|_{2}^2 \mid \bX_{t}, \bX_{0} \Big] \mid \bX_{t} = \bx_{t} \right]
\\
&& \leq \sum_{i=1}^{D} \bbE \left[ \bbE \Big[
\big\vert  X_{u,i} \big\vert^2 \mid \bX_{t}, \bX_{0} \Big] \mid \bX_{t} = \bx_{t} \right]
\leq \frac{ D \sigma_{t-u}^2 \sigma_{u}^2 }{\sigma_{t}^2} + D \left( \frac{ \sigma_{u}^2 + \sigma_{t-u}^2 }{\sigma_{t}^2}  \right)^2 \leq 5D,
\eean
where the third inequality follows from the fact that $\bbE[ X^2] = \text{Var}[X] + ( \bbE[X] )^2$ for any random variable $X$.
Therefore, combining the above displays yields that
\bean
&& \int_{0}^{t-s}
\int_{\bbR^{D}} 
\big\| \bb_{1}(\bz,u) - \bb_{2}(\bz,u) \big\|_2^2 q_1(\bz,u)
\d \bz \d u
\\
&& \leq 4 \int_{s}^{t} \left( \frac{20 D D_{2}^2 + 4 D_{2}^2   }{\sigma_{u}^4}  + \frac{2 D_{2}^2 + 2 D \log (1/\epsilon_n) }{\sigma_{u}^2} \right) \d u < \infty.
\eean

We now apply Corollary~1.2 of \cite{bogachev2016distances}, which yields
\bean
\left\{ \int_{\bbR^{D}} \big\vert p_{s \mid t} (\bz \mid \bx_{t} ) - \widehat p_{s \mid t} ( \bz \mid \bx_{t} ) \big\vert \d \bz \right\}^2
&& \leq \int_{0}^{t-s}
\int_{\bbR^{D}} 
\big\| \bb_{1}(\bz,u) - \bb_{2}(\bz,u) \big\|_2^2 q_1(\bz,u)
\d \bz \d u
\\
&& = 4 \int_{s}^{t}
\int_{\bbR^{D}} 
\big\| \bff_{0}(\bz,u) - \widehat \bff(\bz,u) \big\|_2^2 p_{u\mid t}(\bz \mid \bx_{t})
\d \bz \d u.
\eean
Together with the definition of $d_{\rm TV}$, it follows that
\bean
 \Big\{ d_{\rm TV} \big( p_{ s \mid t}(\cdot \mid \bx_{t}), \widehat p_{ s \mid t}(\cdot \mid \bx_{t}) \big) \Big\}^2
\leq 
 \int_{s}^{t} \int_{\bbR^{D}}  \left\| \widehat \bff (\bz,u) - \bff_0 (\bz, u) \right\|_2^2 p_{u \mid t}(\bz \mid \bx_{t}) \d \bz \d u.
\eean

Recall that $S = [-\mu_{t}\gamma, \mu_{t} \gamma]^{D}$.
Since $p_{0}(\bx) \geq e^{-K}$ for all $\bx \in [-1,1]^{D}$ and $\bx_{t} \in S$, Lemma~6 of \cite{kwon2026nonparametric} implies that $p_{t}(\bx_{t}) \geq D_{3}$, where $D_{3} = D_{3} ( D, K) > 0$.
In particular, this lower bound holds uniformly over $\bx_{t}\in S$.
By Bayes' rule, for every $u \in [s,t)$ and $\bx \in S$, we have
\bean
p_{u \mid t} (\bz \mid \bx )
= \frac{ p_{ t \mid u} ( \bx \mid \bz) p_{u} (\bz) }{p_{t} (\bx)}
\leq D_{3}^{-1}  p_{ t \mid u} ( \bx \mid \bz) p_{u} (\bz) ,
\quad \bz \in \bbR^{D}.
\eean
Combining the last two displays, we have
\bean
&& \int_{S}  \Big\{ d_{\rm TV} \big( p_{ s \mid t}(\cdot \mid \bx), \widehat p_{ s \mid t}(\cdot \mid \bx) \big) \Big\}^2 \d \bx
\\
&& \leq  D_{3}^{-1} \int_{S} \int_{s}^{t} \int_{\bbR^{D}}  \left\| \widehat \bff (\bz,u) - \bff_0 (\bz, u) \right\|_2^2 p_{t \mid u}(\bx \mid \bz ) p_{u} (\bz) \d \bz \d u \d \bx
\\
&& = D_{3}^{-1}\int_{s}^{t} \int_{\bbR^{D}}  \left\| \widehat \bff (\bz,u) - \bff_0 (\bz, u) \right\|_2^2  p_{u} (\bz)  \left( \int_{S} p_{t \mid u}(\bx \mid \bz ) \d \bx \right)   \d \bz \d u 
\\
&& \leq D_{3}^{-1}  \int_{s}^{t} \int_{\bbR^{D}}  \left\| \widehat \bff (\bz,u) - \bff_0 (\bz, u) \right\|_2^2  p_{u} (\bz) \d \bz \d u.
\eean
By Theorem~\ref{thm1}, $\sigma_{u}/\mu_{u} = \sqrt{e^{2u}(1-e^{-2u})} \leq 1$ for every $u \leq C_{1}$, which implies that $C_{1} \leq 2^{-1} \log 2$.
Since $t \leq {\widetilde T} \leq C_{1}$, it follows that $\mu_{t}^{-2} = e^{2t} \leq 2$ and $(2\mu_{t}\gamma)^{-D} \leq 2^{-D/2}\gamma^{-D}$.
Combining with the Cauchy--Schwarz inequality, we have
\bean
&& \int_{S} (2\mu_{t} \gamma )^{-D}  d_{\rm TV} \big( p_{ s \mid t}(\cdot \mid \bx), \widehat p_{ s \mid t}(\cdot \mid \bx) \big)  \d \bx
\\
&& \leq \left[ \int_{S} (2\mu_{t} \gamma )^{-D} \Big\{ d_{\rm TV} \big( p_{ s \mid t}(\cdot \mid \bx), \widehat p_{ s \mid t}(\cdot \mid \bx) \big) \Big\}^2 \d \bx \right]^{1/2}
\\
&& \leq D_{4}
 \left( \int_{s}^{t} \int_{\bbR^{D}}  \left\| \widehat \bff (\bz,u) - \bff_0 (\bz, u) \right\|_2^2  p_{u} (\bz) \d \bz \d u \right)^{1/2},
\eean
where $D_{4} = D_{4} ( D, \gamma, D_{3} ) > 0$.
Taking the expectation and using the assumption yields
\be
\bbE \left[ \int_{S} (2\mu_{t} \gamma )^{-D}  d_{\rm TV} \big( p_{ s \mid t}(\cdot \mid \bx), \widehat p_{ s \mid t}(\cdot \mid \bx) \big)  \d \bx \right]
\leq D_{4} \epsilon_{n} .
\label{eqthm2:tv_bd}
\ee

We now focus on deriving (\ref{eqthm2:mom_hat_bd}).
Applying It\^{o}'s formula \citep{le2016brownian} to $e^{-u} \widehat \bZ_{u}$, we have
\bean
\d \Big( e^{-u} \widehat \bZ_{u} \Big)
&& = 
e^{-u} \d \widehat \bZ_{u}  + \widehat \bZ_{u} \d \Big( e^{-u} \Big)
\\
&& = e^{-u} \d \widehat \bZ_{u}  - e^{-u} \widehat \bZ_{u} \d u
\\
&& = 2 e^{-u} \widehat \bff \big( \widehat \bZ_{u}, t-u \big) \d u + \sqrt{2} e^{-u} \d \bB_{u}.
\eean
Integrating both sides over $[0,t-s]$ yields
\bean
e^{- (t-s) } \widehat \bZ_{t-s} 
= \bx_{t} + \int_{0}^{t-s} 2e^{-u} \widehat \bff \big( \widehat \bZ_{u}, t-u \big) \d u + \int_{0}^{t-s} \sqrt{2} e^{-u} \d \bB_{u}.
\eean
and hence,
\be
\widehat \bZ_{t-s} 
= e^{t-s} \bx_{t} + e^{t-s} \int_{0}^{t-s} 2e^{-u} \widehat \bff \big( \widehat \bZ_{u}, t-u \big) \d u + e^{t-s} \int_{0}^{t-s} \sqrt{2} e^{-u} \d \bB_{u}.
\label{eqthm2:ito}
\ee
Since $\int_{0}^{t-s} 2 e^{-2u} \d u = 1 - e^{-2t+2s} = \sigma_{t-s}^2 $, the random vector $\int_{0}^{t-s} \sqrt{2} e^{-u} \d \bB_{u}$ follows the Gaussian distribution with $\cN ( \mathbf{0}_{D}, \sigma_{t-s}^2 \bI_{D} )$.
By the assumptions $\| \widehat \bff (\cdot, u ) \|_{\infty} \leq \sigma_{u}^{-1} \sqrt{ \log (1/\epsilon_n) }, s \leq u \leq t $ and $\| \bx_{t} \|_{\infty} \leq \mu_{t} \gamma \leq 1$, we have
\bean
\bbP_{n} \Big(
  \widehat Z_{t-s,i}  > \sqrt{L}  \Big)
\leq \bbP \Big(  e^{t-s} \sigma_{t-s} Z > \sqrt{L} - 2 e^{t-s}  \sqrt{ \log (1/\epsilon_n) }  \int_{0}^{t-s} \sigma_{t-u}^{-1} e^{-u} \d u  - e^{t-s}  \Big), \quad \forall i \in [D],
\eean
where $\widehat \bZ_{t-s} = ( \widehat Z_{t-s,1},\ldots, \widehat Z_{t-s,D})$.
A simple calculation yields that
\be
\int_{0}^{t-s} \sigma_{t-u}^{-1} e^{-u} \d u
\leq \int_{0}^{t-s} \sigma_{t-u}^{-1} \d u
= \int_{s}^{t} \sigma_{u}^{-1} \d u
\leq \int_{s}^{t} u^{-1/2} \d u = 2\sqrt{t} - 2\sqrt{s}  \leq 2\sqrt{t},
\label{eqthm2:ito_int_bd}
\ee
where the second inequality holds because $\sigma_{u}^{-1} = 1 / \sqrt{1-e^{-2u}}$ and $1-e^{-2u} \geq u$ for $u \in [0,1/2]$; recall that $t \leq 2
^{-1} \log 2 \leq 1/2$.
Combining the last two displays, we have
\bean
\bbP_{n} \Big(
  \widehat Z_{t-s,i}  > \sqrt{L}  \Big)
  \leq  \bbP \Big(  Z >  \frac{\sqrt{L} - 4  e^{t-s} \sqrt{t \log (1/\epsilon_n) } - e^{t-s} }{e^{t-s} \sigma_{t-s}}  \Big), \quad \forall  i \in [D],
\eean
and hence, by symmetry of the normal distribution,
\bean
\bbP_{n} \Big(
  \big\vert  \widehat Z_{t-s,i}  \big\vert  > \sqrt{L}  \Big)
&& \leq 
  \bbP_{n} \Big( \widehat Z_{t-s,i}  > \sqrt{L}  \Big)
  + \bbP_{n} \Big(
    \widehat Z_{t-s,i} < - \sqrt{L}  \Big)
  \\
&& \leq 2 \bbP \Big(  Z >  \frac{\sqrt{L} - 4  e^{t-s} \sqrt{t \log (1/\epsilon_n) } - e^{t-s} }{e^{t-s} \sigma_{t-s}}  \Big), \quad \forall  i \in [D].
\eean
By the definition of $L$ (see (\ref{eqthm2:L_def})),
\bean
\sqrt{L}
\geq 4 e^{t-s} \sqrt{t \log (1/\epsilon_n) } + e^{t-s} + 2 e^{t-s}\sigma_{t-s} \sqrt{ \log (1/\epsilon_n)},
\eean
which implies that
\bean
\bbP \Big(  Z >  \frac{\sqrt{L} - 4 e^{t-s} \sqrt{t \log (1/\epsilon_n) } - e^{t-s} }{e^{t-s} \sigma_{t-s}}  \Big)
\leq \bbP \Big( Z > 2 \sqrt{\log (1/\epsilon_n) } \Big).
\eean
Combining with the tail probability of the normal distribution, we have
\bean
\bbP_{n} \Big(
  \big\vert  \widehat Z_{t-s,i}  \big\vert  > \sqrt{L}  \Big)
  \leq 2  \bbP \Big( Z > 2 \sqrt{\log (1/\epsilon_n) } \Big)
  \leq 2 \epsilon_n^2,
  \quad \forall i \in [D].
\eean
Moreover, since $L \geq 1$, we have $L \geq \sqrt{L}$.
Hence,
\bean
\bbP_{n} \Big(
  \big\vert  \widehat Z_{t-s,i}  \big\vert  > L  \Big)
  \leq
\bbP_{n} \Big(
  \big\vert  \widehat Z_{t-s,i}  \big\vert  > \sqrt{L}  \Big)
  \leq 2 \epsilon_n^2,
  \quad \forall i \in [D].
\eean
Recall that $\widehat \bZ_{t-s}$ has density $\widehat p_{s \mid t}(\cdot \mid \bx_{t})$, which is the conditional density of $\widehat \bX_{s}$ given $\widehat \bX_{t} = \bx_{t}$.
It follows that, for any $i \in [D]$,
\be
\bbP_{n} \Big( \big\vert  \widehat X_{s,i}  \big\vert > \sqrt{L} \mid \widehat \bX_{t} = \bx_{t} 
\Big) \leq 2 \epsilon_n^2
\quad \text{and} \quad
\bbP_{n} \Big( \big\vert  \widehat X_{s,i}  \big\vert > L \mid \widehat \bX_{t} = \bx_{t} \Big) \leq 2 \epsilon_n^2,
\quad \forall i \in [D].
\label{eqthm2:tail_hat_bd}
\ee

Combining (\ref{eqthm2:ito}) and (\ref{eqthm2:ito_int_bd}), we have
\bean
\big\| \widehat \bZ_{t-s} \big\|_{\infty}
&& \leq  e^{t-s}\| \bx_{t} \|_{\infty}
+ 2 e^{t-s} \sqrt{\log (1/\epsilon_n) }  \int_{0}^{t-s} \sigma_{t-u}^{-1} e^{-u} \d u 
+ e^{t-s} \Big\| \int_{0}^{t-s} \sqrt{2} e^{-u} \d \bB_{u} \Big\|_{\infty}
\\
&& \leq e^{t-s} + 4 e^{t-s}  \sqrt{ t \log ( 1/\epsilon_n) } + e^{t-s} \Big\| \int_{0}^{t-s} \sqrt{2} e^{-u} \d \bB_{u} \Big\|_{\infty}.
\eean
For any $k \geq 1$ and $m \in \bbN$, note that $\vert \sum_{i=1}^{m} a_i \vert^{k} \leq m^{k-1} \sum_{i=1}^{m} \vert a_i \vert^{k}, a_{1},\ldots,a_{m} \in \bbR$.
Since the random vector $\int_{0}^{t-s} \sqrt{2} e^{-u} \d \bB_{u}$ follows the Gaussian distribution $\cN ( \mathbf{0}_{D}, \sigma_{t-s}^2 \bI_{D} )$, it follows that
\bean
\bbE_{n} \Big[ \big\| \widehat \bZ_{t-s} \big\|_{\infty}^{k} \Big]
\leq 3^{k-1} e^{k(t-s)} \Big\{ 1 + 4^{k} \big[ t \log (1/\epsilon_n)  \big]^{k/2}  + \sigma_{t-s}^{k} \bbE \Big[ \big\|  \bZ \big\|_{\infty}^{k} \Big] \Big\},
\quad \forall k \geq 1,
\eean
where $\bZ = (Z_1,\ldots,Z_{D})$ denotes the $D$-dimensional standard normal random vector.
Moreover,
\bean
\bbE \Big[ \big\| \bZ \big\|_{\infty}^{k} \Big]
\leq  \bbE \Big[ \big\| \bZ \big\|_{2}^{k} \Big]
\leq \sqrt{ \bbE \Big[ \big\| \bZ \big\|_{2}^{2k} \Big]} 
\leq \left\{  D^{k-1}  \sum_{i=1}^{D} \bbE \Big[ Z_i^{2k} \Big] \right\}^{1/2}
= D^{k/2} \sqrt{ \bbE \Big[ Z^{2k} \Big] },
\forall k \geq 1,
\eean
where the second inequality follows from the Cauchy--Schwarz inequality.
Note also that $e^{t} \leq \sqrt{2}$ because $t \leq {\widetilde T} \leq 2^{-1} \log 2$.
Combining this with the last two displays, for each $k \geq 1$, there exists a positive constant $\widetilde D_{k} = \widetilde D_{k}(k, D) $ such that
\be
\bbE_{n} \Big[ \big\| \widehat \bZ_{t-s} \big\|_{\infty}^{k} \Big]
\leq \widetilde D_{k} \{ \log ( 1/\epsilon_n) \}^{k/2}
\label{eqthm2:mom_hat_k_bd}.
\ee
Combining the last display with (\ref{eqthm2:tail_hat_bd}), and noting that $\bbE_{n} [ \| \widehat \bZ_{t-s} \|_{\infty}^{k} ] = \bbE_{n} [ \| \widehat \bX_{s} \|_{\infty}^{k} \mid \widehat \bX_{t} = \bx_{t} ], k \geq 1$, we have
\be
\begin{split}
 & \sqrt{ \bbE_{n} \big[ \widehat X_{s,i}^2 \widehat X_{s,j}^2 \mid \widehat \bX_{t} = \bx_{t} \big]  \Big\{  \bbP_{n} \Big(
\big\vert  \widehat  X_{s,i}  \big\vert > \sqrt{L} \mid  \widehat \bX_{t} = \bx_{t} 
\Big)
+ \bbP_{n} \Big(
\big\vert  \widehat  X_{s,j}  \big\vert > \sqrt{L} \mid  \widehat \bX_{t} = \bx_{t}
\Big) \Big\} } 
\\
&  \quad \leq 2 \sqrt{\widetilde D_{4}}  \epsilon_{n}   \log (1/\epsilon_n)  , \quad \forall i, j \in [D], 
\qquad \text{and}
\\
& \sqrt{ \bbE_{n} \big[ \widehat X_{s,i}^2  \mid  \widehat \bX_{t} = \bx_{t} \big]  \bbP_{n} \Big(
\big\vert   \widehat X_{s,i}  \big\vert > L \mid  \widehat \bX_{t} = \bx_{t} 
\Big) } 
\leq \epsilon_n \sqrt{ 2 \widetilde D_{2} \log (1/\epsilon_n) } ,
\quad \forall i \in [D].
\label{eqthm2:mom_hat_bd}
\end{split}
\ee
In particular, the above bounds hold uniformly over $\bx_t\in S$.

Combining (\ref{eqthm2:tot_sec_mom_bd}),  (\ref{eqthm2:mom_true_bd}) and (\ref{eqthm2:mom_hat_bd}), we have
\bean
&& \Big\vert \bbE \big[  X_{s,i}  X_{s ,j} \mid  \bX_{t} = \bx_{t} \big] - \bbE_n \big[ \widehat X_{s,i} \widehat X_{s ,j} \mid \widehat \bX_{t} = \bx_{t} \big]  \Big\vert
\\
&& \leq 2 L d_{\rm TV} \left( p_{s \mid t}(\cdot \mid \bx_{t} ), \widehat p_{s \mid t}(\cdot \mid \bx_{t} ) \right)
+ 10 \epsilon_n 
+  2 \sqrt{\widetilde D_{4}}  \epsilon_{n}   \log (1/\epsilon_n),
\quad \forall i \neq j \in [D].
\eean
Similarly, combining (\ref{eqthm2:tot_fir_mom_bd}), (\ref{eqthm2:mom_true_bd}) and (\ref{eqthm2:mom_hat_bd}), we have
\bean
&& \Big\vert \bbE \big[  X_{s,i}  \mid  \bX_{t} = \bx_{t} \big] - \bbE_n \big[ \widehat X_{s,i} \mid \widehat \bX_{t} = \bx_{t} \big]  \Big\vert
\\
&& \leq 2 L d_{\rm TV} \left( p_{s \mid t}(\cdot \mid \bx_{t} ), \widehat p_{s \mid t}(\cdot \mid \bx_{t} ) \right)
+ \sqrt{10} \epsilon_n 
+ \epsilon_n \sqrt{ 2 \widetilde D_{2} \log (1/\epsilon_n) },
\quad \forall i \in [D],
\eean
Combining the last two displays with (\ref{eqthm2:tv_bd}), we have
\be
\begin{split}
& \bbE \left[ \int_{S} (2 \mu_{t} \gamma)^{-D} \Big\vert \bbE \big[  X_{s,i}  X_{s ,j} \mid  \bX_{t} = \bx \big] - \bbE_n \big[ \widehat X_{s,i} \widehat X_{s ,j} \mid \widehat \bX_{t} = \bx \big]  \Big\vert  \d \bx \right]
\\
& \leq \Big\{ 2 L D_{4}  + 10 + 2 \sqrt{\widetilde D_{4}} \log (1/\epsilon_n) \Big\} \epsilon_n, 
\quad \forall i \neq j \in [D],
\end{split}
\label{eqthm2:sec_hat_bd}
\ee
and
\be
\begin{split}
& \bbE \left[ \int_{S} (2 \mu_{t} \gamma)^{-D} \Big\vert \bbE \big[  X_{s,i} \mid  \bX_{t} = \bx \big] - \bbE_n \big[ \widehat X_{s,i} \mid \widehat \bX_{t} = \bx \big]  \Big\vert \d \bx \right] 
\\
& \leq \Big\{ 2 L D_{4} + \sqrt{10} + \sqrt{ 2 \widetilde D_{2} \log (1/\epsilon_n) } \Big\} \epsilon_n,
\quad \forall i \in [D].
\end{split}
\label{eqthm2:fir_hat_bd}
\ee

Recall that our goal is to derive an upper bound for (\ref{eqthm2:cov_hat_bd_toy}) by combining (\ref{eqthm2:cov_toy_bd1}) and (\ref{eqthm2:cov_toy_bd2}).
The last two displays provide bounds on the differences in the conditional first and second moments.
It remains to bound the conditional means in absolute value to control the difference between the conditional covariances.

By (\ref{eqthm2:conditional_mean}), we have
\bean
\Big\vert \bbE \big[ X_{s,i} \mid \bX_{t} = \bx_{t}, \bX_{0} = \bx_{0} \big] \Big\vert
\leq  \frac{ \sigma_{s}^2 \mu_{t-s} \| \bx_{t} \|_{\infty} + \sigma_{t-s}^2 \mu_{s} \| \bx_{0} \|_{\infty} }{\sigma_{t}^2}
\leq \frac{ \sigma_{s}^2 \mu_{t-s}  + \sigma_{t-s}^2 \mu_{s} }{\sigma_{t}^2}
\leq 2
, \quad \forall i \in [D].
\eean
By the law of total expectation, it follows that
\bean
\Big\vert \bbE \big[ X_{s,i} \mid \bX_{t} = \bx_{t} \big] \Big\vert
\leq
\bbE \left[ \Big\vert \bbE \big[ X_{s,i} \mid \bX_{t}, \bX_{0} \big] \Big\vert \mid \bX_{t} = \bx_{t} \right]
\leq 2,
\quad \forall i \in [D].
\eean
Moreover, by (\ref{eqthm2:mom_hat_k_bd}), we have
\bean
\Big\vert \bbE_{n} \big[ \widehat X_{s,i} \mid \widehat \bX_{t} = \bx_{t} \big] \Big\vert
\leq
\bbE_{n} \big[ \big\| \widehat \bX_{s} \big\|_{\infty} \mid \widehat \bX_{t} = \bx_{t} \big]
\leq \widetilde D_{1} \sqrt{ \log (1/\epsilon_n) } , \quad \forall i \in [D].
\eean
The last two displays hold uniformly over $\bx_{t} \in S$.
Combining these bounds with (\ref{eqthm2:cov_toy_bd1}) and (\ref{eqthm2:cov_toy_bd2}), for all $i \neq j \in [D]$, we have
\bean
&& \bbE \left[
\int_{S}   ( 2\mu_{t} \gamma)^{-D} \Big\vert \text{Cov} \big[ X_{s,i}, X_{s,j} \mid \bX_{t} = \bx \big]  -
\text{Cov}_n \big[ \widehat X_{s,i}, \widehat X_{s,j} \mid  \widehat \bX_{t} = \bx \big] \Big\vert
 \d \bx \right]
 \\
 && \leq \bbE \left[
\int_{S}   ( 2\mu_{t} \gamma)^{-D} \Big\vert \bbE \big[ X_{s,i} X_{s,j} \mid \bX_{t} = \bx \big]  -
\bbE_n \big[ \widehat X_{s,i} \widehat X_{s,j} \mid  \widehat \bX_{t} = \bx \big] \Big\vert \d \bx \right]
\\
&& \quad + 2 \bbE \left[
\int_{S}   ( 2\mu_{t} \gamma)^{-D} \Big\vert \bbE \big[ X_{s,j}\mid \bX_{t} = \bx \big]  -
\bbE_n \big[ \widehat X_{s,j}  \mid  \widehat \bX_{t} = \bx \big] \Big\vert \d \bx \right]
\\
&& \quad + \widetilde D_{1} \sqrt{ \log (1/\epsilon_n) }
\bbE \left[
\int_{S}   ( 2\mu_{t} \gamma)^{-D} \Big\vert \bbE \big[ X_{s,i}\mid \bX_{t} = \bx \big]  -
\bbE_n \big[ \widehat X_{s,i}  \mid  \widehat \bX_{t} = \bx \big] \Big\vert \d \bx \right]
\\
&& \leq \Big\{ 2 L D_{4}  + 10 + 2 \sqrt{\widetilde D_{4}} \log (1/\epsilon_n) \Big\} \epsilon_n
+ \Big\{ 2 + \widetilde D_{1} \sqrt{ \log (1/\epsilon_n) } \Big\}
\Big\{ 2 L D_{4} + \sqrt{10} + \sqrt{ 2 \widetilde D_{2} \log (1/\epsilon_n) } \Big\} \epsilon_n,
\eean
where the second inequality follows from (\ref{eqthm2:sec_hat_bd}) and (\ref{eqthm2:fir_hat_bd}).
Combining this bound with the definition of $L$ (see (\ref{eqthm2:L_def})), and noting that $t \leq  2^{-1} \log 2$, the last display is further bounded by
\bean
D_{5} \epsilon_n  \{ \log (1/\epsilon_n) \}^{3/2},
\eean
where $D_{5} = D_{5} ( D_{4}, \widetilde D_{1}, \widetilde D_{2}, \widetilde D_{4} ) > 0$.
Combining this bound with (\ref{eqthm2:cov_hat_bd_toy}), and noting that $\mu_{t}^2 \leq 1$, we have 
\be
\bbE \left[  \left\vert \bbE \big[  \big\vert \widetilde H_{ij} \big( \bU_{t}, t \big)  \big\vert \big] - \bbE_{n} \big[  \big\vert \widehat H_{ij} \big( \bU_{t}, t \big)  \big\vert \big] \right\vert  \right]
\leq  D_{5} \sigma_{t}^{-4} \epsilon_n  \{ \log (1/\epsilon_n) \}^{3/2}, \quad \forall i \neq j \in [D].
\label{eqthm2:hess_hat_bd}
\ee

Recall that for $i \neq j \in [D]$,
\bean
C_{ij}
\leq  \left\vert  \bbE \big[  \big\vert H_{ij} \big( \bU_{t}, t \big)  \big\vert  \big] -  \bbE \big[   \big\vert \widetilde H_{ij} \big( \bU_{t}, t \big)  \big\vert \big]  \right\vert
+
 \bbE \left[  \left\vert \bbE \big[  \big\vert \widetilde H_{ij} \big( \bU_{t}, t \big)  \big\vert \big] - \bbE_{n} \big[  \big\vert \widehat H_{ij} \big( \bU_{t}, t \big)  \big\vert \big] \right\vert  \right].
\eean
Combining (\ref{eqthm2:hess_trunc_bd}) and (\ref{eqthm2:hess_hat_bd}) with the last display, we have
\bean
C_{ij}
\leq
 \sigma_{t}^{-4} \left[ D_{1} \sigma_{s}^2 \sigma_{t}^{(2 d(i,j) -2) \wedge {\tilde \beta} }  + D_{5} \epsilon_n  \{ \log (1/\epsilon_n) \}^{3/2} \right],
 \quad \forall i \neq j \in [D].
\eean
Note that $\sigma_{u}^2 = 1 - e^{-2u} \leq 2u$ for $u \geq 0$.
Note also that $\sigma_{t}^2 = 1- e^{-2t} \geq t$, since $t \leq {\widetilde T} \leq 2^{-1} \log 2 \leq 1/2$.
Combining these bounds with the last display, we have
\bean
C_{ij} \leq  \frac{ 2 D_{1} s t^{\frac{d(i,j) - 1 }{2} \wedge \frac{{\tilde \beta}}{2}} +  D_{5} \epsilon_n  \{ \log (1/\epsilon_n) \}^{\frac{3}{2}}  }{t^{2}} , 
 \quad \forall i \neq j \in [D].
 \eean
Combining  (\ref{eqthm2:markov}) with the last display, we have
\bean
&& \bbP \left( \widehat G_{t,\tau} \neq G_{0}  \right)
\\
&& \leq
\sum_{ \substack {i,j \in [D]  \\ d(i,j) = 1 } } \frac{ 2 D_{1} s +  D_{5} \epsilon_n  \{ \log (1/\epsilon_n) \}^{\frac{3}{2}}  }{ t^2 \{  C_{\gamma} - C_{2} \sigma_{t}^{{\tilde \beta}} - \tau  \} }
 + \sum_{  \substack {i,j \in [D]  \\ d(i,j) \geq 2 } } \frac{ 2 D_{1} s t^{\frac{{\tilde \beta}}{2}} +  D_{5} \epsilon_n  \{ \log (1/\epsilon_n) \}^{\frac{3}{2}}  }{ t^2 \{ \tau  - C_{2} \sigma_{t}^{{\tilde \beta}} \} }
\\
&& \leq \frac{D_{6}}{t^2} \left[
\frac{s}{ ( \eta_{n,t,1} t^{-\frac{{\tilde \beta}}{2}} ) \wedge \eta_{n,t,2} }
+
\frac{ \epsilon_n  \{ \log (1/\epsilon_n) \}^{\frac{3}{2}} }{\eta_{n,t,1} \wedge \eta_{n,t,2}}
\right],
\eean
where $D_{6} = D_{6} ( D, D_{1}, D_{5} ) > 0$.
Since $s = \underline{T}$, the assertion follows with $C_{3} = D_{6}$.

\end{proof}

\section{Implementation details of the proposed method}
 \label{sec:supp_details}

In this section, we provide practical implementation details of our method, complementing the description in Section~\ref{sec:tuning}.

Recall that the forward process $(\bX_t)_{t \geq 0}$ in Section~\ref{ssec:prelim_diffusion} is the standard OU process~\eqref{eq:OU}.
Our main result (Theorem~\ref{thm2}) extends to more general time-inhomogeneous diffusion processes, including the forward processes underlying the DDPM \citep{ho2020denoising, song2021scorebased} and EDM \citep{karras2022elucidating} frameworks.
In our implementation, we adopt the DDPM framework.
Specifically, the forward process is discretized as follows: for $\overline{T} \in \bbN$ and $t \in \{1, \ldots, \overline{T}\}$,
\bean
\bX_{t} \mid \bX_{0} = \bx_{0} \sim \cN(\mu_{t} \bx_{0}, \sigma_{t}^{2} \bI_{D}),
\eean
where $\mu_{t} = \sqrt{\prod_{l=0}^{t-1} (1 - \alpha_{l})}$, $\sigma_{t} = \sqrt{1 - \mu_{t}^{2}}$, and $\alpha_{l} = \alpha_{\min} + l(\alpha_{\max} - \alpha_{\min}) / (\overline{T} - 1)$ for $l \in \{0, 1, \ldots, \overline{T}-1 \}$.
We set $\alpha_{\min} = 0.001$, $\alpha_{\max} = 0.02$, and $\overline{T} = 500$; the corresponding values of $\mu_{t}$ and $\sigma_{t}$ are displayed in Figure~\ref{fig:noise}.

Within this DDPM framework, the score function is estimated by a suitable class of deep neural networks.
Since the consistency of the graph estimator holds for sufficiently small $t$ by Theorem~\ref{thm2}, we only consider the timesteps $t_{1} < \cdots < t_{M}$ such that $\sigma_{t_{i}} / \mu_{t_{i}} \leq 0.5$ for $t_{i} \in \{2, \ldots, \overline{T}-1\}$, which corresponds to $M = 30$ as shown in Figure~\ref{fig:noise}.
We exclude the timestep $t = 1$ because, in the DDPM implementation, Gaussian noise is typically not added at the last step of the reverse process, which makes the corresponding conditional covariance zero.
Although our graph estimation procedure only requires the score function at $\{t_{1}, \ldots, t_{M}\}$, we train it over the entire range $\{1, \ldots, \overline{T}\}$ as in standard diffusion model training, which enables empirical evaluation of the score function estimator by sampling from the marginal distribution.
The simulation studies in Section~\ref{sec:simulation} and the network analysis in Section~\ref{ssec:network} share the same deep neural network class, while the image data analysis in Section~\ref{ssec:image} uses a different one.
Detailed descriptions are provided in Sections~\ref{ssec:supp_setting}, \ref{ssec:supp_mnist}, and~\ref{ssec:supp_network}, respectively.

\begin{figure*}[!t]
    \centering
    \includegraphics[width=0.6\linewidth]{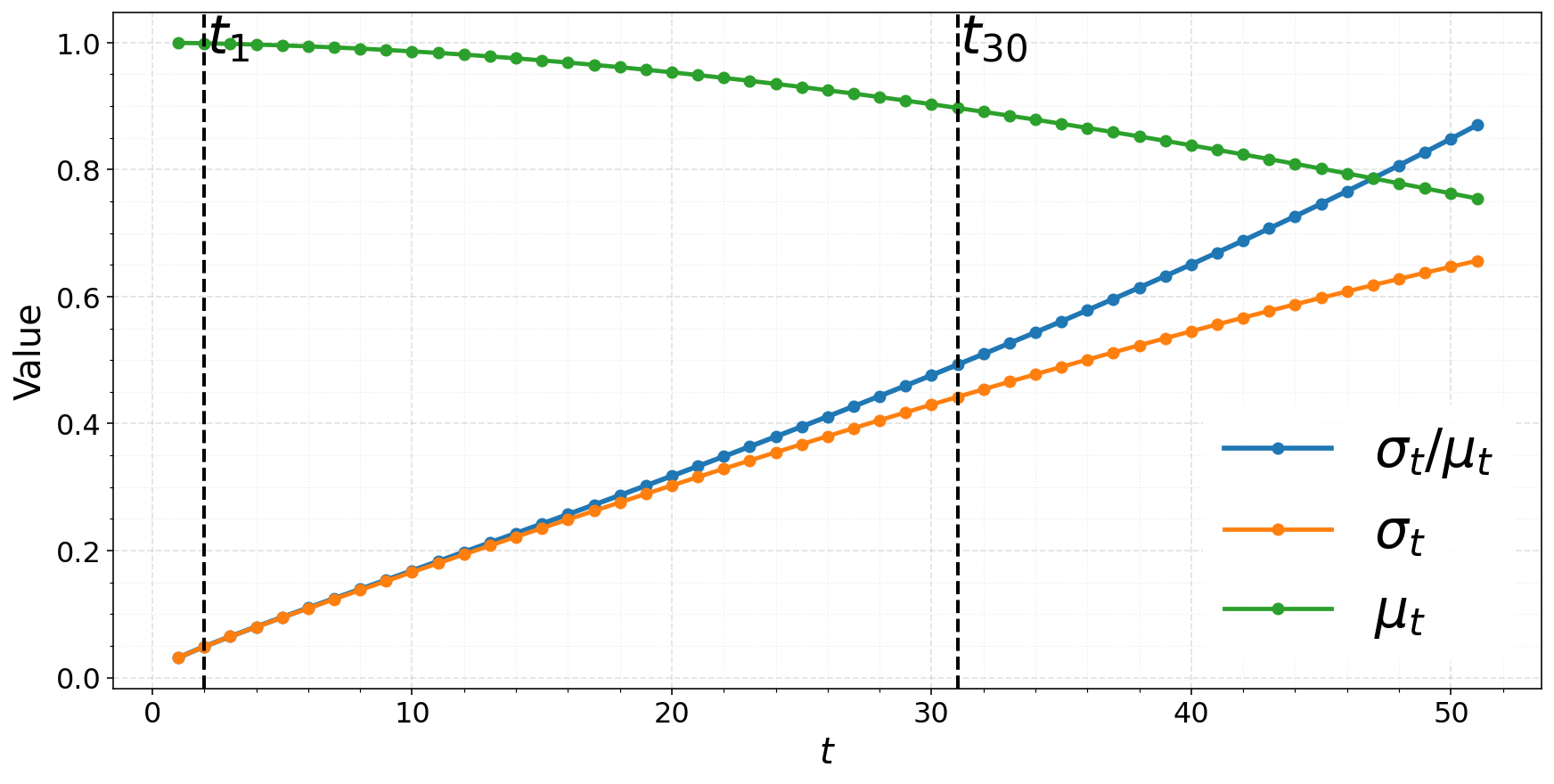}
    \caption{DDPM noise schedule used in our experiments.}
    \label{fig:noise}
\end{figure*}

Once the score function estimator is given, the values $\bbE_n[|\widehat H_{ij}(\bU_{t}, t)|]$ are approximated by Monte Carlo sampling with respect to a suitable random vector $\bU_{t}$.
Specifically, for each $t$, we first draw $N_{1}$ samples from the distribution of $\bU_{t}$; then, conditioning on each such sample $\bx$, we draw $N_{2}$ samples from the conditional distribution of $\widehat \bX_{\underline{T}}$ given $\widehat \bX_{t} = \bx$ and compute the corresponding sample covariance.
In our experiments, we take the distribution of $\bU_{t}$ to be the empirical distribution of the perturbed observations $\mu_{t} \bX^{i} + \sigma_{t} \bZ^{i}$ for $i = 1, \ldots, n$, where $\bZ^{1}, \ldots, \bZ^{n}$ are i.i.d. copies of $\cN(\mathbf{0}_{D}, \bI_{D})$.
We set $N_{1} = 128$ and $N_{2} = 5000$; this sampling and the computation of the sample covariance can be performed efficiently on a GPU.

As discussed in Section~\ref{sec:tuning}, $\widetilde H_{ij}(t)$ is the standardized value of $\bbE_{n}[|\widehat H_{ij}(\bU_{t}, t)|]$ across all pairs $\{(i, j) : i, j \in [D],\ i < j\}$ for each $t \in \{t_{1}, \ldots, t_{M}\}$.
Figure~\ref{fig:examples} illustrates the effect of this standardization on the simple $3$-dimensional Gaussian example from Section~\ref{ssec:supp_illu}, for sample sizes $n \in \{10, 20, 50, 500\}$.
Figure~\ref{fig:estimated_gaussian} shows that $\bbE_{n}[|\widehat H_{ij}(\bU_{t}, t)|]$ approaches $|H_{ij}(t)|$ as $n$ increases.
For very small $t$, however, the error remains large; in particular, $\bbE_{n}[|\widehat H_{13}(\bU_{t}, t)|]$ does not converge to zero as $t \to 0$, even though $H_{13}(0) = 0$.

By contrast, Figure~\ref{fig:estimated_gaussian_std} shows that $\widetilde H_{ij}(t)$ approaches the corresponding standardized version of $|H_{ij}(t)|$.
The $30$-dimensional vectors $(\widetilde H_{ij}(t_{1}), \ldots, \widetilde H_{ij}(t_{30}))$ cluster naturally into two groups, $\{(1, 3)\}$ and $\{(1, 2), (2, 3)\}$, even for moderate sample sizes ($n \geq 20$).
This observation motivates the clustering-based procedure applied to the standardized $30$-dimensional vectors.

Exploiting multiple values of $t$ jointly, rather than committing to a single $t$, yields a robust graph estimator.
As shown in Figure~\ref{fig:estimated_gaussian}, the value $\bbE_{n}[|\widehat H_{ij}(\bU_{t}, t)|]$ is unreliable for very small $t$, and true and false edges may be indistinguishable in that regime.
Nevertheless, clustering on the $30$-dimensional vectors recovers the correct graph, because the unreliable values at small $t$ are outweighed by the informative values at larger $t$.
More generally, the value of $t$ at which edges and non-edges separate varies with the underlying distribution; for some distributions, separation occurs at small $t$ (see Figure~\ref{fig:simul_examples} in Section~\ref{sec:simulation}).
Varying $t$ therefore provides a form of implicit adaptation to the unknown distribution, producing a stable estimator across settings.

\begin{figure*}[!t]
    \centering
    \subfigure[Before standardization]{
        \includegraphics[width=0.45\linewidth]{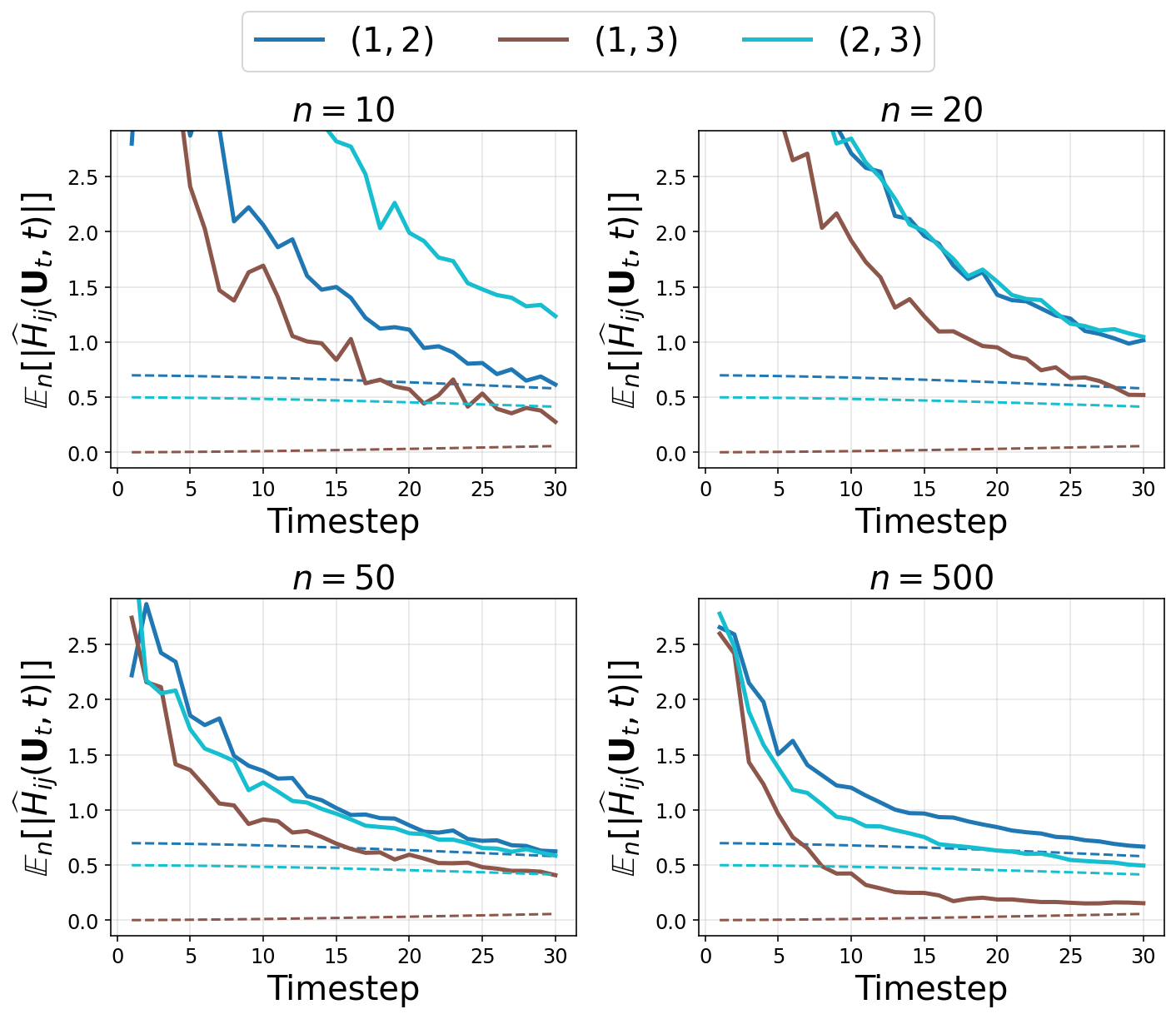}
        \label{fig:estimated_gaussian}
    }
    \subfigure[After standardization]{
        \includegraphics[width=0.45\linewidth]{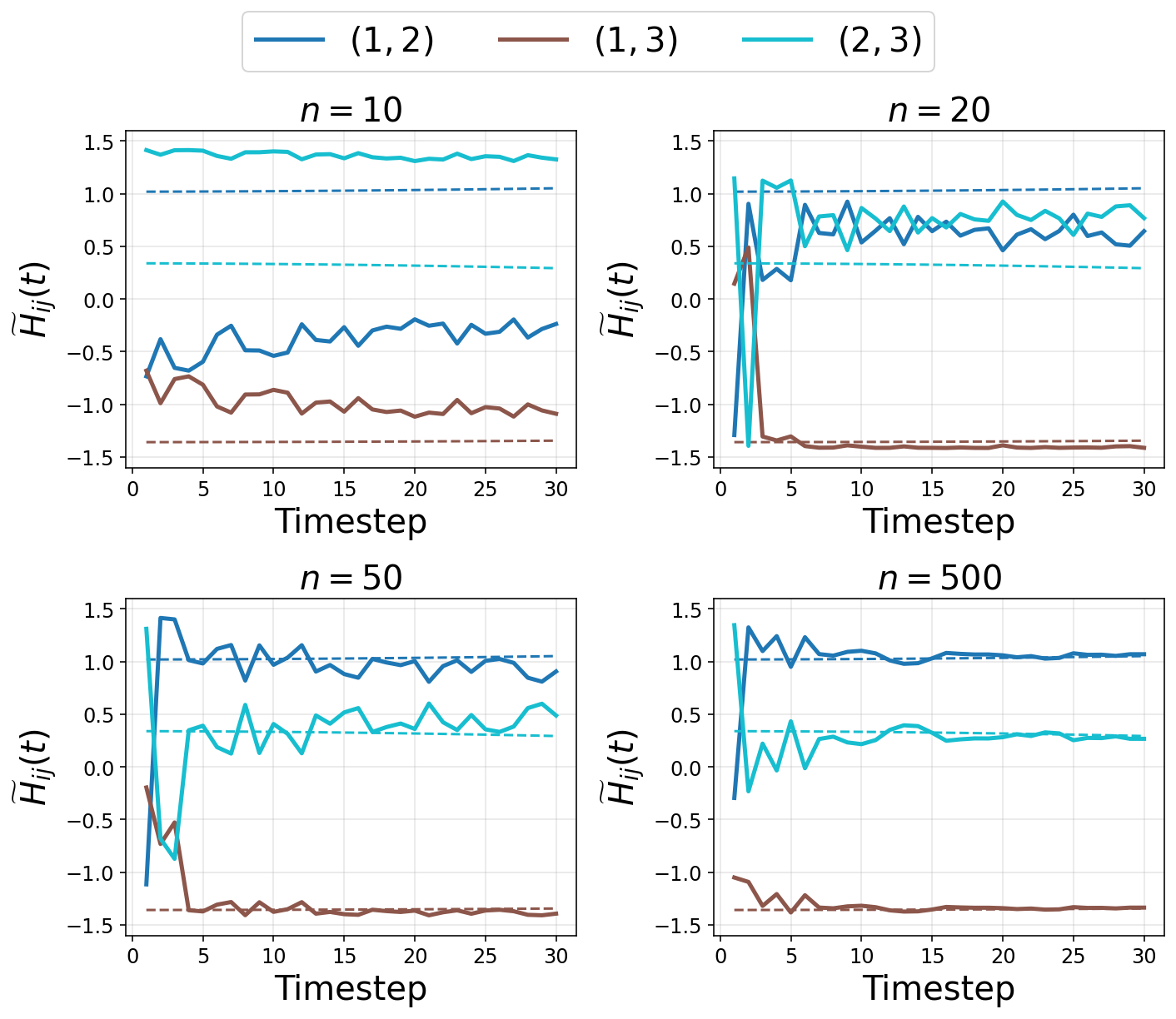}
        \label{fig:estimated_gaussian_std}
    }
    \caption{Estimation results for the $3$-dimensional Gaussian distribution with varying $n$. In the left panel, the solid lines represent the estimated values $\bbE_{n}[\vert \widehat H_{ij}(\bU_{t}, t) \vert]$ and the dashed lines represent the true values $|H_{ij}(t)|$, evaluated at $t_1, \ldots, t_{30}$. In the right panel, the same quantities are shown after standardization, denoted by $\widetilde H_{ij}(t)$ for the estimated values.}
    \label{fig:examples}
\end{figure*}

\section{Details of the main simulations}
\label{sec:supp_simul}

\subsection{Simulation settings}
\label{ssec:supp_setting}

In this subsection, we denote the true random vector by $\bX_{0} = (X_1, \ldots, X_{D})$.
We consider four data distributions, comprising two non-Gaussian and two Gaussian distributions.
We have also conducted analogous experiments under varying data dimensions $D$ and over a wider range of distributions; the results are qualitatively similar and are not reported here.
We therefore fix $D = 20$ throughout.
The corresponding conditional independence graph is, in each case, one of the two graphs $G_{\rm pair}$ and $G_{\rm chain}$ illustrated in Figure~\ref{fig:simul_graphs}.

\begin{figure}[!t]
\centering
\begin{minipage}[t]{0.46\textwidth}
\centering
\vspace{3pt}
$G_{\rm pair}$
\\
\begin{tikzpicture}[
    every node/.style={
        circle, draw, inner sep=0pt,
        minimum size=18pt,
        font=\footnotesize
    }
]
    \node (x1)  at (0.0,0.5)  {$1$};
    \node (x2)  at (0.0,-0.5) {$2$};
    \node (x3)  at (1.3,0.5)  {$3$};
    \node (x4)  at (1.3,-0.5) {$4$};
    \node (d1)  at (2.6,0.5)  [draw=none] {$\cdots$};
    \node (d2)  at (2.6,-0.5) [draw=none] {$\cdots$};
    \node (x19) at (3.9,0.5)  {$19$};
    \node (x20) at (3.9,-0.5) {$20$};

    \draw (x1) -- (x2);
    \draw (x3) -- (x4);
    \draw (x19) -- (x20);
\end{tikzpicture}
\vspace{3pt}
\end{minipage}
\hfill
\begin{minipage}[t]{0.46\textwidth}
\centering
\vspace{3pt}
$G_{\rm chain}$
\\
\vspace{12pt}
\begin{tikzpicture}[
    every node/.style={
        circle, draw, inner sep=0pt,
        minimum size=18pt,
        font=\footnotesize
    }
]
    \node (x1)  at (0.0,0)  {$1$};
    \node (x2)  at (1.2,0)  {$2$};
    \node (d)   at (2.4,0)  [draw=none] {$\cdots$};
    \node (x19) at (3.6,0)  {$19$};
    \node (x20) at (4.8,0)  {$20$};

    \draw (x1) -- (x2);
    \draw (x2) -- (d);
    \draw (d) -- (x19);
    \draw (x19) -- (x20);
\end{tikzpicture}
\vspace{3pt}
\end{minipage}
\caption{Illustration of the graph $G_{\rm pair}$ (left) and $G_{\rm chain}$ (right).}
\label{fig:simul_graphs}
\end{figure}
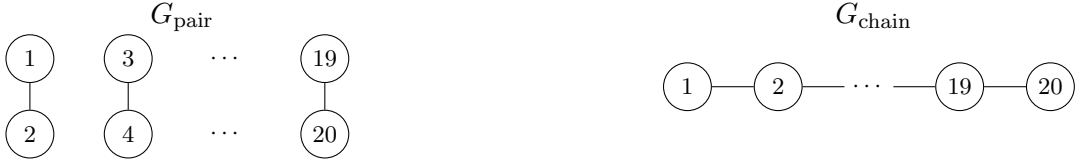

\noindent\textbf{Non-Gaussian: Butterfly distribution $(G_{\rm pair})$}.
We consider the Butterfly distribution \citep{baptista2024learning, zheng2023generalized, liaw2025learning}, in which the odd-indexed coordinates $X_{2i-1}$, $i \in [D/2]$, are mutually independent $\cN(0,1)$ random variables, and the even-indexed coordinates are defined by $X_{2i} = \epsilon_{i} X_{2i-1}$, where $\epsilon_{i} \sim \cN(0,1)$ is independent of $X_{2i-1}$.
The Butterfly distribution is non-Gaussian, and its conditional independence graph is $G_{\rm pair}$.

\noindent\textbf{Non-Gaussian: Gaussian copula $(G_{\rm chain})$.}\quad
We consider a Gaussian copula with an AR$(1)$ correlation structure of parameter $0.8$.
To obtain a multi-modal, non-Gaussian distribution, we use two-component Beta mixtures as marginals.
Specifically, for each coordinate $X_i$, we independently draw four parameters $\widetilde \alpha_{1}, \widetilde \beta_{1}, \widetilde \alpha_{2}, \widetilde \beta_{2}$ from a chi-squared distribution with two degrees of freedom, and define the marginal of $X_{i}$ as the equal-weight mixture of $\mathrm{Beta}(\widetilde \alpha_1 + 1.5, \widetilde \beta_1 + 1.5)$ and $\mathrm{Beta}(\widetilde \alpha_2 + 1.5, \widetilde \beta_2 + 1.5)$.
The additive shift of $1.5$ is included to prevent the marginal density from diverging at the boundary of $[0,1]$.
This distribution is non-Gaussian, and its conditional independence graph is $G_{\rm chain}$.

\noindent\textbf{Gaussian: high correlation and low correlation $(G_{\rm pair})$.}\quad
We consider Gaussian distributions to assess whether nonparametric methods can match the performance of parametric methods designed specifically for the Gaussian setting.
We adopt the conditional independence graph with $G_{\rm pair}$ and consider two correlation regimes, one with high correlation and one with low correlation.
More precisely, for $\rho > 0$, we take $\bX_0 \sim \cN(\mathbf{0}_D, \Sigma_{0})$, where
\bean
\Sigma_{0} = \begin{pmatrix}
A & 0 & \cdots & 0 \\
0 & A & \cdots & 0 \\
\vdots & \vdots & \ddots & \vdots \\
0 & 0 & \cdots & A
\end{pmatrix} \in \bbR^{D \times D},
\qquad
A = \begin{pmatrix} 1 & \rho \\ \rho & 1 \end{pmatrix} \in \bbR^{2 \times 2},
\eean
with $D/2$ blocks $A$ on the diagonal and zeros elsewhere.
We take $\rho = 0.7$ and $\rho = 0.3$, for which the nonzero off-diagonal precision entries equal approximately $-1.37$ and $-0.33$, respectively.

We refer to our method as DDPM throughout this section, since the score function is learned through the DDPM framework discussed in Section~\ref{sec:supp_details}.
Because the score function class considered in the original DDPM is tailored to image generation, we instead take $\cF$ to be a class of fully connected networks equipped with an appropriate time embedding, following Section~7.2 of \citet{kwon2026nonparametric}.
The score function is trained with the Adam optimizer \citep{kingma2014adam}, using a mini-batch size of $100$ and a learning rate of $0.001$, for $1000$ epochs.
Score and Hessian estimation are implemented in \texttt{PyTorch}, and clustering is performed using the \texttt{scikit-learn} package.

\subsection{Other baselines} \label{ssec:supp_other}

We compare DDPM with several existing baselines, which we group into nonparametric methods and parametric (or semi-parametric) methods.
For the nonparametric baselines, we consider SING \citep{baptista2024learning} and L-SING \citep{liaw2025learning}, both briefly reviewed in Section~\ref{ssec:supp_related} and accompanied by publicly available implementation code.
For the parametric and semi-parametric baselines, we consider the graphical lasso (GLASSO) and the nonparanormal (NPN), both also discussed in Section~\ref{ssec:supp_related}.

The performance of SING depends heavily on the polynomial order $p$ of the transport map.
Note that $p = 1$ corresponds to a well-specified model for the Gaussian distributions considered here, since the reference measure is standard normal.
Although $p = 3$ is taken in \cite{baptista2024learning} when estimating the graph underlying the $10$-dimensional Butterfly distribution, SING with $p = 3$ does not scale to the dimension $D = 20$ used in our experiments; over seven hours were required for a single run, even at $n = 300$.
We therefore report only the results for $p = 1$ and adopt the remaining choices, including the regularization parameters and the edge-selection thresholds, from the authors' code; results for various $p$ in lower-dimensional settings are also reported.

For L-SING, we likewise follow the implementation details provided in the accompanying code, in which the transport map between the reference measure and the conditional distribution is parameterized by structured neural networks.
We set the edge-selection threshold to $0.2$, which yields the best performance across the experiments reported in \cite{liaw2025learning}.

For selecting the regularization parameter in GLASSO, we follow the publicly available code of \cite{lyu2024inconsistency}.
Specifically, for each candidate value of the regularization parameter, we first estimate the underlying conditional independence graph and then compute the unpenalized maximum likelihood estimator constrained to this graph.
The estimator with the lowest EBIC is selected, following the procedure of \cite{foygel2010extended}.

For NPN, we first estimate the marginal transformation using the implementation of \cite{liu2009nonparanormal} provided in \cite{zhao2012huge}.
We then apply the same GLASSO procedure described above to the transformed observations.

\subsection{Performance measures}

For each data distribution, we vary the training sample size over $n \in \{100,\allowbreak 200,\allowbreak 300,\allowbreak 400,\allowbreak 500,\allowbreak 1000,\allowbreak 2000,\allowbreak 5000,\allowbreak 10000,\allowbreak 20000,\allowbreak 50000\}$.
At each sample size, every method described above produces an estimated graph $\widehat G = ([D], \widehat E)$, which estimates the true conditional independence graph $G_{0} = ([D], E_{0})$.
Recall that $G_{0}$ is either $G_{\rm pair}$ or $G_{\rm chain}$ in our setting, depending on the data distribution.
To compare the methods, we measure the discrepancy between $\widehat G$ and $G_{0}$ through three metrics, the Hamming distance (HD), the true positive rate (TPR), and the false discovery rate (FDR), defined by
\bean
\mathrm{HD}(G_{0}, \widehat G) = \frac{\vert E_{0} \setminus \widehat E \vert + \vert \widehat E \setminus E_{0} \vert}{2}, \quad
\mathrm{TPR}(G_{0}, \widehat G) = \frac{\vert \widehat E \cap E_{0} \vert}{\vert E_{0} \vert}, \quad
\mathrm{FDR}(G_{0}, \widehat G) = \frac{\vert \widehat E \setminus E_{0} \vert}{\vert \widehat E \vert}.
\eean
The factor of $2$ in HD accounts for the inclusion of each undirected edge as both $(i, j)$ and $(j, i)$ in the edge sets.
When $\widehat E$ is empty, we set $\mathrm{FDR}(G_{0}, \widehat G) = 0$ by convention.
Each experiment is repeated $10$ times, and we report the average of each metric.

\subsection{Performance results}

The overall performance results are summarized in Figure~\ref{fig:simul_d20}.
DDPM and L-SING are the only methods that perfectly recover the graphs across all four data distributions, and DDPM is more sample-efficient than L-SING in every setting.
The two methods exhibit similar TPRs, but DDPM attains a substantially lower FDR across the full range of sample sizes.
This finding aligns with the optimality of diffusion models in learning distributions with sparse graphical structure \citep{kwon2026nonparametric}, which our methodology directly exploits.
We also experimented with various thresholds for L-SING, but DDPM consistently outperformed it across all values considered.

The parametric baselines GLASSO and NPN fail to recover the conditional independence graph of the Butterfly distribution, as also reported in the SING and L-SING analyses \citep{baptista2024learning, liaw2025learning}.
On the Gaussian distributions, where GLASSO and NPN are correctly specified, they recover the true graph with $\mathrm{HD} < 1$ on average for $n \geq 100$ in the high-correlation regime and for $n \geq 300$ in the low-correlation regime.
On the Gaussian copula, NPN is correctly specified and likewise recovers the true graph with $\mathrm{HD} < 1$ on average for $n \geq 100$, while GLASSO is misspecified and fails to recover the graph.

Notably, DDPM is competitive with the corresponding correctly specified models in every setting, suggesting that DDPM can serve as a unified alternative to parametric methods across data distributions.
DDPM achieves $\mathrm{HD} < 1$ on average for $n \geq 100$ in the high-correlation Gaussian setting and for $n \geq 400$ in the low-correlation Gaussian setting.
For the Gaussian copula setting, DDPM attains $\mathrm{HD} < 1$ on average for $n \geq 1000$, which is slightly less sample-efficient than the correctly specified NPN model.
While DDPM tends to produce more false edges than NPN at smaller sample sizes, the residual false edges are not arbitrary.
As shown in Figure~\ref{fig:simul_examples}, the values $\widetilde H_{ij}(t)$ are closely aligned with the graph distance $d(i,j)$, and this alignment is already clearly visible at $n = 100$.
This implies that although DDPM does not perfectly recover the graph $G_{\rm chain}$ at smaller sample sizes, its falsely declared edges are restricted to pairs with $d(i,j) = 2$.
For $n = 1000$, Figure~\ref{fig:simul_examples} shows that the gap between true and false edges widens further, which explains the consistency of DDPM observed for $n \geq 1000$.

Finally, we comment on SING.
Perfect graph recovery on the Butterfly distribution with polynomial order $p = 3$ is reported in \cite{baptista2024learning}, but this polynomial order could not be scaled to $D = 20$ in our experiments.
We also experimented with $p = 2$, but the results were nearly indistinguishable from those reported here for $p = 1$ across all data distributions considered.
At lower dimensions, SING with $p = 3$ perfectly recovers the graph for the Butterfly distribution at large sample sizes, but fails to recover the graph for the Gaussian copula; see Figure~\ref{figure:sing}.
By contrast, DDPM is consistent in both settings and is more sample-efficient than SING for every value of $p$ considered.

\begin{figure*}[!t]
    \centering
    \includegraphics[width=\linewidth]{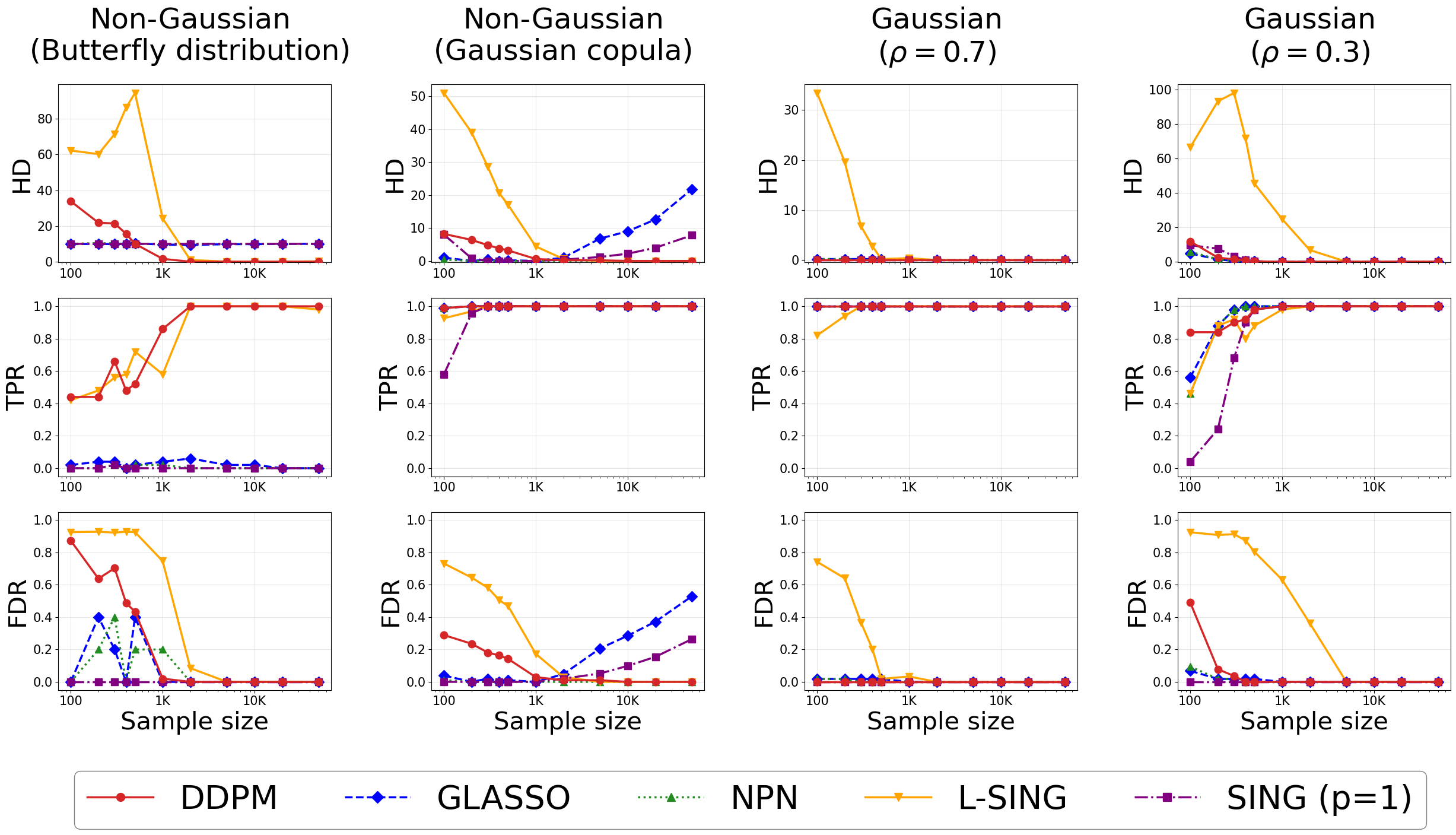}
    \caption{ Graph estimation results under various sample sizes.
    } \label{fig:simul_d20}
\end{figure*}

\begin{figure*}[!t]
    \centering
    \subfigure[Butterfly distribution with $D = 6$.]{
        \includegraphics[width=0.8\linewidth]{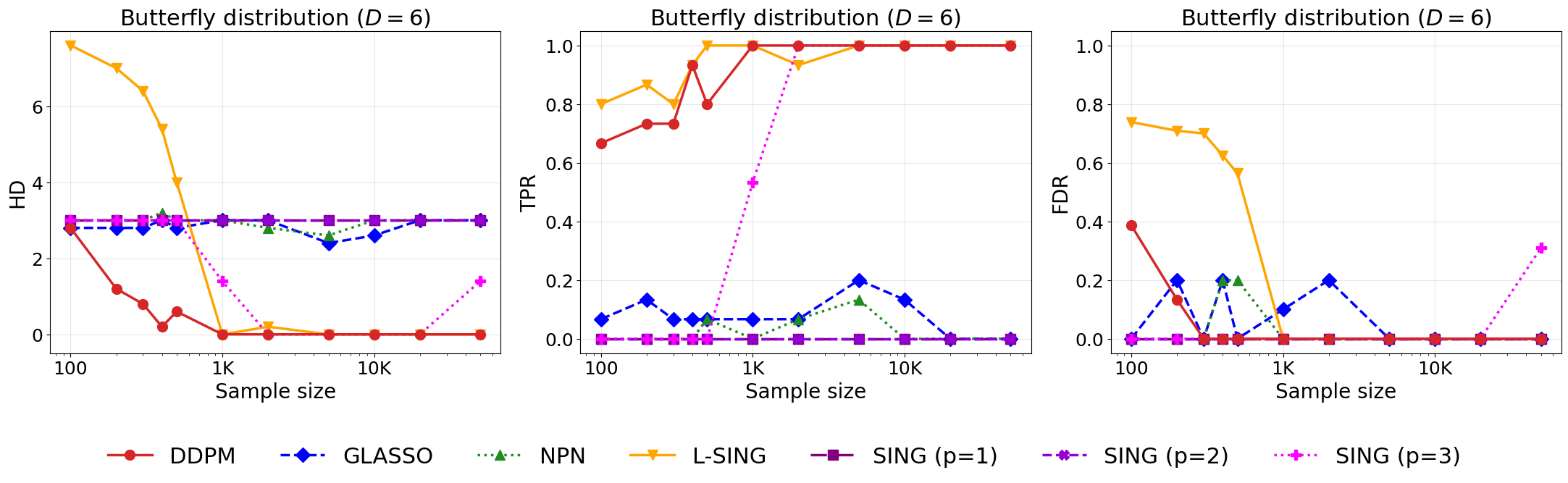}
    }
    \subfigure[Gaussian copula with $D = 5$.]{
        \includegraphics[width=0.8\linewidth]{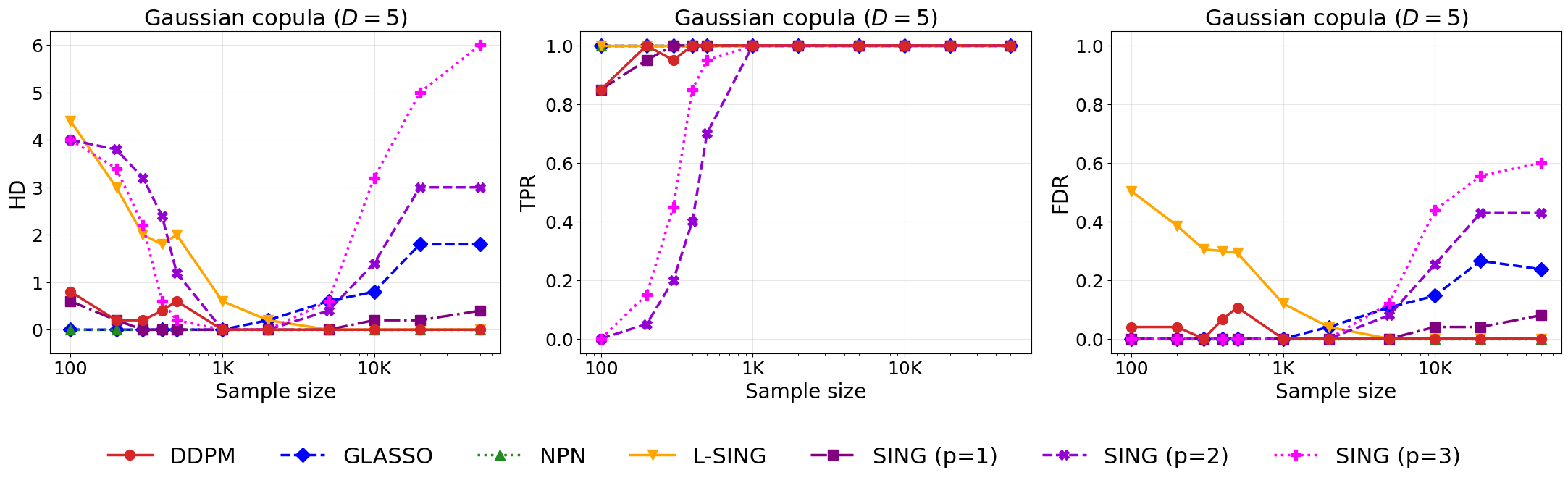}
    }
    \caption{Graph estimation results for non-Gaussian distributions with small $D$.}
    \label{figure:sing}
\end{figure*}

\subsection{Computational cost}
\label{ssec:comp_time}

All experiments are conducted on a single NVIDIA RTX~$4090$ GPU paired with an Intel Xeon Silver $4310$ CPU.
Table~\ref{tab:comp_time} reports the computational cost of each method on the Gaussian copula with $D = 20$ at sample size $n = 5000$, averaged over $10$ replicates.
For DDPM, we report the training and inference times separately; the former corresponds to learning the score function, and the latter to constructing the Hessian estimator via sampling and performing the clustering step.
DDPM and L-SING incur runtimes of comparable order, with the total time roughly $211$ seconds for DDPM and $100$ seconds for L-SING.
The gap between the two is primarily due to our choice of batch size and the number of training epochs when learning the DDPM score function.
We have observed that similar estimation performance can be obtained with larger batch sizes and fewer epochs, which would reduce the training time substantially.
We did not pursue such tuning, since the results reported here are already satisfactory under default hyperparameters.
More generally, both training and inference times scale roughly linearly in the batch size, so further acceleration is straightforward whenever GPU memory permits.
The cost of SING grows sharply with the polynomial order $p$; with $p = 2$, the runtime exceeds $2300$ seconds.
This observation supports our choice of $p = 1$ as the SING benchmark in the main experiments.
By contrast, the parametric and semi-parametric baselines GLASSO and NPN require only a fraction of a second.
\begin{table}[!t]
\caption{Average computation time (in seconds) with sample size $n = 5000$.}
\centering
\begin{tabular}{ccccccc}
\toprule
\multicolumn{2}{c}{DDPM} & \multirow{2}{*}{L-SING} & \multicolumn{2}{c}{SING} & \multirow{2}{*}{GLASSO} & \multirow{2}{*}{NPN} \\
\cmidrule(lr){1-2} \cmidrule(lr){4-5}
Train & Infer & & $p=1$ & $p=2$ & & \\
\midrule
$185.94$ & $25.35$ & $100.31$ & $67.59$ & $2311.10$ & $0.16$ & $0.18$ \\
\bottomrule
\end{tabular}
\label{tab:comp_time}
\end{table}

\section{Details of the real data analyses}
\label{sec:supp_add}

\subsection{Image analysis}
\label{ssec:supp_mnist}

In this subsection, we provide the implementation details of the analysis in Section~\ref{ssec:image}.
For the score function class $\cF$, we adopt the architecture of the original DDPM \citep{ho2020denoising}, which combines a U-Net \citep{ronneberger2015unet} with self-attention \citep{vaswani2017attention}.
The score function is trained with a mini-batch size of $128$, a learning rate of $0.0002$, and $50$ training epochs.
Figure~\ref{figure:mnist_gen} shows randomly selected real MNIST images together with samples generated by our trained diffusion model.
The generated samples are visually close to the real images, indicating that the model is well-trained.

\begin{figure*}[!t]
    \centering
    \includegraphics[width=0.6\linewidth]{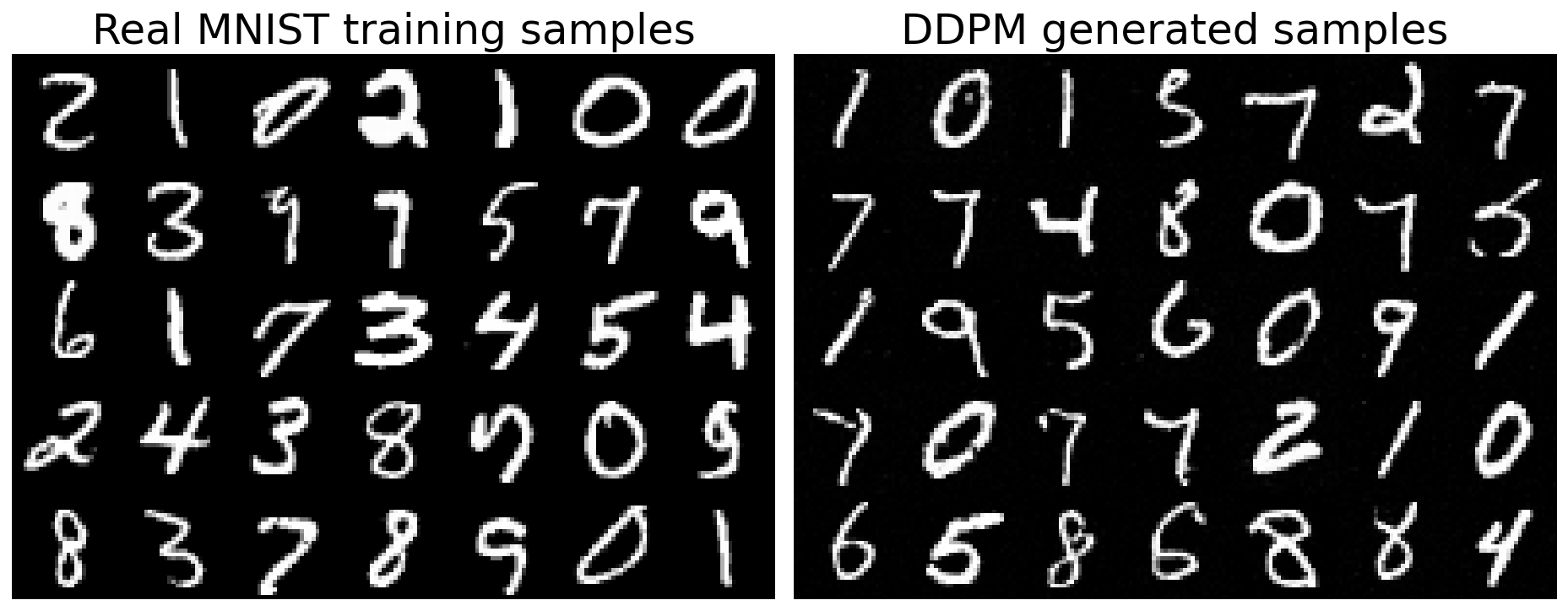}
    \caption{Real MNIST images (left) and samples generated by the trained diffusion model (right).}
    \label{figure:mnist_gen}
\end{figure*}

\subsection{Network analysis}
\label{ssec:supp_network}

In this subsection, we provide the implementation details for the analysis in Section~\ref{ssec:network}.
For the score function class $\cF$, we adopt the architecture described in Section~\ref{sec:supp_details}.
The score function is trained with a mini-batch size of $128$, a learning rate of $0.0005$, and $2000$ training epochs.
We also set $N_{1} = 250$.
Because the competitor and customer/supplier networks differ in the number of edges, we additionally report the Jaccard distance (JD), defined as
\bean
\mathrm{JD}(G_{0}, \widehat G) = 1 - \frac{\vert E_{0} \cap \widehat E \vert}{\vert E_{0} \cup \widehat E \vert}.
\eean
Table~\ref{tab:competitor_mismatch} lists all company pairs for which our estimated graph and the competitor relationship from \textit{Relato} disagree.
The score reported in the table is the average of $\widetilde H_{ij}(t)$ over $t \in \{t_{1}, \ldots, t_{30}\}$; a larger score indicates that the two companies are more likely to be connected in the estimated graph.

\begin{table}[!t]
\caption{Company pairs for which our estimated graph and the competitor relationships from \textit{Relato} disagree, together with their scores. A larger score indicates that the two companies are more likely to be connected in the estimated graph.}
\centering
\small
\begin{tabular}{cllc}
\toprule
\multicolumn{4}{l}{\textbf{Included but non-competitor}} \\
\midrule
Rank & Company A & Company B & Score \\
\midrule
1 & Parker Hannifin & Rockwell Automation & $4.40$ \\
2 & Eaton & Rockwell Automation & $2.98$ \\
3 & FedEx & United Rentals & $2.89$ \\
4 & Ingersoll-Rand & Wabtec & $2.37$ \\
5 & Textron & United Rentals & $2.37$ \\
6 & Parker Hannifin & United Rentals & $2.08$ \\
7 & Dover & IDEX & $1.82$ \\
8 & Paychex & Republic Services & $1.76$ \\
\midrule
\multicolumn{4}{l}{\textbf{Competitor but excluded}} \\
\midrule
Rank & Company A & Company B & Score \\
\midrule
1 & Equifax & Paychex & $1.20$ \\
2 & Boeing & General Dynamics & $0.53$ \\
3 & Boeing & FedEx & $0.37$ \\
4 & Boeing & Lockheed Martin & $0.24$ \\
5 & Boeing & Delta Air Lines & $0.09$ \\
6 & ADP & FedEx & $-0.03$ \\
7 & Boeing & Huntington Ingalls Industries & $-0.04$ \\
8 & Ingersoll-Rand & Johnson Controls & $-0.06$ \\
9 & Dover & Eaton & $-0.44$ \\
10 & Honeywell & Textron & $-0.53$ \\
\bottomrule
\end{tabular}
\label{tab:competitor_mismatch}
\end{table}



\end{document}